\documentclass[a4paper, 10pt]{article}

\usepackage[utf8]{inputenc}    
\usepackage[T1]{fontenc}
\usepackage{lmodern}
\usepackage{graphicx}
\usepackage{textcomp}
\usepackage{amsmath, amsfonts, amssymb, amsthm}
\usepackage{mathtools}
\usepackage{array}
\usepackage{dsfont} 
\usepackage{enumitem} 
\usepackage{geometry} 
\usepackage{braket} 
\usepackage{esint}
\usepackage{mathrsfs} 
\usepackage{caption, subcaption}
\usepackage{tocloft}
\usepackage{float}
\usepackage{xcolor} 
\usepackage{titletoc}
\usepackage{titling}
\usepackage{titlesec}
\usepackage{fullpage}

\usepackage{hyperref}
\hypersetup{pdfauthor={Jean Cazalis}, pdftitle={Dirac cones for a mean-field model of graphene}, colorlinks = true}

\theoremstyle{plain}
\newtheorem{theo}{Theorem}
\newtheorem{prop}[theo]{Proposition}
\newtheorem{cor}[theo]{Corollary}
\newtheorem{lemma}[theo]{Lemma}
\newtheorem{hyp}{Assumption}

\theoremstyle{definition}

\newtheorem{rem}[theo]{Remark}

\let\OLDthebibliography\thebibliography
\renewcommand\thebibliography[1]{
	\OLDthebibliography{#1}
	\setlength{\parskip}{0pt}
	\setlength{\itemsep}{0pt plus 0.3ex}
}

\newcommand{\intervalle}[4]{\mathopen{#1}#2
	\mathclose{}\mathpunct{},#3
	\mathclose{#4}}
\newcommand{\intff}[2]{\intervalle{[}{#1}{#2}{]}}
\newcommand{\intof}[2]{\intervalle{(}{#1}{#2}{]}}
\newcommand{\intfo}[2]{\intervalle{[}{#1}{#2}{)}}
\newcommand{\intoo}[2]{\intervalle{(}{#1}{#2}{)}}

\newcommand{\petito}[1]{o\mathopen{}\left(#1\right)}
\newcommand{\grandO}[1]{O\mathopen{}\left(#1\right)}
\newcommand{\abs}[1]{\left\lvert#1\right\rvert}
\newcommand{\absL}[1]{\left\lvert#1\right\rvert}
\newcommand{\norme}[1]{\left\lVert#1\right\rVert}
\newcommand{\normeL}[1]{\left\lVert#1\right\rVert}
\newcommand{\pdtsc}[2]{\left\langle#1,#2\right\rangle}
\newcommand{\lper}{L_\mathrm{per}}
\newcommand{\lk}{L_\mathbf{k}}
\newcommand{\lscr}{\mathscr{L}}

\newcommand{\enstq}[2]{\left\{#1\mathrel{}\middle|\mathrel{}#2\right\}}
\newcommand{\enstqbis}[2]{\left\{#1\mathrel{~}\middle|\mathrel{~}#2\right\}}
\newcommand{\restreinta}{\mathclose{}|\mathopen{}}
\newcommand{\et}{\quad \text{and} \quad}
\newcommand{\ou}{\quad \text{where} \quad}
\newcommand{\ensemblenombre}[1]{\mathbb{#1}}
\newcommand{\N}{\ensemblenombre{N}}
\newcommand{\Z}{\ensemblenombre{Z}}

\newcommand{\R}{\ensemblenombre{R}}
\newcommand{\C}{\ensemblenombre{C}}

\newcommand{\pt}{\, .}
\newcommand{\trm}{\underline{\tr}}

\newcommand{\limit}[1]{\quad\underset{#1}{\longrightarrow}\quad}
\newcommand{\diff}{\mathop{}\mathopen{}\mathrm{d}}

\DeclareMathOperator{\tr}{Tr}
\DeclareMathOperator{\spn}{span}

\DeclareMathOperator{\id}{id}
\DeclareMathOperator{\sign}{sign}
\DeclareMathOperator{\dist}{d}
\DeclareMathOperator{\supp}{supp}

\renewcommand\thesubsection{\arabic{section}.\arabic{subsection}}
\renewcommand\thesubsubsection{\arabic{section}.\arabic{subsection}.\arabic{subsubsection}}

\titleformat{\subsection}[runin]{\normalfont\bfseries}{\thesubsection}{.5em}{}[.]
\titleformat{\subsubsection}[runin]{\normalfont\bfseries}{\thesubsubsection}{.5em}{}[.]

\title{Dirac cones for a mean-field model of graphene}
\author{Jean Cazalis\thanks{CNRS and CEREMADE, University of Paris-Dauphine, PSL University, 75016 Paris, France; email: \href{mailto:cazalis@ceremade.dauphine.fr}{cazalis@ceremade.dauphine.fr}}}
\date{\today}

\begin{document}

	\maketitle
	
	\begin{abstract}
		In this article, we show that, in the dissociation regime and under a non-degeneracy assumption, the reduced Hartree-Fock theory of graphene presents Dirac points at the vertices of the first Brillouin zone and that the Fermi level is exactly at the coincidence point of the cones. For this purpose, we first consider a general Schrödinger operator $H=-\Delta+V_L$ acting on $L^2(\R^2)$ with a potential $V_L$ which is assumed to be periodic with respect to some lattice with length scale $L$. Under some assumptions which covers periodic reduced Hartree-Fock theory, we show that, in the limit $L\to\infty$, the low-lying spectral bands of $H_L$ are given to leading order by the tight-binding model. For the hexagonal lattice of graphene, the latter presents singularities at the vertices of the Brillouin zone. In addition, the shape of the Bloch bands is so that the Fermi level is exactly on the cones.
	\end{abstract}
	
	\tableofcontents
	
	\section{Introduction}
	
	In this article, we study the spectral properties of a two-dimensional periodic Schrödinger operator $H=-\Delta +V_L$ acting on $L^2(\R^2)$. We assume that the potential $V_L$ comes from a many-sites lattice 
	\begin{align*}
	\lscr_L^\mathbf{R} \coloneqq L\left(\lscr + \mathbf{R}\right) = \enstq{L(\mathbf{u+r})}{\mathbf{u}\in\lscr ,~ \mathbf{r\in R}} \, ,
	\end{align*}
	where $\lscr_L \coloneqq L\lscr\subset\R^2$ is a two-dimensional Bravais lattice with length scale $L>0$ and $\mathbf{R}\subset\R^2$ is a (finite) collection of sublattice shifts. Such a periodic operator exhibits an electronic band structure, described as a Bloch bundle, which gives the range of energies that an electron, moving in the potential $V_L$, may attain~\cite{reed1978methodsIV, kittel2004introduction, kuchment2016overview}. The electronic, optical and magnetic properties of crystals depend on the form of these bands. In particular, the dynamics of a wave packet moving in the structure is strongly influenced by the bands geometry at the vicinity of the initial energy-momentum datum~\cite{ashcrfot1976solid, teufel2003adiabatic, allaire2005homogenization}.
	
	If $\lscr^\mathbf{R}$ is the honeycomb lattice, which is appropriate for describing graphene, then we expect the Bloch bundle to have conical singularities, called \emph{Dirac points}, at the vertices of the first Brillouin zone $\Gamma^*$~\cite{wehling2014dirac}. This terminology comes from the fact that wave packets whose energy-momentum is initially concentrated near these singularities evolve according to a two dimensional Dirac wave equation, the equation for massless relativistic fermions~\cite{fefferman2013wavepackets, arbunich2018rigorous}. The presence of Dirac points in honeycomb structures was first proved for the tight-binding model of graphene by Wallace~\cite{wallace1947bandtheory} and is now established for more realistic ones, including continuous models~\cite{hainzl2012groundstate, fefferman2012honeycomb, lee2016diraccones, fefferman2017honeycomb, berkolaiko2018symmetry, leeThorp2018elliptic}. Unfortunately, these works only consider one electron moving in a periodic material. But real systems have infinitely many electrons interacting with each other. The main motivation for this article is to show that Dirac points also appear when interactions between electrons are taken into account through a nonlinear term in the potential $V_L$.
	
	The simplest model for interacting electrons that one can think of is the \emph{periodic reduced Hartree-Fock (rHF) model}~\cite{solovej1991proof}. Hartree-Fock theories are standard approximation methods for atomic models where the electronic wave function is assumed to have the form of a Slater determinant~\cite{hartree1928wave, lieb1977hartreefock, lions1987solutions}. In the thermodynamic limit, these models converge to periodic nonlinear models~\cite{catto2001thermodynamic, catto2002thermodynamic}. The solution satisfies a nonlinear equation, called the mean-field equation~\cite{cances2008newapproach, ghimenti2008properties}. When the charges interact through the three-dimensional Coulomb potential $\frac{1}{\abs{\mathbf{x}}}$, the mean-field potential $V^\mathrm{MF}_L$ in the periodic rHF theory for the crystal whose nuclei are located at the vertices of $\lscr^\mathbf{R}_L$ is solution of
	\begin{align}
	\label{eq: mean-field_equation_intro}
	V^\mathrm{MF}_L = \left[\mathds{1}_{\intof{-\infty}{\epsilon_L}}(-\Delta+V^\mathrm{MF}_L)(\mathbf{x,x}) - \sum_{\mathbf{r\in}\lscr_L^\mathbf{R}} \delta_{\mathbf{r}}\right]\ast \frac{1}{|\cdot|} + \sum_{\mathbf{r}\in\lscr_L^\mathbf{R}} V^\mathrm{pp}(\cdot - \mathbf{r}) \pt
	\end{align}
	In this equation, $\delta$ is the Dirac delta, $V^\mathrm{pp}$ a pseudo-potential modeling the core electrons and $\epsilon_L\in\R$ a Lagrange multiplier called the Fermi level, which can be interpreted as a chemical potential and is used to adjust the number of electron per unit cell, so that the periodic measure in the square bracket of~\eqref{eq: mean-field_equation_intro} is locally neutral. This is necessary for the convolution with $\frac{1}{\abs{\mathbf{x}}}$ to make sense. It seems natural to expect that graphene will exhibit Dirac cones in rHF theory. However, this does not immediately follow from the existing results which only deal with the linear case.
	
	The \emph{dissociation regime} $L\to\infty$, corresponds to taking the nuclei of the crystal far from each others. The mean-field potential $V^\mathrm{MF}_L$ then resembles a superposition of mono-atomic potentials:
	\begin{align}
	\label{eq: almost_superposition_intro}
	V^\mathrm{MF}_L \simeq \sum_{\mathbf{r}\in\lscr^\mathbf{R}} V^\mathrm{MF}(\cdot - L\mathbf{r}) \, ,
	\end{align}
	for some potential $V^\mathrm{MF}$ solution of a nonlinear equation for one atom. Schrödinger operators whose potential is given by an exact periodic superposition of potential wells have already been studied in the literature~\cite{simon1984semiclassicalIII, outassourt1984, mohamed1991estimations, daumer1993periodique, fefferman2017honeycomb}. It is known that, in the semiclassical limit or in the dissociation regime, the width of the bands is exponentially small, determined by quantum tunneling, and that the geometry of the low-lying bands is given by the tight-binding model associated with the crystal. The latter model is often used to efficiently compute band structures in solid-state physics~\cite{goringe1997tightbinding}. However, these results do not cover periodic rHF theory.	
	
	In addition, it is also physically very important that the model describes a \emph{Dirac semi-metal}, as graphene~\cite{neta2009electronic, cooper2012experimental}. In other words, we want to show that $\epsilon_L$ is exactly equal to the energy at which the cones touch. Otherwise, the small excitations of the Fermi sea would not behave as Dirac fermions. In this article, we partially solve both questions. We study the dissociation limit $L\to\infty$ where we can prove the expected result under a reasonable assumption. The opposite regime, called the \emph{weak contrast regime } $L\to 0$, is studied in~\cite[Chapter 4]{cazalisthesis} whose results are summarized in Appendix~\ref{sec:the-weak-constrast-regime}. There, it is shown that the expected cones exist and that the Fermi level $\epsilon_L$ does \emph{not} coincide with the cones energy: the model describes a \emph{metal}. The general result is therefore \emph{not} true for all values of $L$. Because our argument is not quantitative, we are unable to determine in which regime lies the finite physical value $L_\mathrm{phy}\simeq 5.36$ of graphene~\cite{cooper2012experimental}. However, a numerical investigation with DFTK software~\cite{herbst2021dftk} suggests that $L_\mathrm{phy}$ is large enough to be in the Dirac semi-metal phase, see the discussion following Corollary~\ref{cor:dirac_points_rHF} below.
	
	For the purpose of studying the dissociation regime, we consider a general potential $V_L$ and we exhibit conditions under which the low-lying bands in the dispersion relation of $-\Delta+V_L$ can be approached by the corresponding tight-binding model. These conditions include periodic rHF theory with three-dimensional Coulomb interactions. Also, when the lattice has honeycomb symmetries and under a non-degeneracy condition, we prove the presence of Dirac points and we show, as expected, that the Fermi level $\epsilon_L$ is equal to the energy level of the cones.
	
	It has been shown that interactions can modify the shape of the Dirac cone in graphene~\cite{giuliani2010anomalous, hainzl2012groundstate}. Here we show that no such modification occurs in the rHF model. This is essentially because the mean-field potential is local. It would be interesting to study the full Hartree-Fock model, where the exchange term has been predicted to renormalize the cone with a logarithmically divergent effective velocity~\cite{hainzl2012groundstate}.
	
	Since our main motivation is the study of crystals sharing the symmetries of graphene, we work in $2D$, although many of our arguments hold the same in arbitrary dimension.
	
	This article is a shortened version of Chapter 3 of the authors's PhD thesis~\cite{cazalisthesis}. There, the reader will find more detailed proofs and additional comments.
	
	\subsection*{Organization of the paper}
	In Section~\ref{sec:statement-of-the-main-results}, we state our main results. Theorem~\ref{theo: feshbach-schur} is about the convergence to the tight-binding model of the periodic Schrödinger operator $H_L=-\Delta+V_L$ where $V_L$ satisfies some assumptions. Theorem~\ref{theo: existence_dirac_cones} states that, under a non-degeneracy condition, the dispersion relation of $H_L$ presents Dirac points  when $\lscr^\mathbf{R}$ is the honeycomb lattice. Theorem~\ref{theo: rHF} states that the assumptions in Theorem~\ref{theo: feshbach-schur} cover the periodic rHF theory with three-dimensional Coulomb interactions plus a pseudo-potential term which must satisfy a ionization condition. In Section~\ref{sec:two-dimensional-lattices-at-dissociation-with-coulomb-singularities}, we show Theorem~\ref{theo: feshbach-schur} whose proof strongly relies on the Feshbach-Schur method. In Section~\ref{sec:example-of-the-honeycomb-lattice}, we show Theorem~\ref{theo: existence_dirac_cones}. Section~\ref{sec:reduced-periodic-hartree-fock-model-at-dissociation} is devoted to the proof of Theorem~\ref{theo: rHF} which uses the concentration-compactness method. In Appendix~\ref{sec:the-weak-constrast-regime}, we consider the weak contrast regime $L\to 0$ and we expose the results from~\cite[Chapter 4]{cazalisthesis}. At last, in Appendix~\ref{sec:perturbation-theory-for-singular-potentials}, we state a perturbation theory result for singular potentials.
	
	\subsection*{Acknowledgments}
	The author would like to thank his PhD advisor M. Lewin for valuable discussions, A. Levitt for his assistance with DFTK software and D. Gontier for helpful advice with numerical implementations. This project has received funding from the European Research Council (ERC) under the European Union’s Horizon 2020 research and innovation programme (grant agreement MDFT No 725528 of M. Lewin).
	
	\section{Statement of the main results}\label{sec:statement-of-the-main-results}
	
	In this section, we state our main results. First, we recall the basic geometric features of two-dimensional many-site lattices. Then, we consider a periodic potential $V_L$ and we add conditions under which the dispersion relation of $-\Delta+V_L$ is given to leading order by the tight-binding model (see Theorem~\ref{theo: feshbach-schur}). In Theorem~\ref{theo: existence_dirac_cones}, we make this statement more precise when $\lscr^\mathbf{R}$ is the honeycomb lattice: under a non-degeneracy condition, the dispersion relation presents Dirac points. Thereafter, we describe the periodic rHF theory with three-dimensional Coulomb interactions and we state in Theorem~\ref{theo: rHF} that this model satisfies the conditions mentioned above.
	
	\subsection{Lattices}\label{sec:bravais-lattices-and-crystals}
	
	The scalar product of two vectors $\mathbf{u}$ and $\mathbf{v}$ of $\R^2$ is denoted by $\mathbf{u \cdot v}$ and the associated euclidean norm by $\abs{\mathbf{u}} = \sqrt{\mathbf{u \cdot u}}$. Let $(\mathbf{u}_1,\mathbf{u}_2)$ be a basis of $\R^2$. We consider the two-dimensional \emph{Bravais lattice}
	\begin{align*}
	\lscr \coloneqq \Z \mathbf{u}_1+\Z \mathbf{u}_2 \subset \R^2 \pt
	\end{align*}
	We denote by $\Gamma$ its \emph{Wigner-Seitz cell}. This is a choice of primitive cell whose interior consists of the vectors which are closer to the origin than any other vertex of $\mathscr{L}$. The \emph{reciprocal lattice }$\mathscr{L}^*$ of $\mathscr{L}$ is also a Bravais lattice and is defined by
	\begin{align*}
	\mathscr{L}^*\coloneqq \Z \mathbf{v}_1 \oplus \Z \mathbf{v}_2 \, ,
	\end{align*}
	where the reciprocal basis is determined by the orthogonality relations $\mathbf{v}_i \cdot \mathbf{u}_j = 2\pi\delta_{ij}$ for all $(i,j)\in\{1,2\}^2$. The Wigner-Seitz cell of the reciprocal lattice $\mathscr{L}^*$, denoted by $\Gamma^*$, is called the \emph{first Brillouin zone}.
	
	In the sequel, we consider lattices formed as a superposition of several shifted copies of $\mathscr{L}$. Let $N \in \N$ be the number of vertices per unit cell and $\mathbf{R} = (\mathbf{r}_1,\dots,\mathbf{r}_N) \in \Gamma^N$ (where $\mathbf{r}_i \neq \mathbf{r}_j$ if $i\neq j$) be their positions in $\Gamma$. The lattice associated with $\mathbf{R}$ is 
	\begin{align}
	\label{eq: many_sites_lattice}
	\mathscr{L}^\mathbf{R} \coloneqq  \mathscr{L} + \mathbf{R}= \enstq{\mathbf{u+r}}{\mathbf{u} \in \mathscr{L},~ \mathbf{r} \in \mathbf{R}} \pt
	\end{align}
	Let $L \geq 1$ be a length parameter. Thereafter, we will use the subscript $L$ to denote the dilation by a factor $L$. For instance, we write
	\[
	\lscr_L = L\lscr,\quad\Gamma_L = L \Gamma , \quad \mathscr{L}_L^* = L^{-1} \mathscr{L}^* ,\quad \Gamma_L^* = L^{-1} \Gamma^* \et \mathscr{L}_L^\mathbf{R} = L\mathscr{L}^\mathbf{R} \pt
	\]
	Let $\textbf{k} \in \R^2$ and $p\in \intff{1}{\infty}$. The space of locally $p$-integrable functions satisfying pseudo-periodic boundary conditions with quasi-momentum $\mathbf{k}$ is denoted by
	\begin{align*}
	L^p_\textbf{k}(\Gamma_L) \coloneqq \enstq{\varphi \in L^p_\mathrm{loc}(\R^2,\C)}{\forall \textbf{u} \in\mathscr{L}_L,~\forall \textbf{x} \in \R^2,~\varphi(\textbf{x}+\textbf{u}) = e^{i \textbf{k} \cdot \textbf{u}} \varphi(\textbf{x}) ~\mathrm{a.e}} \pt
	\end{align*}
	We will also denote by $L^p_\mathrm{per}(\Gamma_L) \coloneqq  L^p_{0}(\Gamma_L)$ the space of locally $p$-integrable functions which are invariant under the shifts of $\mathscr{L}_L$.
	
	\subsection{A class of periodic operators on \texorpdfstring{$\lscr_L$}{LL}}\label{sec:a-class-of-periodic-operators-on-lscrl}
	
	We want to study nonlinear models where the effective potential $V_L$ is close but not exactly given by a $\lscr_L$-periodic superposition of potentials wells, see \eqref{eq: almost_superposition_intro}. In this section, we describe the class of Schrödinger operators we consider. 
	
	\subsubsection{Periodic operator}\label{sec:periodic-operator}
	
	Our first assumption is about the local singularities that $V_L$ may present.
	\begin{hyp}[Singularities]
		\label{hypo_0}
		We consider a family $\{V_L\}_{L\geq 1}$ of real-valued potentials on $\R^2$ such that
		\begin{enumerate}[noitemsep, label=(\roman*),ref= \ref*{hypo_0}\textit{(\roman*)}]
			\item \label{hypo_01} $V_L \in \lper^p(\Gamma_L)$ for some $p\in\intoo{1}{\infty}$;
			\item \label{hypo_04} $\norme{V_L}_{\lper^\infty(\Gamma_L)}$ goes uniformly to zero at distance $L$ of the vertices of $\lscr_L^\mathbf{R}$ when $L$ goes to infinity:
			\begin{align*}
			\forall \rho>0,\quad \lim\limits_{L\to\infty}\normeL{V_L \mathds{1}_{\dist(\cdot,\lscr_L^\mathbf{R}) \geq L\rho}}_{\lper^\infty(\Gamma_L)} = 0 \pt
			\end{align*}
		\end{enumerate}
	\end{hyp}
	\begin{rem}
	By Assumption~\ref{hypo_01}, the potential $V_L$ may present local singularities of the form $\abs{\mathbf{x}}^{-\alpha}$ with $\alpha < 2$, including the three-dimensional Coulomb singularity $\alpha=1$. For simplicity, their location is constrained by Assumption~\ref{hypo_04} to the vertices of $\lscr^\mathbf{R}_L$. Also, the dependence of $V_L$ on $L$ can be highly nonlinear as long as its $L^p$-norm does not blow up faster than polynomials (see estimate~\eqref{rem: grow_polynomial} below). Later in Section~\ref{sec:the-rhf-model-for-periodic-systems-at-dissociation}, we will consider a nonlinear model, namely the periodic reduced Hartree-Fock model, and show that the corresponding $V_L$ satisfies Assumption~\ref{hypo_0}.
	\end{rem}
	In order to simplify the analysis, we assume that $\lscr_L$ and $V_L$ are invariant under the same symmetry group. If a group $G$ acts on some set $X$ then the action of $g \in G$ on $x \in X$ will be denoted by $g \cdot x$ or $g[x]$. Let $G \subset E_2(\R)$ be a subgroup of $E_2(\R)$, the \emph{Euclidean group }(or group of isometries) of $\R^2$. The action of $G$ on $\R^2$ is defined by $g\cdot \mathbf{x} \coloneqq g\mathbf{x}$ and its action on measurable functions by $(g\cdot v)(\mathbf{x}) \coloneqq v(g^{-1}\mathbf{x})$.
	
	The \emph{symmetry group }of a periodic two-dimensional pattern is the group of euclidean transformations $G$ leaving this pattern invariant. There exists only 17 distinct classes of such groups, called \emph{wallpaper groups }(or plane crystallographic groups)~\cite{martin1982transformation, armstrong1988groups}. A fundamental domain is a subset which contains exactly one point from each orbit of the action of $G$. Then, the pattern is uniquely determined by the specification of a fundamental domain and its symmetry group $G$.
	
	We denote by $G\subset E_2(\R)$ the symmetry group of $\lscr^\mathbf{R}$.
	The following assumption means that any fundamental domain of the lattice $\lscr^\mathbf{R}$ contains exactly one vertex.
	\begin{hyp}[$\lscr^\mathbf{R}$ has a single orbit]
		\label{hypo_1}
		The lattice $\lscr^\mathbf{R}$ is \emph{vertex-transitive:} the group $G \subset E_2(\R)$  acts \emph{transitively }on $\lscr^\mathbf{R}$.
	\end{hyp}
	\begin{rem} 
		Assumption~\ref{hypo_1} implies a constraint on the number of sites per primitive cell. An enumeration of possibilities shows that $N\in\{1, 2, 3, 4, 6, 8,9, 18, 36\}$. Later, we explain how one could relax this assumption (see Remark~\ref{rem: relaxing_hypo}).
	\end{rem}	
	
	We denote by $G_L$ the symmetry group of $\lscr_L^\mathbf{R}$ which has the same point group as $G$.
	Our next assumption states that $V_L$ has all the same symmetries as $\lscr_L^\mathbf{R}$. In particular, if the potential $V_L$ presents a singularity at one vertex of $\lscr^\mathbf{R}_L$ then the same singularity appears at each vertex, up to an orthogonal transformation.
	\begin{hyp}[Symmetries of $V_L$]
		\label{hypo_2}
		For all $L \geq 1$, the potential $V_L$ is invariant under the action of $G_L$: $\forall g\in G_L$, $g\cdot V_L = V_L$.
	\end{hyp}
	Now, we introduce
	\[
	\boxed{H_L = -\Delta + V_L \, ,}
	\]
	the $\mathscr{L}_L$-\emph{periodic operator }on $L^2(\R^2)$ associated with the potential $V_L$. 
	It is well known that the operator $H_L$ is bounded from below (see for instance~\cite[Section 1.2]{cycon1987schrodinger} or Proposition~\ref{lemma: kato_type_inequality}). Since we have assumed that $V_L\in \lper^p(\Gamma_L)$ for some $p\in\intoo{1}{\infty}$, we consider the Friedrichs self-adjoint extension of this operator. It admits the decomposition in fibers~\cite[Section XIII.16]{reed1978methodsIV} (see also~\cite[Section 2]{knauf1989coulombic})
	\[
	H_L \simeq \fint_{\Gamma_L^*}^\oplus H_{L}(\mathbf{k}) \diff \mathbf{k}  \pt
	\]
	There, for all $\mathbf{k} \in \Gamma_L^*$, the operator $H_L(\mathbf{k})=-\Delta + V_L$ acts on $\lk^2(\Gamma_L)$ and is self-adjoint on the domain
	\[
	\mathcal{D}(H_L(\mathbf{k})) = \enstq{\varphi \in H^1_\mathbf{k}(\Gamma_L)}{\left(-\Delta + V_L\right) \varphi \in L^2_\mathbf{k}(\Gamma_L) }\, ,
	\]
	where $H^1_\mathbf{k}(\Gamma_L) = \enstq{\varphi \in \lk^2(\Gamma_L)}{\partial_1 \varphi, \partial_2 \varphi \in \lk^2(\Gamma_L)}$ denotes the Sobolev space with pseudo-periodic boundary conditions. Moreover, because $H_L(\mathbf{k})$ has a compact resolvent, its spectrum is purely discrete and accumulates at $+\infty$
	\[
	\sigma(H_L\mathbf{(k)}) = \{\mu_{1,L}(\mathbf{k}) \leq \mu_{2,L}(\mathbf{k}) \leq \dots\} \pt
	\]
	The maps $\mathbf{k} \mapsto \mu_{n,L}(\mathbf{k})$ for $n\geq 1$ are called \emph{band functions}. They are Lipschitz and piecewise analytic functions of $\mathbf{k}$ (consequence of~\cite[Proposition 2.3]{birman99periodic} and Hartogs's theorem~\cite{krantz2001function}). The operator $H_L$ has only absolutely continuous spectrum~\cite{knauf1989coulombic, birman99periodic} given by the range of the band functions~\cite[Theorem XIII.85]{reed1978methodsIV}
	\[
	\sigma(H_L) = \bigcup_{\mathbf{k} \in \Gamma_L^*} \sigma(H_L(\mathbf{k})) =  \bigcup_{n \in \N^*} \mu_{n,L}(\Gamma_L^*) \pt
	\]
	The purpose of this article is to study the geometry of the spectral bands of the periodic operator $H_L$ in the \emph{dissociation }regime, that is, for $L\gg 1$.
	
	\subsubsection{Reference operator}\label{sec:reference-operator}
	
	So far, the assumptions on the family $\{V_L\}_{L\geq 1}$ do not give any information about the local behavior of $V_L$ when $L$ is large. We want $V_L$ to be approximately given by a periodic superposition of potential wells. This motivates the introduction of a \emph{reference potential }$V$ which we will assume to be the limit of $V_L$ at each vertex of $\lscr^\mathbf{R}_L$.
	\begin{hyp}[Reference potential]
		\label{hypo_3}
		Let $V$ be a real-valued potential such that
		\begin{enumerate}[noitemsep, label=(\roman*),ref= \ref*{hypo_3}\textit{(\roman*)}]
			\item \label{hypo_31} $V\in L^p(\R^2)+L^\infty(\R^2)$ where $p \in \intoo{1}{\infty}$ is the number introduced in Assumption~\ref{hypo_01};
			\item \label{hypo_32} $V(\mathbf{x}) = \grandO{\abs{\mathbf{x}}^{-1-\epsilon}}$ as $\abs{\mathbf{x}} \to \infty$ for some $\epsilon>0$;
			\item \label{hypo_33} The associated mono-atomic Schrödinger operator defined by 
			\begin{align}
			\label{eq: reference_operator}
			\boxed{H = -\Delta + V\, ,}
			\end{align}
			admits at least one negative eigenvalue.
		\end{enumerate}
	\end{hyp}
	The potential $V$ belongs the Kato class~\cite[Definition 1.10]{cycon1987schrodinger} hence is infinitesimally $(-\Delta)$-form bounded. We always work with the Friedrichs extension of $H$ which defines a self-adjoint operator on the domain
	\begin{align*}
	\mathcal{D}(H)  = \enstq{u \in H^1(\R^2)}{(-\Delta + V)u \in L^2(\R^2)} \pt
	\end{align*}
	By standard perturbation theory~\cite{kato1995perturbation, reed1978methodsIV}, its essential spectrum is given by $\sigma_\mathrm{ess}(H) = \intfo{0}{\infty}$ and the discrete spectrum of $H$ is negative. We denote it by
	\begin{align*}
	\sigma_\mathrm{d}(H) = \{ - \mu_1 < - \mu_2 \leq - \mu_3 \leq \dots \leq 0 \} \subset \intof{-\infty}{0}\pt
	\end{align*}
	To lighten the notation, we also denote by $-\mu$ the lowest eigenvalue of $H$, which is non-degenerate~\cite{goelden1977nondegeneracy}, by $g = \mu_1 - \mu_2 >0$ the \emph{spectral gap }above $-\mu$ and by $v$ the associated normalized eigenfunction. We have
	\begin{align*}
	Hv = (-\Delta+V)v =-\mu v \, ,
	\end{align*}
	where all the terms make sense in $H^{-1}(\R^2)$.
	
	\subsubsection{Effective mono-atomic operators}\label{sec:effective-mono-atomic-operator}
	
	In this section, we introduce effective operators by localizing the periodic potential $V_L$ around a vertex of $\lscr_L^\mathbf{R}$. The eigenfunctions and eigenvalues of these operators will be useful in order to precisely approximate the modes and the dispersion relation of the periodic operator $H_L$.
	\begin{figure}[h]
		\centering
		\begin{subfigure}[t]{0.45\textwidth}
			\centering
			\includegraphics{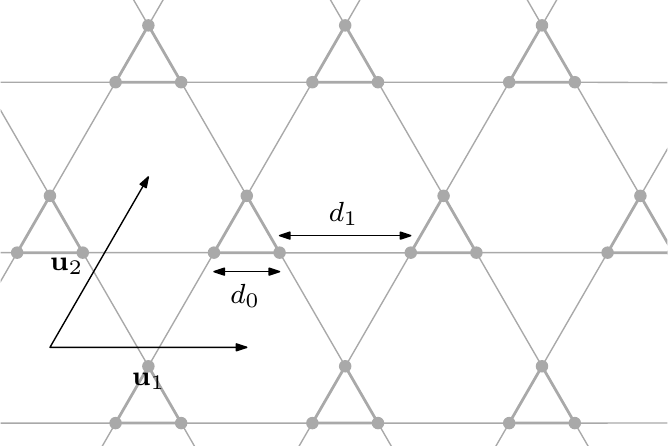}
			\caption{$\lscr=\mathbf{u}_1\Z+\mathbf{u}_2\Z$ is the triangular lattice, $N = \absL{\mathbf{R}}=3$ and $m=1$. Vertices linked by thick gray lines (resp. thin gray lines) are nearest neighbors (resp. second nearest neighbors). The nearest neighbor distance $d_0$ and the second nearest neighbor distance $d_1>d_0$ are displayed.}
			\label{fig:nn_lattice}
		\end{subfigure}
		\hfill
		\begin{subfigure}[t]{0.45\textwidth}
			\centering
			\includegraphics{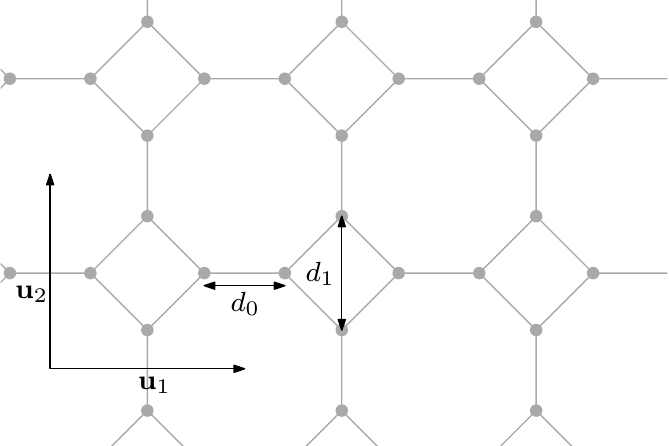}
			\caption{$\lscr=\mathbf{u}_1\Z+\mathbf{u}_2\Z$ is the square lattice, $N = \absL{\mathbf{R}}=4$ and $m=2$. Vertices linked by gray lines are nearest neighbors (none second nearest neighbor link is shown). The nearest neighbor distance $d_0$ and the second nearest neighbor distance $d_1>d_0$ are displayed. There are two kinds of edges: octogon-octogon edges and square-octogon edges. Thus $m=2$.} 
			\label{fig:square-octogon}
		\end{subfigure}
		\caption{Examples of lattices $\lscr^\mathbf{R}$ with different parameters.}
		\label{fig:examples}	
	\end{figure}
	
	We introduce the nearest neighbor distance $d_0 >0 $ of the unscaled lattice $\mathscr{L}^\mathbf{R}$ (see Figure~\ref{fig:examples})
	\begin{align}
	\label{eq: min_distance_lattice}
	d_0 \coloneqq \min \left\{\abs{\mathbf{r-r'}}\mathrel{~}\middle|\mathrel{~} (\mathbf{r,r'}) \in \left(\mathscr{L}^\mathbf{R}\right)^2 \et \mathbf{r\neq r'}\right\} \pt
	\end{align}
	We also denote by $d_1 > d_0$ the second nearest neighbor distance. Let $\delta \in \intoo{0}{1/2}$
	and $\chi \in \mathcal{C}^\infty_c(\R^2)$ a localization function such that
	\begin{align}
	\label{eq: definition_chi}
	0 \leq \chi \leq 1,\quad\chi \equiv 1 \quad \text{on} \quad B\left(0,\frac{1+\delta}{2}d_0\right) \et \supp \chi \subset  B\left(0,\left(\frac{1}{2}+\delta\right)d_0\right) \pt
	\end{align}
	Notice that we have $\frac{1}{2}d_0 < \left(\frac{1}{2}+\delta\right)d_0 < d_0$. In addition, we also require $\chi$ to be \emph{radial} and to satisfy the technical assumption
	\begin{align}
	\label{eq: chi_regularity}
	\sqrt{1-\chi} \in \mathcal{C}^1(\R^2) \pt
	\end{align}
	For $L\geq 1$ and $\mathbf{r} \in \mathscr{L}^\mathbf{R}$, we introduce the localization functions near the vertex $L\mathbf{r}$
	\begin{align}
	\label{eq: definition_Vlr}
	\chi_{L,\mathbf{r}}(\mathbf{x}) \coloneqq \chi (L^{-1}\mathbf{x} - \mathbf{r}) \et V_{L,\mathbf{r}}\coloneqq \chi_{L,\mathbf{r}} V_L \pt
	\end{align}
	The potential $V_{L,\mathbf{r}}$ belongs to $L^p(\R^2)$ and is compactly supported within the ball $B(L\mathbf{r},\left(\frac{1}{2}+\delta \right)Ld_0)$ which contains one and only one vertex of $\lscr_L^\mathbf{R}$. We recall that the space $L^p(\R^2)+L^\infty(\R^2)$, endowed with 
	\begin{align*}
	\norme{V}_{L^p(\R^2)+L^\infty(\R^2)} \coloneqq \inf\enstq{\norme{V_1}_{L^p(\R^2)}+\norme{V_2}_{L^\infty(\R^2)}}{V = V_1+V_2,~V_1\in L^p(\R^2),~V_2\in L^\infty(\R^2)} \, ,
	\end{align*}
	defined a Banach space. Our next assumption is that, up to a translation, $V_{L,\mathbf{r}}$ is asymptotically given by the reference potential $V$ when $L \to\infty$.
	\begin{hyp}[The wells are asymptotically equivalent to $V$]
		\label{hypo_4}
		For all $\mathbf{r}\in\lscr^\mathbf{R}$, we have
		\begin{align*}
		\lim\limits_{L\to\infty}\normeL{V_{L,\mathbf{r}} - V(\cdot - L\mathbf{r})}_{L^p(\R^2)+L^\infty(\R^2)} = 0 \pt
		\end{align*}
	\end{hyp}
	A first consequence of Assumption~\ref{hypo_4} is that the discrete spectrum of the \emph{effective mono-atomic} Schrödinger operator associated with $V_{L,\mathbf{r}}$, defined by
	\begin{align}
	\label{eq: effective_mono_atomic_operator}
	\boxed{H_{L,\mathbf{r}} \coloneqq -\Delta + V_{L,\mathbf{r}} \, ,}
	\end{align}
	is non-empty (see~\cite[Section XII-3]{kato1995perturbation} and also Proposition~\ref{prop: perturbation_theory} in Appendix~\ref{sec:perturbation-theory-for-singular-potentials}). In addition, its lowest eigenvalue is non-degenerate. By symmetry arguments, we can show that the operators $\{H_{L,\mathbf{r}}\}_{\mathbf{r} \in \mathscr{L}^\mathbf{R}}$ are unitarily equivalent and thus share the same spectrum (see Section~\ref{sec:properties-of-the-mono-atomic-operators} for the details). We denote by $-\mu_L$ their common lowest eigenvalue and by $v_{L,\mathbf{r}}$ the associated normalized eigenfunction:
	\begin{align}
	\label{eq: eigen_equation_Hlr}
	H_{L,\mathbf{r}} v_{L,\mathbf{r}} = (-\Delta + V_{L,\mathbf{r}})v_{L,\mathbf{r}} = -\mu_Lv_{L,\mathbf{r}} \, ,
	\end{align}
	where each term makes sense in $H^{-1}(\R^2)$.
	
	\subsubsection{Convergence to the tight-binding model}
	
	In this section, we state the main result of this article. The first theorem provides expansion to the leading order of the dispersion of $-\Delta+V_L$ when $V_L$ satisfies the assumptions enumerated in the previous section.
	
	We introduce the set of nearest neighbors pairs (see Figure~\ref{fig:examples})
	\begin{align}
	\label{eq: definition_NN}
	\mathscr{P}^\mathbf{R} \coloneqq \enstqbis{\{\mathbf{r,r'}\} \in \lscr^\mathbf{R} \times \lscr^\mathbf{R}}{\abs{\mathbf{r-r'}}=d_0} \pt
	\end{align}
	The group $G$ acts isometrically on $\lscr^\mathbf{R}$ hence it also defines an action on $\mathscr{P}^\mathbf{R}$. When this action  is transitive, we say that $\mathscr{P}^\mathbf{R}$ is \emph{edge-transitive}. However, this is \emph{not }necessarily the case (see Figure~\ref{fig:square-octogon} for an example of non edge-transitive lattice). Because $\mathscr{P}^\mathbf{R}$ is invariant under the shifts of $\lscr$, this action has finitely many orbits, which are denoted by $\mathcal{O}_1,\dots,\mathcal{O}_m$. For all $k \in \{1,\dots,m\}$, we consider a representative $\mathbf{p}_k = \{\mathbf{r}_k,\mathbf{r}_k'\}\in\mathcal{O}_k$ and we introduce the following interaction coefficient
	\begin{align}
	\label{eq: interaction_coefficient}
	\boxed{\theta_{L,k}\coloneqq  \langle v_{L,\mathbf{r}_k}, V_{L,\mathbf{r}_k}(1-\chi_{L,\mathbf{r}_k'})v_{L,\mathbf{r}_k'}\rangle_{L^2(\R^2)} \, ,}
	\end{align}
	where $V_{L,\mathbf{r}}$ and $v_{L,\mathbf{r}}$ are respectively defined in \eqref{eq: definition_Vlr} and \eqref{eq: eigen_equation_Hlr}. Using Assumption~\ref{hypo_2}, we can show that the quantity $\theta_{L,k}$ does not depend on the choice of the pair $\mathbf{p}_k \in \mathcal{O}_k$ (see Section~\ref{sec:properties-of-the-mono-atomic-operators}). In addition, we show in Proposition~\ref{prop: interaction_matrix} that $\theta_{L,k}$ is exponentially small when $L$ is large: for any $\epsilon >0$, there exists $C_\epsilon$ such that
	\begin{align*}
	\abs{\theta_{L,k}} \leq C_\epsilon e^{-(1-\epsilon)\sqrt{\mu}d_0L} \pt
	\end{align*}
	It is expected that in many cases this is essentially optimal, that is, $\theta_{L,k}$ is of order $e^{-\sqrt{\mu}d_0L}$ up to polynomial factors.

	The following theorem is the main result of this article.
	\begin{theo}[Convergence to the tight-binding model]
		\label{theo: feshbach-schur}
		Let $\delta\in\intoo{0}{1/2}$ be the parameter introduced in the definition \eqref{eq: definition_chi} of the cut-off function $\chi$. Under Assumptions~\ref{hypo_0}--\ref{hypo_4}, for all $\epsilon\in\intoo{0}{\delta}$ and for all $\mathbf{k} \in \Gamma_L^*$, the $N$ first Bloch eigenvalues satisfy
		\begin{align}
		\label{eq: main_theorem}
		\boxed{\mu_{j,L}(\mathbf{k}) = -\mu_L + \lambda_j \left(\sum_{k=1}^m \theta_{L,k} B_k(L\mathbf{k})\right) + \grandO{e^{-(1+\delta - \epsilon)\sqrt{\mu}d_0L}+e^{-(1 - \epsilon)\sqrt{\mu}d_1L}} \, ,}
		\end{align}
		for $L$ large enough and where the $O$ is independent from $\mathbf{k}$. Here $d_1>d_0$ denotes the second nearest neighbor distance in $\lscr^\mathbf{R}$, $\lambda_j(M)$ denotes the $j$\textsuperscript{th} lowest eigenvalue of a matrix $M$ and $B_k(L\mathbf{k})$ is the $N\times N$ matrix defined by
		\begin{align}
		\label{eq: definition_B}
		\forall \mathbf{k} \in \Gamma^*,\quad\forall (\mathbf{r,r'}) \in \mathbf{R}^2,\quad B_k(\mathbf{k})_{\mathbf{r,r'} }= \sum_{\substack{\mathbf{u} \in \lscr \\ (\mathbf{\mathbf{r},\mathbf{u+r'}}) \in \mathcal{O}_k}} e^{i\mathbf{k}\cdot \mathbf{u}}  \pt
		\end{align}
	\end{theo}	
	If $m=1$ then the second term in \eqref{eq: main_theorem} is the tight-binding model associated with the crystal $\lscr^\mathbf{R}$. In the case of the honeycomb lattice (which is introduced in Section~\ref{sec:the-triangular-and-the-hexagonal-lattices}), we have $N=2$ and $m=1$. The matrix $B_\mathrm{HC}(\mathbf{k})\coloneqq B_1(\mathbf{k})$ is given by
	\begin{align*}
	B_\mathrm{HC}(\mathbf{k}) = \begin{pmatrix}
	0 & 1+e^{i\mathbf{k\cdot u}_1}+e^{i\mathbf{k\cdot u}_2}  \\
	1+e^{-i\mathbf{k\cdot u}_1}+e^{-i\mathbf{k\cdot u}_2}& 0
	\end{pmatrix} \pt
	\end{align*}
	This is the matrix associated with the tight-binding model of graphene, also known as the \emph{Wallace model}~\cite{wallace1947bandtheory}. The dispersion relation, $\mu_\pm(\mathbf{k}) = \pm \abs{ 1+e^{i\mathbf{k\cdot u}_1}+e^{i\mathbf{k\cdot u}_2}}$, exhibits \emph{Dirac points} at the six vertices of the first Brillouin zone, see Figure~\ref{fig:chap_wallace}. In this case, a more precise result is provided later in Theorem~\ref{theo: existence_dirac_cones}.
	\begin{figure}
		\centering
		\includegraphics[width=0.8\textwidth]{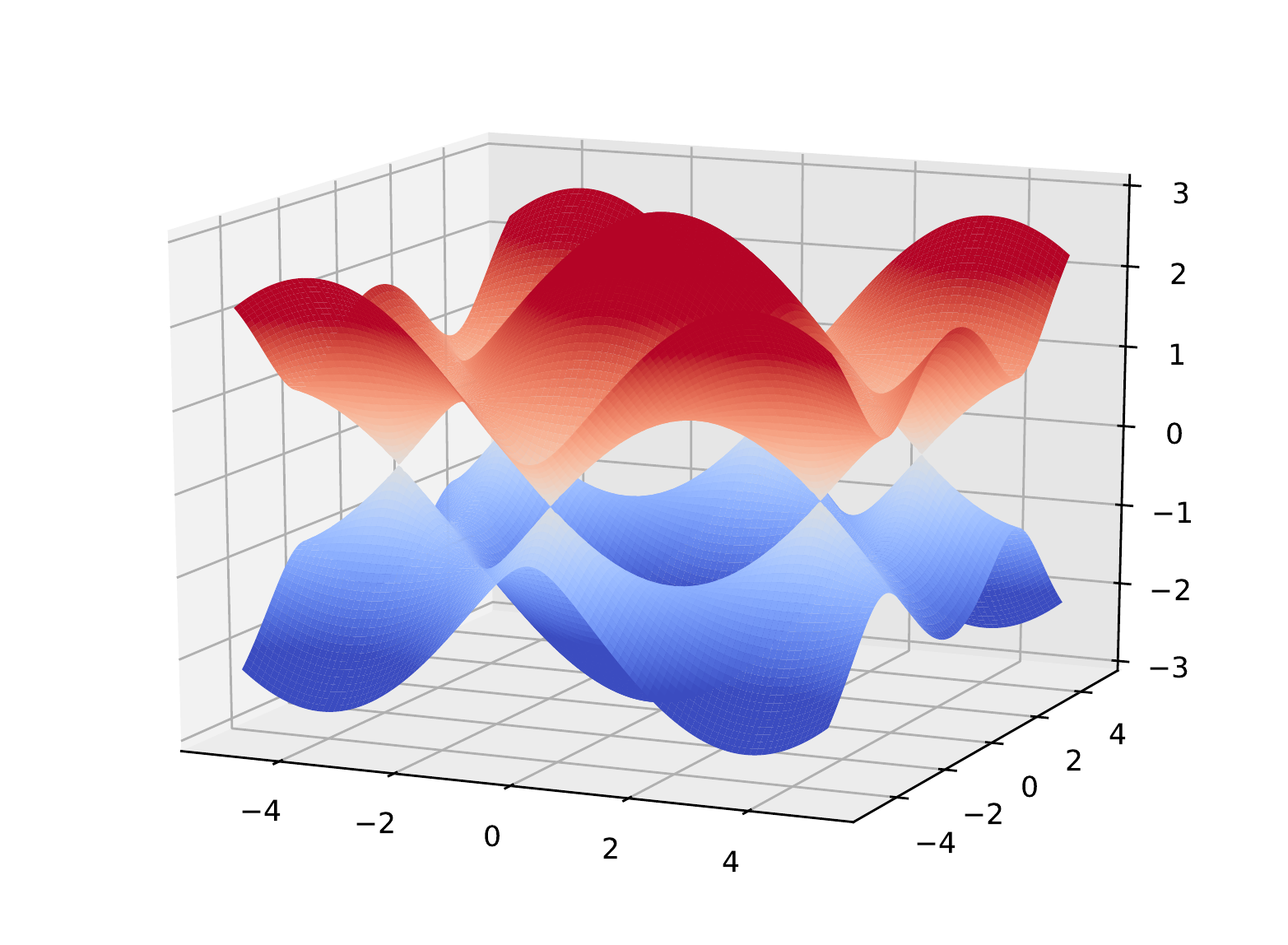}
		\caption{The Wallace model. The dispersion relation is invariant with respect to the rotation by $\pi/3$ about the origin. We observe conical singularities, Dirac points, at the vertices of the Brillouin zone $\Gamma^*$.}
		\label{fig:chap_wallace}
	\end{figure}
	
	For generic potentials, we expect $\theta_{L,k}$ to be non zero and of order $e^{-\sqrt{\mu}d_0L}$. If this is indeed the case then the tight binding model defined by the matrices \eqref{eq: definition_B} would give the leading order of the dispersion relation of $H_L$.
	
	Showing that the interaction coefficients $\theta_{L,k}$ are really of order $e^{-\sqrt{\mu}d_0L}$ seems challenging in the general setting we consider. In~\cite{fefferman2017honeycomb}, Fefferman, Weinstein and Lee-Thorp consider superpositions of localized potential wells centered on the vertices of $\lscr^H$, that is the case $V_L = L\sum_{\mathbf{r}\in\lscr^H} V^\mathrm{at}(\cdot- \mathbf{r})$ with $V^\mathrm{at}$ bounded. Under some symmetry, support and spectral assumptions on $V$, they are able to show that the dispersion relation of $-\Delta+V_L$ converges uniformly toward the Wallace model in the high contrast regime $L\to\infty$. In particular, they show that the interaction coefficient satisfies the bounds $e^{-c_1L} \lesssim \theta_L \lesssim e^{-c_2L}$ for some constants $c_1<c_2$ depending on $V^\mathrm{at}$. In~\cite{daumer1993periodique}, Daumer considers an exact periodic superposition of smooth potential wells which decay polynomially. For an isolated band and in the dissociation regime, the first order of the interaction coefficient is determined and shown to be of order $e^{-\sqrt{\mu}d_0L}$, corrected by explicit polynomial factors.
	
	Many results similar to Theorem~\ref{theo: feshbach-schur} have appeared in the literature. When the potential is homogeneous of degree $-1$, the regime $L\to\infty$ is equivalent to a semiclassical limit in which multiple wells potentials have been studied. In~\cite{simon1984semiclassicalIII}, Simon shows that, in the semiclassical regime, the width of the ground state band of a Schrödinger operator with smooth and periodic potential is given by the minimum action among all instantons connecting two distinct minima of the external potential. We also refer to the series of papers by Helffer and Sjöstrand~\cite{helffer1984multipleI, helffer1985puitsII, helffer1985multipleIII, helffer1985puitsIV, helffer1986puitsV, helffer1987puitsVI} and, in a periodic setting, by their collaborators~\cite{outassourt1984, mohamed1991estimations}. See also~\cite{daumer1996schrodinger} where Daumer considers finitely many wells at dissociation. In general, the results from this literature are more precise (for instance, in~\cite{daumer1996schrodinger}, the author determines exactly the tunneling coefficient) but the dependence, when it exists, of the potential on the semiclassical parameter is easier to handle than in our setting. Indeed, Assumption~\ref{hypo_0} allows for potentials with a highly nonlinear dependence in $L$, which is needed for the study of some nonlinear quantum models at dissociation (see Section~\ref{sec:reduced-periodic-hartree-fock-model-at-dissociation}).
	
	The overall strategy for proving Theorem~\ref{theo: feshbach-schur} is the following. We first study the projection of $H_L$ on the subspace spanned by the family $\{v_{L,\mathbf{r}}\}_{\mathbf{r}\in\lscr^\mathbf{R}}$ which approximates the spectral subspace associated with the $N$ low-lying bands of $H_L$. Then we use Feshbach-Schur method to recover the exact spectrum of $H_L(\mathbf{k})$. The Feshbach-Schur method, used to reduce the dimension of perturbative eigenvalue problems, was first developed by Schur in matrix theory~\cite{schur1917uber} and by Feshbach in nuclear physics~\cite{feshbach1958unified}. It has been reformulated by Bach, Fröhlich and Sigal in~\cite{bach1998renormalization, bach1998quantum}. It was used in the context of periodic operators in~\cite{fefferman2012honeycomb, fefferman2017honeycomb}.
	\begin{rem}
		\label{rem: relaxing_hypo}
		In fact, Assumption~\ref{hypo_1} is not necessary for the purpose of this work even if it will make the analysis clearer. Without this assumption, the number of vertices in each fundamental domain of $\lscr^\mathbf{R}$ (which is equal to the number of orbits of $\lscr^\mathbf{R}$ under the action of $G$) is not restricted to one anymore. In this more general situation, a result analogous to Theorem~\ref{theo: feshbach-schur} would hold if we associate to each orbit a reference potential whose lowest eigenvalue could differ from the one of referential potentials associated with the other orbits. Indeed, in this case, one could decompose $H_L$ into a direct sum of single orbit operators which do not interact to leading order.
	\end{rem}
	
	\subsubsection{The honeycomb lattice case}
	\begin{figure}[h]
		\centering
		\begin{subfigure}[t]{0.45\textwidth}
			\centering
			\includegraphics{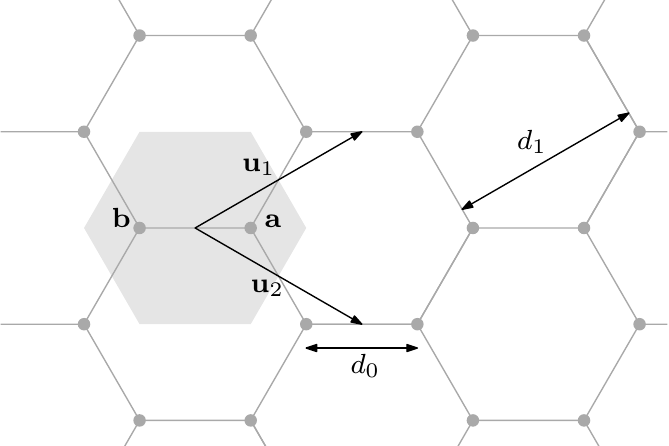}
			\caption{$\lscr = \mathbf{u}_1\Z+\mathbf{u}_2\Z$ is the triangular lattice, $N=\abs{\mathbf{R}}=2$ and $m=1$. The honeycomb lattice is the superposition of two sifted versions of $\lscr$ : $\lscr^H = (\lscr+\mathbf{a})\cup(\lscr+\mathbf{b})$. The Wigner-Seitz cell $\Gamma$ is colored in gray. Vertices linked by gray lines are nearest neighbors. The nearest neighbor distance $d_0$ and the second nearest neighbor distance $d_1>d_0$ are displayed.}
			\label{fig:honeycomb_lattice}
		\end{subfigure}
		\hfill
		\begin{subfigure}[t]{0.45\textwidth}
			\centering
			\includegraphics{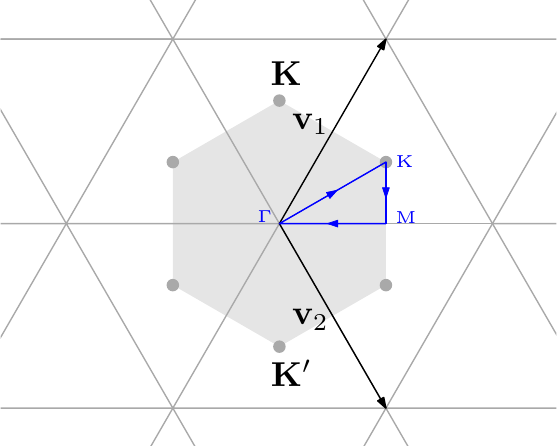}
			\caption{The vectors $\mathbf{v}_1$ and $\mathbf{v}_2$ generate the reciprocal lattice $\lscr^*$. The first Brillouin zone $\Gamma^*$ is colored in gray. The vertices $\mathbf{K}$ and $\mathbf{K}'$ of $\Gamma^*$ are shown, the others are obtained by rotation by $2\pi/3$ and $4\pi/3$ about the origin. The blue path corresponds to the section joining the special symmetry points $\Gamma$, $\mathrm{K}$ and $\mathrm{M}$ of the Brillouin zone. It is used in the chemistry literature for displaying dispersion relation data \protect\cite{cooper2012experimental}.}
			\label{fig:triangular_lattice_BZ}
		\end{subfigure}
		\caption{The honeycomb lattice $\lscr^H$.}	
		\label{fig:honeycomb_lattice_BZ}
	\end{figure}
	In estimate~\eqref{eq: main_theorem} of Theorem~\ref{theo: feshbach-schur}, the big $O$ does not explicitly depend on the quasi-momentum $\mathbf{k}\in\Gamma^*_L$. When $\lscr^\mathbf{R} = \lscr^H$ is the honeycomb lattice, we can make the estimates more precise in the vicinity of the vertices of the first Brillouin zone $\Gamma_L^*$. Then, under the additional assumption that the interaction coefficient has the expected order, we show that the dispersion relation has Dirac points at these vertices.
	
	In order to be more explicit, we first briefly recall some features of the honeycomb lattice (this lattice is fully described in Section~\ref{sec:the-triangular-and-the-hexagonal-lattices}). The honeycomb lattice is a two-sites lattice which forms a hexagonal tilling of the plane, see Figure~\ref{fig:honeycomb_lattice}. The underlying lattice is the triangular lattice and its first Brillouin zone is a regular hexagon whose vertices are denoted by $\mathbf{K}$ or $\mathbf{K}'$, see Figure~\ref{fig:triangular_lattice_BZ}. As already mentioned, there is only one interaction coefficient \eqref{eq: interaction_coefficient}, denoted by $\theta_L$. We also denote by $\mathbf{k}\in\Gamma^*_L\mapsto\mu_{\pm,L}(\mathbf{k})$ the two lowest band functions of $H_L$. Recall that $\delta\in\intoo{0}{1/2}$ is defined in~\eqref{eq: definition_chi} and that $d_0$ is the nearest neighbor distance, equal to $\frac{1}{\sqrt{3}}$ in this section.
	\begin{theo}[Dirac points]
		\label{theo: existence_dirac_cones}
		Assume that there exists $c >0$ and $\delta' \in \intfo{0}{\delta}$ small enough such that $\abs{\theta_L} \geq c e^{-(1+\delta')\sqrt{\mu}d_0L}$.
		Let $r>0$ and let $\mathbf{K}_\star \in \{\mathbf{K},\mathbf{K}'\}$ be a vertex of the first Brillouin zone $\Gamma^*$. Then, when $L \to \infty$, we have the expansion  
		\begin{align*}
		\mu_{\pm,L}\left(\frac{\mathbf{K}_\star + \kappa}{L}\right) = -\mu_L + \petito{\abs{\theta_L}} \pm \frac{\sqrt{3}}{2} \abs{\theta_L} \abs{\kappa} \left(1+E(\kappa)\right)(1+\petito{1})\, ,
		\end{align*}
		where $\abs{E(\kappa)} \leq C \abs{\kappa}$ for all $\kappa \in B(0,r)$ and where the $o$'s do not depend on $\kappa$.
	\end{theo}
	For the sake of simplicity, we have only stated Theorem~\ref{theo: existence_dirac_cones} in the case of the honeycomb lattice. This case corresponds to graphene, a layer of carbon atoms located at the vertices of a such lattice~\cite{neta2009electronic}. However, the conclusions of Theorem~\ref{theo: existence_dirac_cones} remain valid for any lattice $\lscr^\mathbf{R}$ having the same symmetries as graphene, that is, $\mathcal{PT}$-symmetry (parity / time-reversal symmetry) and rotation by $2\pi/3$ symmetry~\cite{berkolaiko2018symmetry}. Lattices with different symmetries are also expected to present Dirac points, see the review~\cite{wang2015rare}.
	
	\subsection{Application: the periodic reduced Hartree-Fock model}\label{sec:the-rhf-model-for-periodic-systems-at-dissociation}
	
	Now, we illustrate the use of Theorems~\ref{theo: feshbach-schur} and~\ref{theo: existence_dirac_cones} in a nonlinear situation. We consider the two-dimensional periodic reduced Hartree-Fock model with three-dimensional Coulomb interactions. In addition, we assume that the external potential is corrected by a pseudo-potential. This nonlinear model is obtained as the thermodynamic limit of the Hartree-Fock model where the exchange term is neglected~\cite{catto2001thermodynamic, cances2008newapproach, hainzl2009thermodynamicII}. It is the simplest model to describe graphene while taking interactions into account.
	
	If we assume that the particles interact through the three-dimensional Coulomb interaction $\frac{1}{\abs{\mathbf{x}}}$ then the $\mathscr{L}_L$-periodic \emph{interaction kernel}, denoted by $W_L$, is given by
	\begin{align}
	\label{eq:periodic_kernel}
	W_L = L^{-1}M' + \sum_{\mathbf{u}\in\lscr} \left(\frac{1}{\abs{\cdot - L\mathbf{u}}} - \frac{1}{\abs{\Gamma_L}} \int_{\Gamma_L} \frac{\diff \mathbf{y}}{\abs{\cdot - L\mathbf{u - y}}}\right) \, ,
	\end{align}
	for some constant $M'\in\R$ chosen so that $W_L\geq0$. The properties of the interaction kernel $W_L$ are given in Section~\ref{sec:properties-of-the-periodic-interaction-kernel-gl}. Let $\rho$ be a positive and $\mathscr{L}_L$-invariant locally finite measure. Its \emph{self-interaction energy }is defined by
	\[
	D_L(\rho,\rho)\coloneqq \iint_{\Gamma_L\times \Gamma_L} \rho(\mathbf{x}) W_L(\mathbf{x-y})\rho(\mathbf{y}) \diff \mathbf{x} \diff \mathbf{y} \in \intff{0}{\infty} \pt
	\]
	The \emph{interaction energy} of two $\mathscr{L}_L$-periodic charge distributions $\rho$ and $\mu$ is defined by
	\[
	D_L(\rho,\mu) \coloneqq \iint_{\Gamma_L\times \Gamma_L} \rho(\mathbf{x}) W_L(\mathbf{x-y})\mu(\mathbf{y}) \diff \mathbf{x} \diff \mathbf{y} \, ,
	\]
	whenever $D_L(\rho,\rho) < \infty$ and $D_L(\mu,\mu) < \infty$.
	
	We consider the lattice $\mathscr{L}_L^\mathbf{R}$ defined in \eqref{eq: many_sites_lattice} where a pointwise nucleus of charge $+1$ is placed at each vertex. Besides the three-dimensional Coulomb interaction $\frac{1}{\abs{\mathbf{x}}}$, we assume that the external potential induced by the lattice has an additional term, localized around each nucleus and which corresponds to a pseudo-potential, that is, an effective potential modeling the behavior of the core electrons which do not explicitly appear in the model~\cite{chelikowsky2011electrons}. Then the potential generated by $\mathscr{L}_L^\mathbf{R}$ has the following form 
	\begin{align}
	\label{eq:definition_V_periodic}
	- \sum_{\mathbf{r} \in \mathbf{R}} W_L(\cdot - L\mathbf{r}) + \sum_{\mathbf{r\in R}}\sum_{\mathbf{u}\in\lscr} V^\mathrm{pp}(\cdot - L(\mathbf{u+r})) \, ,
	\end{align}
	where $V^\mathrm{pp}\in L^p(\R^2)$ for some $p>1$. For simplicity, we also choose $V^\mathrm{pp}$ radial and compactly supported. These two additional assumptions are certainly not optimal and could be relaxed without great effort, see Remark~\ref{rem:assumptions_pseudo-potential}. We denote by $-W_L^\mathbf{R}$ (resp. $V_L^\mathbf{R}$) the left side (resp. the right side) of \eqref{eq:definition_V_periodic}. An \emph{admissible state }is an element of the space
	\begin{align}
	\label{eq: admissible_states}
	\mathcal{S}_{\mathrm{per},L} \coloneqq \enstq{\gamma = \gamma^* \in \mathcal{B}(L^2(\R^2))}{0 \leq \gamma \leq 1,~\forall \mathbf{u} \in \mathscr{L}_L,~ \tau_\mathbf{u} \gamma= \gamma \tau_\mathbf{u} \et \trm_{\lscr_L} ((1 - \Delta) \gamma) < \infty} \, ,
	\end{align}
	where $\tau_\mathbf{u}$ denotes the translation operator which shifts particles by $\mathbf{u}$ and where the symbol $\trm_{\lscr_L}$ denotes the trace per cell~\cite{cances2008newapproach}. 
	In this setting, the periodic\emph{ reduced Hartree-Fock model }(rHF) consists in solving the following minimization problem
	\begin{align}
	\label{eq: rHF_model}
	E_L \coloneqq \inf \enstqbis{\mathcal{E}_L(\gamma)}{\gamma \in \mathcal{S}_{\mathrm{per},L} \et \trm_{\mathscr{L}_L} (\gamma) = N/q} ,
	\end{align}
	where the periodic rHF energy functional is given by
	\begin{align}
	\label{eq: rHF_energy}
	\mathcal{E}_L(\gamma) \coloneqq \trm_{\lscr_L} (-\Delta \gamma) + \int_{\Gamma_L} \left(-W_L^\mathbf{R}+V_L^\mathbf{R}\right) \rho_\gamma + \frac{q}{2}D_L(\rho_\gamma,\rho_\gamma)\pt
	\end{align}
	The first term is the kinetic energy per unit cell of the infinitely many (valence) electrons. The second term is the Coulomb interaction with the lattice of nuclei, the third one the correction term, whereas the last term the (mean-field) electronic repulsion. Here the factor $q$ is the number of spin states ($q=2$ for electrons). We restrict ourselves to paramagnetic states, hence $q$ only shows up in the energy. The normalization $\trm_{\mathscr{L}_L} (\gamma) = N/q$ ensures the neutrality of the system. We do not include the energy of the nuclei in the cell $\Gamma_L$ because it does not play a role in our analysis. An adaptation of~\cite[Theorem 1]{cances2008newapproach} and~\cite[Theorem 2.1]{catto2001thermodynamic} shows that \eqref{eq: rHF_model} is well-posed in the sense that the minimization problem \eqref{eq: rHF_model} admits a unique minimizer $\gamma_L$. We denote by $\rho_L(\mathbf{x}) \coloneqq \gamma_L(\mathbf{x,x})$ its one-body density. Then $\gamma_L$ is the unique solution of the mean-field equation
	\begin{align*}
	\boxed{\gamma_L= \mathds{1}_{\intof{-\infty}{\epsilon_L}} \left(H_L^\mathrm{MF}\right) \ou H_L^\mathrm{MF} \coloneqq -\Delta -W_L^\mathbf{R}+V_L^\mathbf{R} +  q\rho_L *_L W_L\pt}
	\end{align*}
	Here, $H_L^\mathrm{MF}$ is the mean-field hamiltonian and $\epsilon_L\in \R$ is a Lagrange multiplier called the \emph{Fermi level} chosen to ensure that $\trm_{\lscr_L}(\gamma)=N/q$. The notation $*_L$ stands for convolution on $\Gamma_L$. Notice that the mean-field potential
	\begin{align*}
	V_L^\mathrm{MF} \coloneqq  - W_L^\mathbf{R}+V_L^\mathbf{R} + q\rho_L *_L W_L\, ,
	\end{align*}
	depends on $L$ in a highly nonlinear way. If the action of $G_L$ on $\lscr_L$ is free then we can regroup $-W_L^\mathbf{R}$ and $q\rho_L *_L W_L$ and write $V_L^\mathrm{MF}$ as in \eqref{eq: mean-field_equation_intro} plus the correction $V_L^\mathbf{R}$ thanks to the fact that there is no charge and no dipole.
	
	Now, we describe the reference potential occurring in the limit $L\to\infty$. We consider the minimization problem
	\begin{gather}
	\label{eq: min_problem_intro}
	I \coloneqq \inf \enstqbis{\mathcal{E}(u)}{u \in H^1(\R^2) \et \int_{\R^2} \abs{u}^2 = 1 } \, ,
	\end{gather}
	with the energy functional
	\begin{align*}
	\mathcal{E}(u) \coloneqq \int_{\R^2} \abs{\nabla u}^2+ \int_{\R^2}\left(-\frac{1}{\abs{\mathbf{x}}} + V^\mathrm{pp}(\mathbf{x})\right) \abs{u(\mathbf{x})}^2\diff \mathbf{x}+  \frac{1}{2} \iint_{\R^2 \times \R^2} \frac{\abs{u(\mathbf{x})}^2 \abs{u(\mathbf{y})}^2}{\abs{\mathbf{x-y}}} \diff \mathbf{x} \diff \mathbf{y} \pt
	\end{align*}
	Since $V^\mathrm{pp}\in L^p(\R^2)$ and is compactly supported, one can show that~\eqref{eq: min_problem_intro} admits a minimizer $v\in H^1(\R^2)$ by 
	following the arguments in~\cite[Section VII]{lieb1981thomasfermi}. In addition, we assume that $V^\mathrm{pp}$ satisfies a \emph{ionization condition }in the sense that the mean-field operator
	\begin{align*}
	\boxed{H^\mathrm{MF} \coloneqq -\Delta + V^\mathrm{MF} \ou V^\mathrm{MF} \coloneqq -\frac{1}{|\cdot|} + V^\mathrm{pp} + \abs{v}^2 \ast \frac{1}{|\cdot|} \, ,}
	\end{align*}
	admits at least one negative eigenvalue $-\mu<0$. In that case, it is non-degenerate, the minimizer $v$ is unique and it is the eigenfunction associated with $-\mu$. Also, up to a phase factor, $v$ is positive everywhere. 
	
	The following theorem states that the mean-field potential $V_L^\mathrm{MF}$ of the rHF model~\eqref{eq: min_problem_intro} belongs to the class described in Section~\ref{sec:a-class-of-periodic-operators-on-lscrl}.
	\begin{theo}
		\label{theo: rHF}
		If the lattice $\lscr^\mathbf{R}$ satisfies Assumption~\ref{hypo_1} then Assumptions~\ref{hypo_0}--\ref{hypo_4} are valid with $V^\mathrm{MF}$ as reference potential. In particular, the $N$ lowest Bloch eigenvalues of $H_L^\mathrm{MF}$ satisfy \eqref{eq: main_theorem}.
	\end{theo}
	To our knowledge, nonlinear periodic models as \eqref{eq: rHF_model} in the dissociation or the semiclassical regime have not been addressed very much in the literature. Albanese shows in~\cite{albanese1988localised} that the time dependent Hartree equation with a periodic potential consisting of a periodic array of deep wells admits a solution where all single orbital is exponentially decaying if the distance separating the wells is large enough. In a non periodic setting, the Hartree-Fock model (with the exchange term) in the dissociation regime was studied by Daumer in~\cite{daumer1994hartree}. By fixed-point methods and under assumptions ensuring that spectral tunneling can be neglected, the author constructed solutions to the Hartree-Fock equations which are also minimizers for the Hartree-Fock energy. Recently, in a series of papers~\cite{rougerie2018interacting, olgiati2020hartree, olgiati2021bosons}, Olgiati, Rougerie and Spehner consider bosonic systems trapped in a symmetric double-well potential in the limit where the distance between the wells increases to infinity and the potential barrier is high.
	\begin{rem}[Assumptions on $V^\mathrm{pp}$]
		\label{rem:assumptions_pseudo-potential}
		\begin{enumerate}[noitemsep, label=(\roman{*})]
			\item There exists pseudo-potentials $V^\mathrm{pp}$ which satisfy the ionization condition. For instance, one can take $V^\mathrm{pp} = - \eta V$ with $\eta \geq 0$ large enough and $V\in L^\infty(\R^2)$ non-negative, non-zero and radial (see~\cite[Chapter 3  -- Appendix A]{cazalisthesis}).
			\item To ensure that $V_L^\mathrm{MF}$ has the same symmetries as $\lscr^\mathbf{R}$, it is sufficient for $V^\mathrm{pp}$ to be invariant with respect to the point group of $G$. But requiring rotation invariance makes the proofs simpler. Also, the technical assumption that $V^\mathrm{pp}$ must have compact support could be replaced by an appropriate decay at infinity assumption without modifying most of the arguments in the proof of Theorem~\ref{theo: rHF}. We think that assuming $V^\mathrm{pp}(\mathbf{x}) = \grandO{\abs{\mathbf{x}}^{-2-\epsilon}}$ for some $\epsilon>0$ should do.
			\item The assumption that~\eqref{eq: min_problem_intro} admits a minimizer and that the mean-field operator has a negative eigenvalue is crucial for our analysis. Otherwise, Assumption~\ref{hypo_3} would not be satisfied.
			\item In~\cite[Appendix A]{cazalisthesis}, a numerical analysis suggests that there is a negative eigenvalue for $V^\mathrm{pp}=0$. This fact is known in $3D$~\cite{lieb1977hartreefock, lieb1981thomasfermi, benguria1981thomasfermi} but we have not found it stated anywhere in $2D$. 
		\end{enumerate}
	\end{rem}	
	The following corollary is a direct consequence of Theorem~\ref{theo: existence_dirac_cones} and Theorem~\ref{theo: rHF}.
	\begin{cor}[Dirac points in rHF for the honeycomb lattice]
		\label{cor:dirac_points_rHF}
		Assume that $\lscr^\mathbf{R}=\lscr^H$ is the honeycomb lattice and that $q=2$. We denote by $\theta_L$ the interaction coefficient and we assume there exists $c >0$ and $\delta' \in \intfo{0}{\delta}$ such that $\abs{\theta_L} \geq c e^{-(1+\delta')\sqrt{\mu}d_0L}$. Then, the conclusions of Theorem~\ref{theo: existence_dirac_cones} hold: the dispersion relation of $H^\mathrm{MF}_L$ admits Dirac points. In addition, the Fermi level $\epsilon_L$ is exactly equal to the energy level of the cones.
	\end{cor}
		The coincidence of the Fermi level with the cones energy physically means that the small excitations of the Fermi sea behave as Dirac fermions: the dissociation regime describes a Dirac semi-metal.
		In Appendix~\ref{sec:the-weak-constrast-regime}, we present the results
		from~\cite[Chapter 4]{cazalisthesis} where we study the
		regime $L\to 0$ which corresponds to a perturbative regime of the free Laplacian $-\Delta$. There, we show that the operator $H_L^\mathrm{MF}$ admits Dirac points at
		the vertices of the Brillouin zone but that the Fermi level $\epsilon_L$
		\emph{does not} coincide with the energy level of the cones and lies within a Bloch band: this regime describes a metal. Consequently,
		there must be at least one phase transition as $L$ goes from 0 to
		infinity. To numerically conjecture in which phase the physical value
		$L_\mathrm{phy} \simeq 5.36$ should lie, we perform numerical
		simulations. The computations are done without the
		pseudo-potential term, that is $V^\mathrm{pp}=0$. 
		
		We use the DFTK software~\cite{herbst2021dftk}, a recently-developed
		package for the simulation of solids within density functional theory
		and related models. The numerical method proceeds by discretizing the
		Brillouin zone using equispaced points, and expanding the Bloch states
		in a truncated Fourier basis. The nuclear and Hartree potentials are
		computed in Fourier space. Compared to the usual setting of 3D density
		functional theory, the space dimension is forced to 2, and the formulas
		for the electrostatic (both nuclear and Hartree) potentials are modified
		to represent the kernel~\eqref{eq:periodic_kernel}.
		The resulting algebraic equations are solved using an nonlinear solver,
		employing iterative methods and fast Fourier transforms for efficiency.
		
		In the numerical results below, we use a truncation to plane waves
		having a maximal kinetic energy of 2000, and a $30\times30\times30$ grid of the
		Brillouin zone.	The convergence of both Brillouin zone sampling and the nonlinear solver
		was helped by adding an artificial Gaussian smearing of 0.001. The
		results, kindly provided by Antoine Levitt, are presented in
		Figure~\ref{fig: phase_transition}. The code is available as an example in the public release
		of DFTK.
				
		For any value of $L$, the mean-field operator displays conical singularities at the vertices $\mathbf{K}_\star \in\{\mathbf{K,K'}\}$. When $L=0$, the (renormalized) dispersion relation is given by the free Laplacian $-\Delta$. In particular, the first Bloch eigenvalue along the quasi-momenta section $\intof{\mathrm{K}}{\mathrm{M}}$ (resp. with quasi-momenta $\mathbf{K}_\star$) is doubly degenerate (resp. triply degenerate). For $L\leq 1.4$, the first two Bloch bands overlap and the Fermi level is lower than the cones energy; we are in the metal phase. For $L\geq 2$, the bands only overlap at the Dirac points and the Fermi level must coincide with the cones energy; we have attained the Dirac semi-metal phase. These results suggest that there is at least one phase transition for some $L_\mathrm{crit} \in\intoo{1.4}{2}$ and that $L_\mathrm{phy}$ is in the Dirac semi-metal phase.
		\begin{figure}[h]
			\centering
			\begin{subfigure}{0.32\textwidth}
				\centering
				\includegraphics[width=\textwidth]{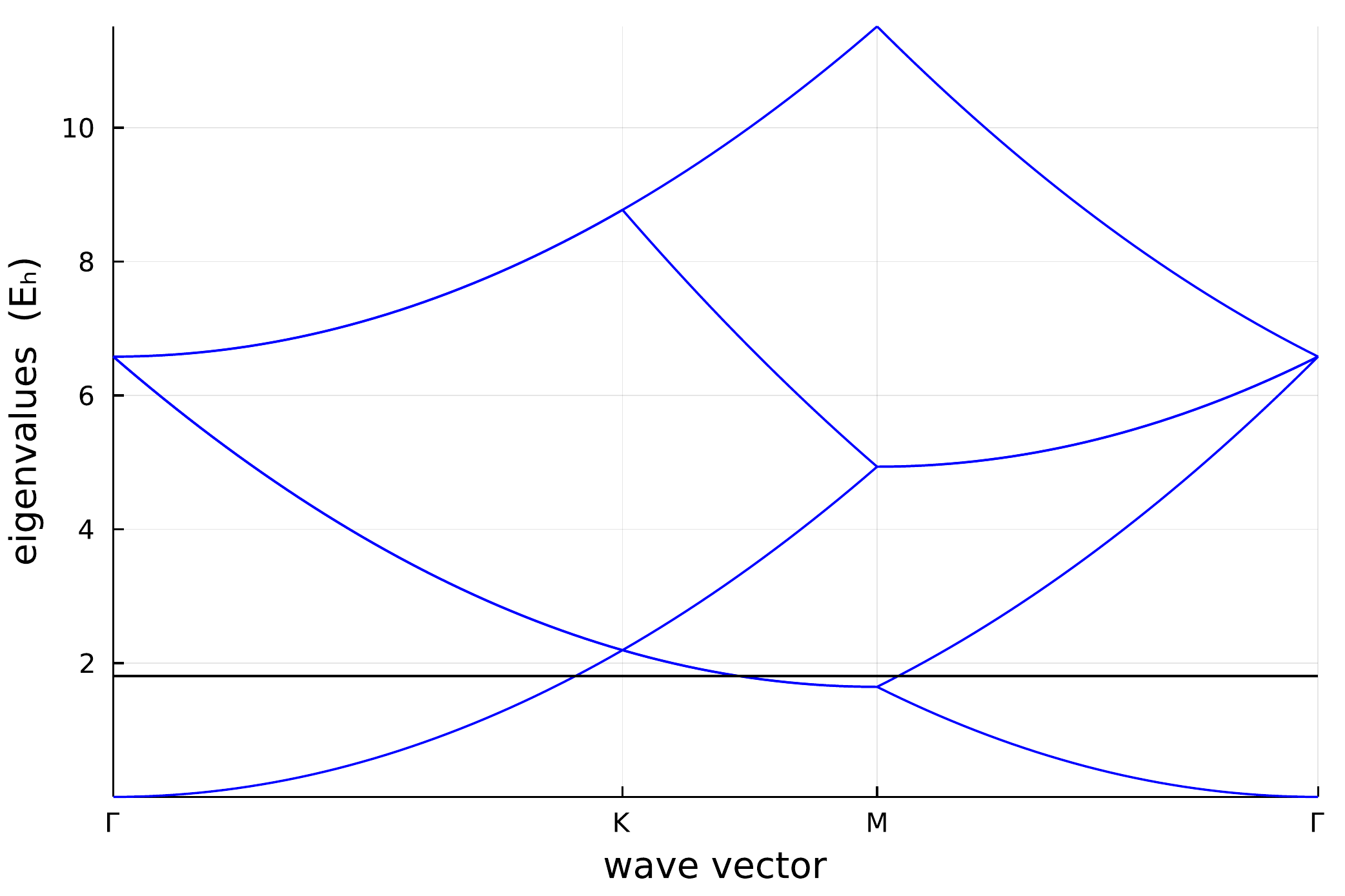}
				\caption{$L=0$}
				\label{fig:L=0}
			\end{subfigure}
			\hfill
			\begin{subfigure}{0.32\textwidth}
				\centering
				\includegraphics[width=\textwidth]{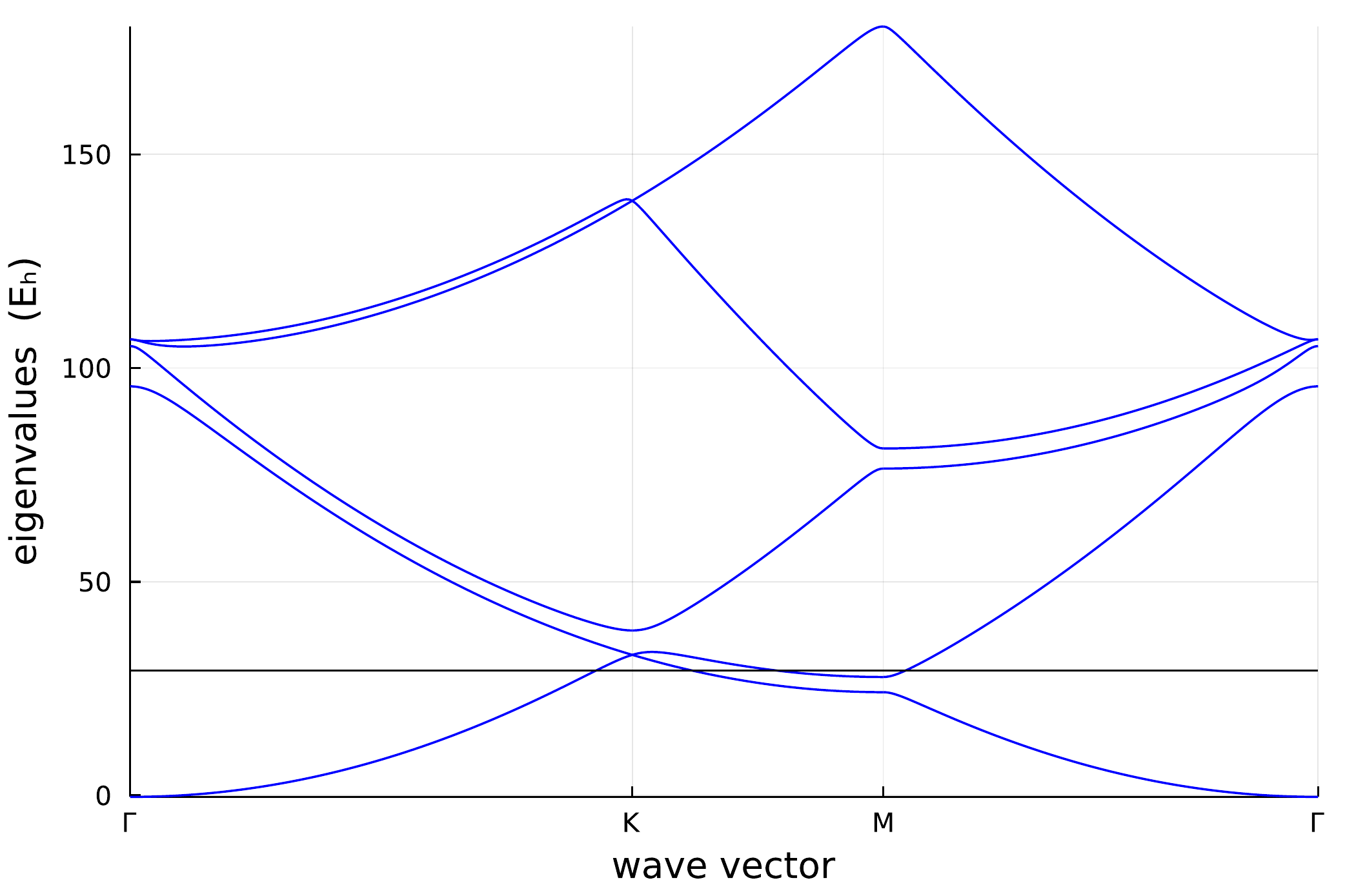}
				\caption{$L=0.5$}
				\label{fig:L=0.5}
			\end{subfigure}
			\hfill
			\begin{subfigure}{0.32\textwidth}
				\centering
				\includegraphics[width=\textwidth]{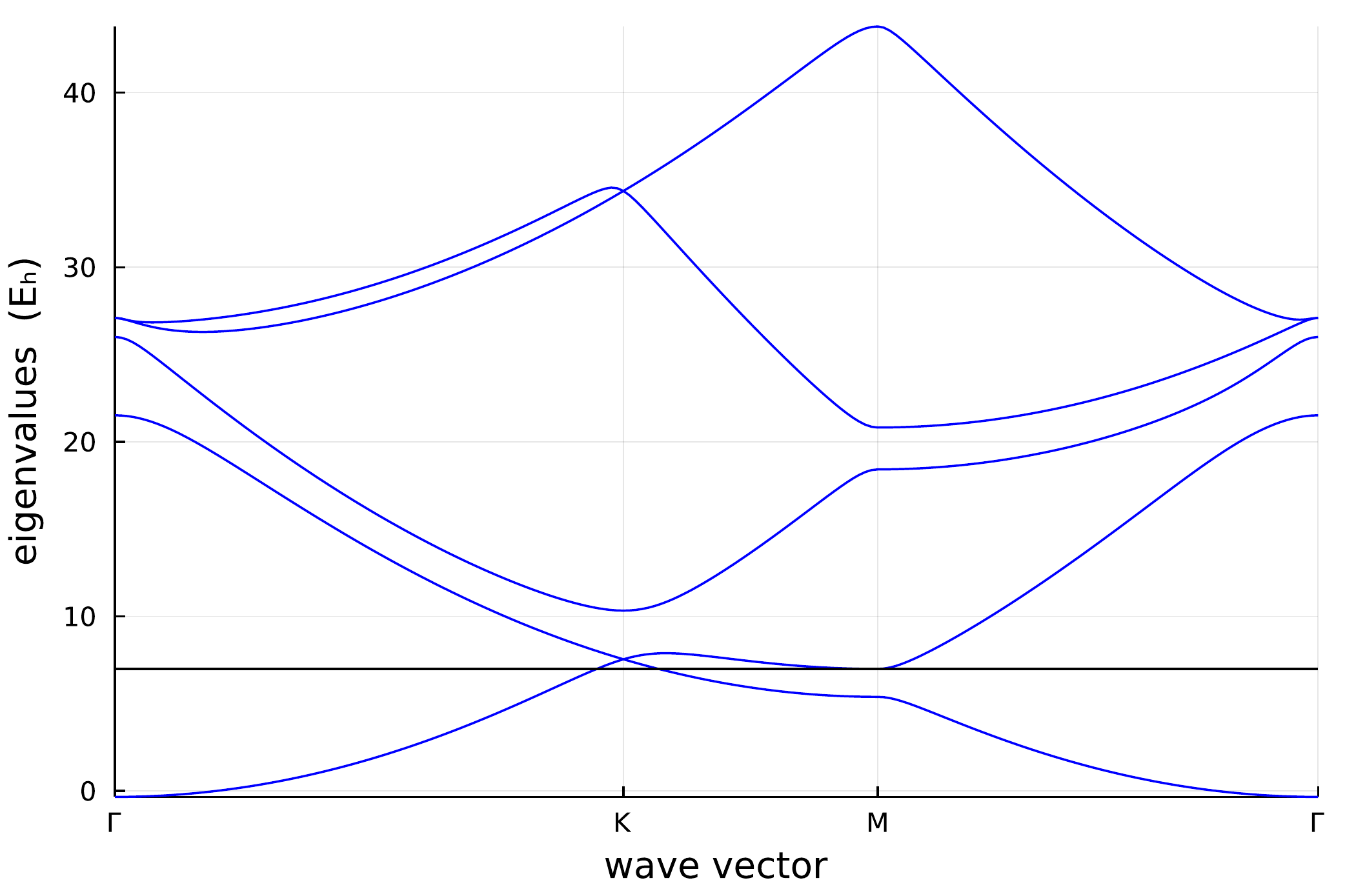}
				\caption{$L=1$}
			\end{subfigure}

			\begin{subfigure}{0.32\textwidth}
				\centering
				\includegraphics[width=\textwidth]{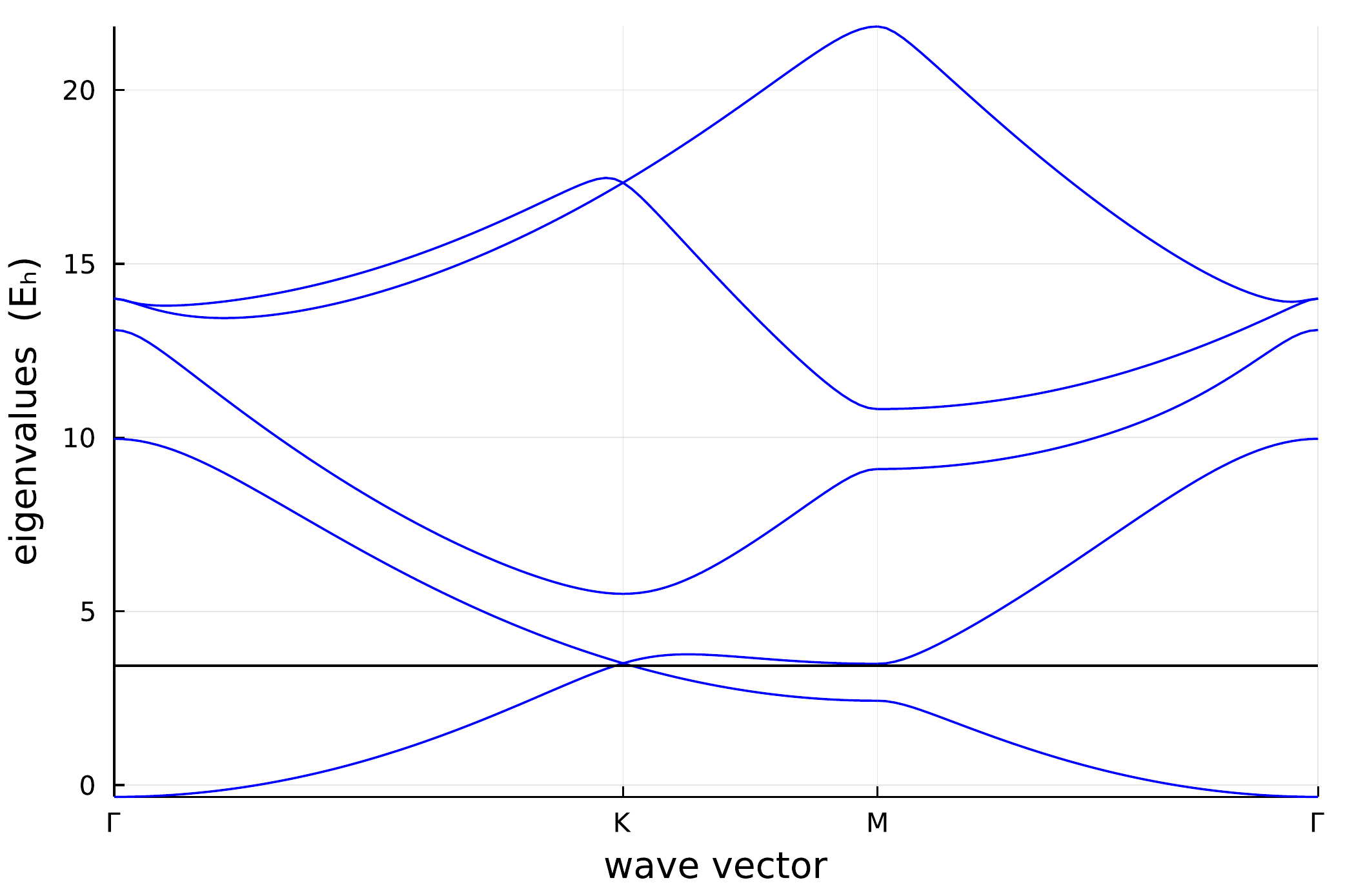}
				\caption{$L=1.4$}
			\end{subfigure}
			\hfill
			\begin{subfigure}{0.32\textwidth}
				\centering
				\includegraphics[width=\textwidth]{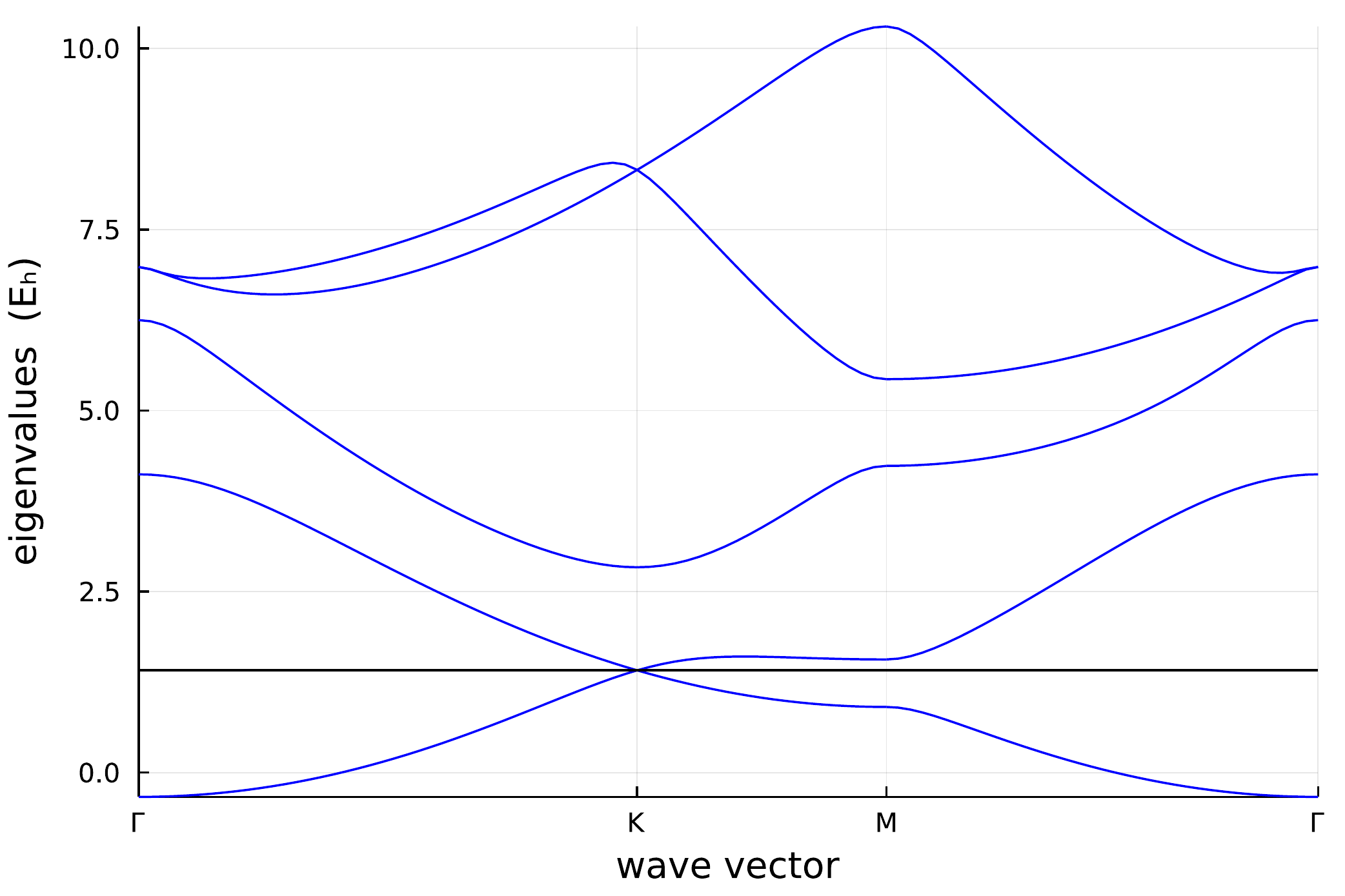}
				\caption{$L=2$}
			\end{subfigure}
			\hfill
			\begin{subfigure}{0.32\textwidth}
				\centering
				\includegraphics[width=\textwidth]{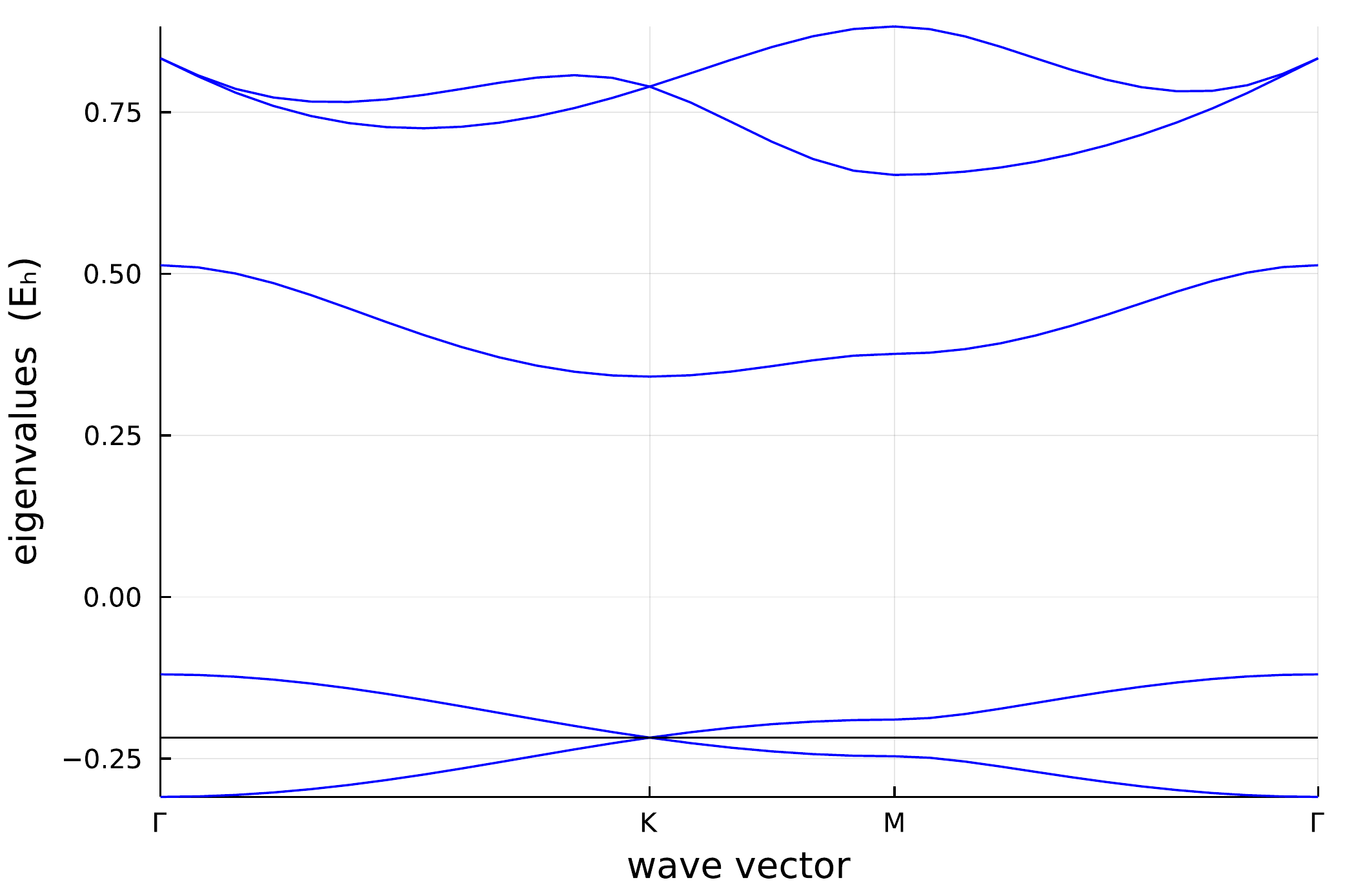}
				\caption{$L=6$}
			\end{subfigure}
			\caption{Numerical estimations by DFTK software of the fifth first Bloch eigenvalues of $H_L^\mathrm{MF}$ along the section $(\Gamma, \mathrm{K},\mathrm{M},\Gamma)$ joining the special symmetry points, see Figure~\ref{fig:triangular_lattice_BZ}. The Fermi level for $q=2$ is shown in black. For $L=0$, Figure~\ref{fig:L=0} shows the dispersion relation of $-\Delta$ on the lattice of fixed size $2\cdot\lscr^H$.}
			\label{fig: phase_transition}
		\end{figure}		
	
	\section{Proof of Theorem~\ref{theo: feshbach-schur}}\label{sec:two-dimensional-lattices-at-dissociation-with-coulomb-singularities}
	
	In this section, we consider a Schrödinger operator $H_L = -\Delta +V_L$ which commutes with the shifts of the scaled Bravais lattice $\mathscr{L}_L = L\lscr$ introduced in Section~\ref{sec:bravais-lattices-and-crystals}. We assume that the potential $V_L$ satisfies Assumptions~\ref{hypo_0}--\ref{hypo_4}. Under these assumptions, we employ the Feshbach-Schur method~\cite{bach1998renormalization, gustafson2020mathematical} in order to give the first two orders of the low-lying dispersion surfaces of $H_L$, in the regime where $L$ is large. This strategy has already been used in a similar context in~\cite{outassourt1984, fefferman2017honeycomb}, for instance. 
	
	\subsection{Strategy of proof}\label{sec:strategy-of-proof}
	
	The proof of Theorem~\ref{theo: feshbach-schur} crucially uses the Feshbach-Schur method . Let us briefly recall it. We consider $P$ and $P^\perp$, two orthogonal projections on a Hilbert space $\mathcal{H}$, such that $P+P^\perp=1$. Let $H$ be a self-adjoint operator on $\mathcal{H}$. It can be represented as the block matrix
	\begin{align*}
	H = \begin{pmatrix}
	A & C^* \\
	C & B
	\end{pmatrix} \, ,
	\end{align*}
	with $A=PAP$, $B= P^\perp HP^\perp$ and $C = P^\perp H P$.
	Let $E\in\sigma_\mathrm{d}(H)$ and assume that there exists $\epsilon>0$ such that
	\begin{align}
	\label{eq:energy_inequality_strategy}
	B - E \geq \epsilon P^\perp \pt
	\end{align}
	Then the eigenvalue problem $H\psi=E\psi$ is equivalent to
	\begin{align*}
	\left(A - C^*(B-E)^{-1} C - E\right) P\psi = 0 \pt
	\end{align*}
	The other component of $\psi$ is recovered thanks to the relation $P^\perp \psi = - (B- E)^{-1} C P\psi$. If the operator $C^*(B-E)^{-1} C$ is bounded then perturbation theory~\cite{kato1995perturbation} implies that the distance between $E$ and the spectrum of $A$ is estimated by
	\begin{align}
	\label{eq:estimate_spectrum}
	\dist(E, \sigma(A)) \leq \normeL{C^*(B-E)^{-1} C} \leq \normeL{(B-E)^{-1/2} C}^2 \pt
	\end{align}
	We see that, in order to correctly estimate $E$, the choice of the orthogonal projection $P$ should allow to estimate both $\sigma(A)$ and $\normeL{(B-E)^{-1/2} C}$.
	
	In the proof of Theorem~\ref{theo: feshbach-schur}, we apply this method to the Bloch operator $H=H_L(\mathbf{k}) \coloneqq -\Delta+V_L$ which acts on $\lk^2(\Gamma_L)$. We choose the projection $P=P_L(\mathbf{k})$ following ideas in~\cite{outassourt1984, daumer1993periodique, daumer1996schrodinger, fefferman2017honeycomb}. The projection $P_L(\mathbf{k})$ is chosen as the orthogonal projection on the subspace spanned by the Bloch-Floquet transforms of the functions $v_{L,\mathbf{r}}$, defined in \eqref{eq: eigen_equation_Hlr} as the first eigenfunctions of the operators $H_{L,\mathbf{r}}$. To construct this projector, we first show, in Section~\ref{sec:orthomalization-procedure}, that the familly $\{v_{L,\mathbf{r}}\}_{\mathbf{r\in}\lscr^\mathbf{R}}$ is almost orthonormal and from it, we form an orthonormal family $\{w_{L,\mathbf{r}}\}_{\mathbf{r\in}\lscr^\mathbf{R}}$ by applying the Gram-Schmidt process. This family shares many properties of (composite) Wannier functions: periodicity and localization (see Proposition~\ref{prop: properties_wannier}). When applying the Bloch-Floquet transform to each function $w_{L,\mathbf{r}}$, we obtain an orthonormal system of $N=\abs{\mathbf{R}}$ quasi-periodic functions which defines the orthogonal projection $P_L(\mathbf{k})$. The spectrum of $A=P_L(\mathbf{k})H_L(\mathbf{k})P_L(\mathbf{k})$ is computed in Corollary~\ref{cor: expansion_A_k}. The energy inequality \eqref{eq:energy_inequality_strategy} is shown in Proposition~\ref{prop: energy_estimate}. Finally, we estimate the right side of \eqref{eq:estimate_spectrum} in Proposition~\ref{prop: estimation_key}.
	
	\subsection{Notation}\label{sec:notations-and-conventions}
	
	We describe some notations we use in the sequel. For two quantities $A$ and $B$, we write $A \lesssim B$ whenever there exists a constant $C$ independent from any relevant parameters and such that $A \leq C B$.
	If $A_L$ and $B_L$ are sequences labeled by $L$ and $\alpha >0$ then the notation $A_L = \grandO{B^{\alpha-}_L}$ means that for every $\epsilon>0$ there exists a constant $C_\epsilon >0$ such that for $L$ large enough (depending on $\epsilon$) we have $\abs{A_L} \leq C_\epsilon B_L^{\alpha - \epsilon}$. We use the notation $\grandO{L^{-\infty}}$ to denote a $\grandO{L^{-k}}$ for all $k \in \N$, where the $O$ may depend on $k$.
	
	The open ball centered on $\mathbf{r} \in \R^2$ with radius $R\geq 0$ is denoted by $B(\mathbf{r},R)$. The distance between two closed sets $A,B \subset \R^d$ is denoted $\dist(A,B)$ or $\dist(a,B)$ (resp. $\dist(A,b)$) when $A=\{a\}$ (resp. $B = \{b\}$) is reduced to a singleton.
	
	\subsection{A regularity result}\label{sec:a-regularity-result}
	
	Let $V\in L^p(\R^2)+L^\infty(\R^2)$ with $p\in\intoo{1}{2}$. We consider the Friedrichs extension of $H=-\Delta+V$ which is self-adjoint on $\mathcal{D}(H) = \enstq{u \in H^1(\R^2)}{(-\Delta+V)u \in L^2(\R^2)}$, see~\cite{reed1975methodsII, cycon1987schrodinger}. The following proposition shows that $\mathcal{D}(H) \subset H^p(\R^2)$.
	\begin{prop}
		\label{prop: regularity}
		Let $V  \in L^p(\R^2) + L^\infty(\R^2)$ where $p\in\intoo{1}{\infty}$. Let $\alpha=\min(\frac{p}{2},1)$. For all $\nu >0$ large enough, we have
		\begin{align}
		\label{prop: regularity_estimate1}
		\normeL{(-\Delta+\nu)^{\alpha}(-\Delta+V+\nu)^{-1}(-\Delta+\nu)^{1-\alpha}} \leq 2 \pt
		\end{align}	
		Then there exists $C >0$ such that for all $u\in L^2(\R^2)$ satisfying $(-\Delta + V)u \in L^2(\R^2)$ we have
		\begin{align*}
		\norme{u}_{H^p(\R^2)} \leq 2 \normeL{(-\Delta + V)u}_{L^2(\R^2)} + C\norme{u}_{L^2(\R^2)} \, ,
		\end{align*}
		where we can choose
		\begin{align}
		\label{eq:regularity_constant}
		C \lesssim \max \left[1, \norme{V}^{\frac{p}{p-1}}_{L^p(\R^2)+L^\infty(\R^2)}\right]\pt
		\end{align}
	\end{prop}
	\begin{lemma}
		\label{lemma: regularity_aux}
		Let $d\geq 1$ and $p \geq 1$ such that $p > \frac{d}{2}$. Let $\alpha \in \intff{1 - \frac{p}{2}}{\frac{p}{2}} \cap\intff{0}{1}$. There exists a constant $C(d,p) $ such that for all $V \in L^p(\R^d)$ and $\nu > 0$ we have
		\begin{align*}
		\normeL{(-\Delta + \nu)^{-1+\alpha} V(-\Delta+ \nu)^{-\alpha}} \leq C(d,p)\nu^{\frac{d}{2p}-1} \normeL{V}_{L^p(\R^d)} \pt
		\end{align*}
	\end{lemma}
	\begin{proof}
		First, we assume $\alpha \in \intoo{0}{1}$. Notice that the conditions on $p$ and $\alpha$ ensure that $\frac{p}{1-\alpha}\geq 2$ and $\frac{p}{\alpha} \geq 2$.
		Using Hölder's inequality for trace ideals~\cite[Theorem 2.8]{simon2005trace} and the Kato-Seiler-Simon inequality~\cite[Theorem 4.1]{simon2005trace}, we have
		\begin{align*}
		\normeL{(-\Delta + \nu)^{-1+\alpha} V(-\Delta+ \nu)^{-\alpha}}
		&\leq \normeL{(-\Delta + \nu)^{-1+\alpha} \abs{V}^{1-\alpha}}_{\mathfrak{S}_{p/(1-\alpha)}}\normeL{ \abs{V}^{\alpha}(-\Delta + \nu)^{-\alpha}}_{\mathfrak{S}_{p/\alpha}} \\
		&\leq C(d,p) \nu^{\frac{d}{2p}-1}\normeL{V}_{L^p(\R^d)} \, ,
		\end{align*}
		where the constant $C(d,p)$ reads
		\begin{align*}
			C(d,p) = (2\pi)^{-d/p} \left(\int_{\R^d} \frac{\diff \mathbf{k}}{(\abs{\mathbf{k}}^2+1)^{p}}\right)^{1/p} = \left[\frac{B\left(\frac{d}{2},p-\frac{d}{2}\right)}{2^d \pi^{\frac{d+1}{2}} \Gamma(\frac{d-1}{2})}\right]^{1/p} \pt
		\end{align*}
		We have denoted by $B(x,y)$ and $\Gamma(x)$ respectively the Euler beta function and the Euler Gamma function.
		If $\alpha \in \{0,1\}$ then we must have $p\geq 2$ due to our assumptions. Then the proof is the same except we do not use Hölder's inequality.
	\end{proof}
	\begin{proof}[Proof of Proposition~\ref{prop: regularity}]
		For $V = V_1 + V_2 \in L^p(\R^2) + L^\infty(\R^2)$ with $p\in\intoo{1}{\infty}$, we write
		\begin{align*}
		-\Delta + V + \nu = (-\Delta + \nu)^{1/2} \left(1 + (-\Delta+\nu)^{-1/2}V(-\Delta+\nu)^{-1/2}\right)(-\Delta + \nu)^{1/2} \pt
		\end{align*}
		Using Lemma~\ref{lemma: regularity_aux} with $\alpha=1/2$, for all $\nu \geq \max \left[\left(4 C(2,p) \normeL{V_1}_{L^p(\R^2)}\right)^\frac{p}{p - 1}, 4\normeL{V_2}_{L^\infty(\R^2)}\right]$ we have
		\begin{align*}
		\normeL{(-\Delta+\nu)^{-1/2}V(-\Delta+\nu)^{-1/2}} \leq 1/2 \pt
		\end{align*}
		Hence, the operator $-\Delta + V + \nu$ is invertible and its inverse is given by a Neumann series
		\begin{align*}
		(-\Delta + V + \nu)^{-1} = (-\Delta + \nu)^{-1/2} \sum_{n\geq 0} (-1)^n\left[(-\Delta+\nu)^{-1/2}V(-\Delta+\nu)^{-1/2}\right]^n (-\Delta + \nu)^{-1/2} \pt
		\end{align*}
		We multiply by $(-\Delta+\nu)^{\alpha}$ for $\alpha = \min(\frac{p}{2},1) \in \intof{1/2}{1}$ and we obtain
		\begin{align*}
		(-\Delta+\nu)^{\alpha}(-\Delta + V + \nu)^{-1}
		= \sum_{n\geq 0} (-1)^n \left[(-\Delta+\nu)^{-1+\alpha}V (-\Delta+\nu)^{-\alpha}\right]^n (-\Delta+\nu)^{-1+\alpha} \pt
		\end{align*}
		The term of order $n$ is bounded by Lemma~\ref{lemma: regularity_aux}
		\begin{align*}
		\normeL{\left[(-\Delta+\nu)^{-1+\alpha}V (-\Delta+\nu)^{-\alpha}\right]^n} \leq \left(C(2,p)\nu^{\frac{1}{p}-1} \normeL{V_1}_{L^p(\R^2)} + \nu^{-1}\normeL{V_2}_{L^\infty(\R^2)}\right)^n \leq 2^{-n} \pt
		\end{align*}
		Hence, we have obtained estimate \eqref{prop: regularity_estimate1}. This also shows that $\mathcal{D}(-\Delta+V) \subset \mathcal{D}((-\Delta+\nu)^\alpha) = H^{\min(p,2)}(\R^2)$ with the inequality
		\begin{align*}
		\normeL{(-\Delta+\nu)^\alpha u}_{L^2(\R^2)} \leq 2\nu^{-1+\alpha} \normeL{(-\Delta+V+\nu)u}_{L^2(\R^2)} \pt
		\end{align*}
		Thus, we find, for all $u\in\mathcal{D}(-\Delta+V)$ and $\nu\geq 1$,
		\begin{align*}
		\normeL{(-\Delta+1)^\alpha u}_{L^2(\R^2)}
		\leq 2 \normeL{(-\Delta+V)u}_{L^2(\R^2)} + C(\norme{V_1}_{L^p(\R^2)},\norme{V_2}_{L^\infty(\R^2)}) \normeL{u}_{L^2(\R^2)} \, ,
		\end{align*}
		with
		\begin{align*}
		C(\normeL{V_1}_{L^p(\R^2)},\norme{V_2}_{L^\infty(\R^2)}) = 2\cdot \max\left[1, (4C(2,p)\norme{V_1}_{L^p(\R^2)})^\frac{p\alpha}{p- 1}, 4\norme{V_2}^\alpha_{L^\infty(\R^2)} \right] \pt
		\end{align*}
		Estimate~\eqref{eq:regularity_constant} follows since $\alpha\leq 1$ and $\frac{p}{p-1}\geq1$. This concludes the proof of Proposition~\ref{prop: regularity}.
	\end{proof}
	
	\subsection{Kato's inequalities}
	
	In this section, we state Kato type estimates on the periodic potential $V_L$ in both $L^2(\R^2)$ and $\lk^2(\Gamma_L)$ spaces. In addition, we show that the constants appearing in the upper bounds grow polynomially with $L$.
	
	Before stating our proposition, we notice that Assumption~\ref{hypo_04} and Assumption~\ref{hypo_4} implies that the $\norme{V_L}_{\lper^p(\Gamma_L)}$ grows at most polynomially with $L$:
	\begin{align}
	\label{rem: grow_polynomial}
	\exists M \in \R,\quad \norme{V_L}_{\lper^p(\Gamma_L)} = \grandO{L^M}
	\end{align}
	
	\begin{prop}[Kato's inequalities]
		\label{lemma: kato_type_inequality}
		Let $L\geq 1$. For all $\epsilon >0$ there exists $C(\epsilon,L) >0$ such that for all $u \in H^1(\R^2)$, we have
		\begin{align}
		\label{eq: kato_type_inequality1}
		\pdtsc{u}{\abs{V_L} u}_{L^2(\R^2)} \leq \epsilon \norme{\nabla u}_{L^2(\R^2)}^2 + C(\epsilon,L) \norme{u}^2_{L^2(\R^2)} \, ,
		\end{align}
		and such that for all $\mathbf{k} \in \Gamma^*_L$ and for all $u \in H^1_\mathbf{k}(\Gamma_L)$, we have
		\begin{align}
		\label{eq: kato_type_inequality2}
		\pdtsc{u}{\abs{V_L} u}_{L^2_\mathbf{k}(\Gamma_L)} \leq \epsilon \norme{\nabla u}_{L^2_\mathbf{k}(\Gamma_L)}^2 + C(\epsilon,L) \norme{u}^2_{L^2_\mathbf{k}(\Gamma_L)} \pt
		\end{align}
		For any $r\in\intoo{p}{\infty}$, we can take $C(\epsilon,L) = \grandO{\epsilon^{-\frac{2r}{p}}L^{M\frac{p+r}{p}}}$ where $M$ is the constant appearing in~\eqref{rem: grow_polynomial}.
	\end{prop}
	\begin{proof}
		We start with the proof of \eqref{eq: kato_type_inequality1}. Let $A>0$ and $u \in H^1(\R^2)$. Let $q \in \intoo{p}{\infty}$ and $q' \in \intoo{1}{\frac{p}{p-1}}$ such that $1/q+1/q'=1$. Using the periodicity of $V_L$ and Hölder's inequality, we write
		\begin{align}
		\label{eq: kato_inequality_11}
		\pdtsc{u}{\abs{V_L} u}_{L^2(\R^2)} \leq \norme{V_L \mathds{1}_{\abs{V_L} \geq A}}_{\lper^q(\Gamma_L)} \left(\sum_{\mathbf{u} \in \lscr} \norme{u}^2_{L^{2q'} (\Gamma_L+L\mathbf{u})}\right) + A \norme{u}^2_{L^2(\R^2)} \pt
		\end{align}
		Adapting the proof of~\cite[Theorem XIII.96]{reed1978methodsIV}, one can show the estimate
		\begin{align}
		\label{eq: kato_inequality_12}
		\sum_{\mathbf{u} \in \lscr} \norme{u}^2_{L^{2q'} (\Gamma_L+L\mathbf{u})} \lesssim \norme{u}^2_{H^1(\R^2)} \pt
		\end{align}
		We insert \eqref{eq: kato_inequality_12} into \eqref{eq: kato_inequality_11} then we can choose $A$ large enough such that \eqref{eq: kato_type_inequality1} holds.
		To show \eqref{eq: kato_type_inequality2}, we write for $A>0$ and $u \in H^1_\mathbf{k}(\Gamma_L)$
		\begin{align*}
		\pdtsc{u}{\abs{V_L} u}_{L^2_\mathbf{k}(\Gamma_L)} \leq \norme{V_L \mathds{1}_{\abs{V} \geq A}}_{\lper^q(\Gamma_L)} \norme{u}^2_{\lk^{2q'}(\Gamma_L)} + A \norme{u}^2_{\lk^2(\Gamma_L)} \pt
		\end{align*}
		We conclude using the Sobolev embedding $H^1_\mathbf{k}(\Gamma_L) \subset \lk^{2q'}(\Gamma_L)$ (where the constant of continuity does not depend on $L$ by~\cite[Theorem 2.28]{aubin1998somenonlinear}) then taking $A$ large enough. For the last statement, we consider $r>p$ such that $\frac{1}{q} = \frac{1}{2}\left(\frac{1}{p} + \frac{1}{r}\right)$ and use Hölder's inequality and Tchebychev's inequality
		\begin{align*}
		\norme{V_L \mathds{1}_{\abs{V_L} \geq A}}_{\lper^q(\Gamma_L)} 
		\leq \sqrt{\norme{V_L}_{\lper^p(\Gamma_L)} \norme{\mathds{1}_{\abs{V_L} \geq A}}_{\lper^r(\Gamma_L)}}
		\leq \norme{V_L}_{\lper^p(\Gamma_L)}^{\frac{r+p}{2r}} A^{-\frac{p}{2r}} \lesssim L^{M\frac{r+p}{2r}} A^{-\frac{p}{2r}}\, ,
		\end{align*}
		where $M\in\R$ is the constant appearing in~\eqref{rem: grow_polynomial}. This concludes the proof of Proposition~\ref{lemma: kato_type_inequality}.	
	\end{proof}
	\begin{rem}
		\label{rem: kato_type_inequality_consequence}
		In some cases, the constant $C(\epsilon,L)$ appearing in \eqref{eq: kato_type_inequality1} or \eqref{eq: kato_type_inequality2} does not depend on $L$, for instance, if there exists $C >0$ such that for any $\mathbf{r} \in \mathbf{R}$ we have $\abs{V_L(\mathbf{x-r})} \leq C \abs{\mathbf{x-r}}^{-1}$ in some neighborhood of $\mathbf{r}$.
	\end{rem}
	As a consequence of Proposition~\ref{lemma: kato_type_inequality}, the closure of the operator $(-\Delta+1)^{-1/2} V_L(-\Delta+1)^{-1/2}$ is bounded with norm at most polynomial in $L$.
	\begin{cor}
		\label{cor: kato_type_inequality_consequence}
		We have
		\begin{align}
		\label{eq: kato_type_inequality_consequence}
		\normeL{\frac{1}{\sqrt{-\Delta+1}} V_L \frac{1}{\sqrt{-\Delta +1}}}_{\mathcal{B}(L^2(\R^2))} = \grandO{ L^M} \text{ and }\normeL{\frac{1}{\sqrt{-\Delta+1}} V_L \frac{1}{\sqrt{-\Delta +1}}}_{\mathcal{B}(\lk^2(\Gamma_L))} = \grandO{L^M} \, ,
		\end{align}
		where the second $O$ is uniform in $\mathbf{k}$.
	\end{cor}
	\begin{proof}
		Indeed, we use the estimate $C(L,\epsilon) \lesssim \epsilon^{-\frac{2r}{p}}L^{M\frac{p+r}{p}}$ for some $r>p$ and we take $\epsilon = L^{M\left(1 - \frac{r}{p+2r}\right)}$ to obtain
		\begin{align*}
		\forall \psi \in L^2(\R^2),\quad \normeL{\sqrt{\abs{V_L}} \frac{1}{\sqrt{-\Delta+1}}\psi}^2_{L^2(\R^2)} \lesssim  L^{M\left(1 - \frac{r}{p+2r}\right)} \norme{\psi}^2_{L^2(\R^2)} \lesssim L^M \norme{\psi}^2_{L^2(\R^2)}\pt
		\end{align*}
		Consequently, the operator $A \coloneqq \sqrt{\abs{V_L}} \frac{1}{\sqrt{-\Delta+1}}$ is bounded on $L^2(\R^2)$. Let $B \coloneqq \frac{1}{\sqrt{-\Delta+1}}\sqrt{\abs{V_L}}$ defined on the domain $\mathcal{D}(\sqrt{\abs{V_L}}) = \enstq{\psi \in L^2(\R^2)}{\sqrt{\abs{V_L}} \psi \in L^2(\R^2)}$. The operator $B$ is densely defined and closed on $L^2(\R^2)$ (see~\cite[Example 3.8]{schmugden2012unbounded}). By~\cite[Proposition 1.7]{schmugden2012unbounded}, we have $B^* = A$. This implies that $B^{**} = \overline{B} \in \mathcal{B}(L^2(\R^2))$ and we deduce that the operator
		\[B \sign(V_L) A = \frac{1}{\sqrt{-\Delta +1}}V_L \frac{1}{\sqrt{-\Delta+1}}\, ,\]
		admits a bounded extension on $L^2(\R^2)$ with norm at most $\grandO{L^M}$. The proof of the right side of~\eqref{eq: kato_type_inequality_consequence} is the same.
	\end{proof}
	
	\subsection{Exponential bounds on \texorpdfstring{$v$}{v}}
	
	In the next proposition, we give pointwise and integral exponential bounds on the first eigenfunction $v$ of the self-adjoint operator $H=-\Delta+V$ where the potential $V$ is the reference potential appearing in Assumption~\ref{hypo_3}.
	
	Recall that we can choose the phase of $v$ such that $v>0$ by~\cite{goelden1977nondegeneracy}. By Proposition~\ref{prop: regularity}, we also have $v \in H^{\min(p,2)}(\R^2)$. In particular, the Sobolev embeddings imply that $v \in L^\infty(\R^2) \cap \mathcal{C}_0^{0,\min(p-1,1)}(\R^2)$ where $\mathcal{C}_0^{\ell,\theta}(\R^2)$ denotes the spaces of $\mathcal{C}^\ell(\R^2)$ functions which vanish at infinity as well as their first $\ell$ derivatives and such that the derivatives of order $\ell$ are Hölder continuous with exponent $\theta\in\intff{0}{1}$.
	\begin{prop}[Exponential bounds on $v$]
		\label{prop: exponential_decay_v}
		There exists a constant $C>0$ such that
		\begin{gather*}
		\forall \mathbf{x} \in \R^2,\quad\frac{1}{C} \frac{e^{-\sqrt{\mu} \abs{\mathbf{x}}}}{1+\sqrt{\abs{\mathbf{x}}}} \leq v(\mathbf{x}) \leq C\frac{e^{-\sqrt{\mu} \abs{\mathbf{x}}}}{1+\sqrt{\abs{\mathbf{x}}}} \, ,\\
		\forall R \geq 0,\quad \int_{\abs{\mathbf{x}} \geq R} \left(\lvert v(\mathbf{x}) \rvert^2+ \lvert \nabla v(\mathbf{x}) \rvert^2\right) \diff \mathbf{x} \leq Ce^{-2\sqrt{\mu}R} \pt
		\end{gather*}
	\end{prop}
	\begin{proof}
		The proofs of the pointwise bounds can be found in~\cite[Corollary 2.2 \& Remark 2.2]{hoffmann1985pointwise}.		
		The integral bound on $v$ is a direct consequence of the pointwise ones. To get the integral bound on $\nabla v$, we multiply $-\Delta v + \mu v = -V v$ by $\eta v$ where $\eta \equiv 1$ on $\abs{\mathbf{x}} \geq R$ and $\eta \equiv 0$ on $\abs{\mathbf{x}} \leq R - \delta$ for $\delta >0$ small enough and we integrate by parts.
	\end{proof}
	
	\subsection{Properties of the mono-atomic operators}\label{sec:properties-of-the-mono-atomic-operators}
	
	We recall that the effective potential $V_{L,\mathbf{r}}$ and its associated mono-atomic Schrödinger operator $H_{L,\mathbf{r}}=-\Delta+V_{L,\mathbf{r}}$ are defined in \eqref{eq: definition_Vlr} and~\eqref{eq: effective_mono_atomic_operator}. In this section, we give spectral properties of this operator. Because they are standard, most statements are given without proof, see~\cite{reed1975methodsII, reed1978methodsIV, cycon1987schrodinger}.
	
	The function $V_{L,\mathbf{r}}$ belongs to the Kato class and we can consider its Friedrichs extension which defines a self-adjoint operator on the domain
	\begin{align*}
	\mathcal{D}\left(H_{L,\mathbf{r}}\right) = \enstqbis{u \in H^1(\R^2)}{\left(-\Delta + V_{L,\mathbf{r}}\right)u \in L^2(\R^2)} \pt
	\end{align*}
	The essential spectrum of $H_{L,\mathbf{r}}$ is given by $\sigma_\mathrm{ess}(H_{L,\mathbf{r}}) = \intfo{0}{\infty}$. Since $H$ has a negative eigenvalue by  Assumption~\ref{hypo_4}, it follows from perturbation theory  (see Corollary~\ref{cor:perturbation_theory} and~\cite[Chapter XII]{reed1978methodsIV}) that the discrete spectrum of $H_{L,\mathbf{r}}$ is non-empty and that the lowest eigenvalue of $H_{L,\mathbf{r}}$ is non-degenerate for $L$ large enough.
	
	Now, we denote by $s_L$ the scaling operator by $L$ defined by $s_L \mathbf{x} = L\mathbf{x}$. Then we can write $\lscr_L^\mathbf{R} = s_L[\lscr^\mathbf{R}]$ and the group $G_L \coloneqq s_L G s_L^{-1} \subset E_2(\R)$. By Assumption~\ref{hypo_2} and the fact that $\chi$ is radial, we have 
	\begin{align*}
	\forall g \in G,\quad \forall \mathbf{r} \in \lscr^\mathbf{R},\quad (s_Lgs_L^{-1})\cdot V_{L,\mathbf{r}} = V_{L,g\cdot \mathbf{r}} \pt
	\end{align*}
	In particular, we notice that the action of $G_L$ on the set $\{V_{L,\mathbf{r}}\}_{\mathbf{r} \in \lscr^\mathbf{R}}$ only operates on the labels $\mathbf{r}\in\lscr^\mathbf{R}$ and, as a consequence of Assumption~\ref{hypo_1}, this action is also transitive. Therefore, the operators $\{H_{L,\mathbf{r}}\}_{\mathbf{r} \in \mathscr{L}^\mathbf{R}}$ are unitarily equivalent and share the same spectrum. We denote their common discrete spectrum by
	\begin{align*}
	\sigma_\mathrm{d}(H_{L,\mathbf{r}}) = \{ - \mu_{1,L} < - \mu_{2,L} \leq -\mu_{3,L} \leq \dots \leq 0 \} \subset \intof{-\infty}{0},\quad \forall \mathbf{r} \in \mathscr{L}^\mathbf{R} \pt
	\end{align*}
	To lighten the notations, we will write $\mu_L$ instead of $\mu_{1,L}$. We recall that we denote by $v_{L,\mathbf{r}}$ the normalized eigenfunction of $H_{L,\mathbf{r}}$ associated with $-\mu_L$. We can choose the overall phase of $v_{L,\mathbf{r}}$ such that $v_{L,\mathbf{r}}>0$ and we have $v_{L,\mathbf{r}} \in  H^{\min(p,2)}(\R^2) \cap \mathcal{C}_0^{0,\min(p-1,1)}(\R^2)$. 
	Moreover, by the non-degeneracy of $- \mu_L$, the group $G_L$ also acts transitively on the set $\{v_{L,\mathbf{r}}\}_{\mathbf{r} \in \mathscr{L}^\mathbf{R}}$ and, as for $V_{L,\mathbf{r}}$, operates only on the labels $\mathbf{r} \in \lscr^\mathbf{R}$
	\begin{align}
	\label{eq: action_v}
	\forall g \in G,\quad \forall \mathbf{r} \in \lscr^\mathbf{R},\quad (s_Lgs_L^{-1})\cdot v_{L,\mathbf{r}} =  v_{L,g \cdot\mathbf{r}} \pt
	\end{align}
	In particular, these relations imply that the interaction coefficient $\theta_{L,k}$ defined in \eqref{eq: interaction_coefficient} does not depend on the choice of the pair $\mathbf{p}_k \in\mathcal{O}_k$. 
	
	Assumption~\ref{hypo_4} implies that the first eigenvalue/eigenfunction pair of $H_{L,\mathbf{r}}$ converges, up to a translation by $L\mathbf{r}$, to the one of $H$ (see for instance Corollary~\ref{cor:perturbation_theory} from the Appendix~\ref{sec:perturbation-theory-for-singular-potentials}):
	\begin{align}
	\label{eq: closedness_modified_original_operator}
	\lim\limits_{L\to\infty}\mu_L = \mu  \et \forall \mathbf{r} \in \mathscr{L}^\mathbf{R},\quad \lim\limits_{L\to\infty}\norme{v_{L,\mathbf{r}} - v(\cdot - L\mathbf{r})}_{H^{\min(p,2)}(\R^2)} = 0 \pt
	\end{align}
	Notice that by the Sobolev embeddings, we also have convergence in $L^\infty(\R^2)$
	\begin{align}
	\label{eq:convergence_uniform}
	\lim\limits_{L\to\infty}\norme{v_{L,\mathbf{r}} - v(\cdot - L\mathbf{r})}_{L^\infty(\R^2)} = 0 \pt
	\end{align}
	Also, the right side of \eqref{eq: closedness_modified_original_operator} implies that $v$ is invariant under the action of the point group of $G_L$. Also, if we denote by $g_L$ the spectral gap of $H_{L,\mathbf{r}}$ above its lowest eigenvalue $-\mu_L$ then we have $g_L= g + \petito{1}$.
	
	\subsection{Uniform exponential bounds on \texorpdfstring{$v_{L,\mathbf{r}}$}{vLr}}
	
	In the next proposition, we give exponential bounds on $v_{L,\mathbf{r}}$. We use Assumption~\ref{hypo_4} to show they are independent from $L$.
	\begin{prop}[Uniform exponential bounds on $v_{L,\mathbf{r}}$]
		\label{prop: exponential_decay_v_tilde}
		Let $\epsilon \in\intoo{0}{1}$. There exists $L_\epsilon\geq 1$ and $C_\epsilon >0$ such that for all $L \geq L_\epsilon$ we have
		\begin{gather}
		\notag
		\forall \mathbf{x} \in \R^2,\quad \forall \mathbf{r} \in \lscr^\mathbf{R},\quad \frac{1}{C_\epsilon} e^{-(1 + \epsilon)\sqrt{\mu} \abs{\mathbf{x}-L\mathbf{r}}} \leq v_{L,\mathbf{r}}(\mathbf{x}) \leq C_\epsilon e^{-(1-\epsilon)\sqrt{\mu} \abs{\mathbf{x}-L\mathbf{r}}} \, ,\\
		\label{eq: exponential_decay_v_tilde_integrale}
		\forall R > 0,\quad \int_{\abs{\mathbf{x} - L\mathbf{r}} \geq R} \left(\lvert v_{L,\mathbf{r}}(\mathbf{x}) \rvert^2+ \lvert \nabla v_{L,\mathbf{r}}(\mathbf{x}) \rvert^2\right) \diff \mathbf{x} \leq C_\epsilon e^{-2(1-\epsilon)\sqrt{\mu}R} \pt
		\end{gather}
	\end{prop}
	\begin{proof}
		The function $v_L$ satisfies the equation $ \left(-\Delta + V_{L,\mathbf{r}} + \mu_L\right)v_{L,\mathbf{r}} = 0$ in $H^{-1}(\R^2)$ where $V_{L,\mathbf{r}}$ is compactly supported in the ball $B\left(L\mathbf{r},\left(\frac{1}{2}+\delta\right)Ld_0\right)$. Let $\epsilon \in \intoo{0}{\mu}$ and $R_\epsilon> 0$ such that $\abs{V(\mathbf{x})} \leq \epsilon/4$ for all $\abs{\mathbf{x}} \geq R_\epsilon$. Let $\Omega_\epsilon = B(0,R_\epsilon)^c$. We consider $Y_\pm$ the solution of the equation
		\[
		\left(-\Delta + \mu \pm \epsilon \right)Y_\pm = 0 \quad \text{on} \quad \Omega_\epsilon \pt
		\]
		We have $Y_\pm (\mathbf{x}) = K_0\left((\mu\pm\epsilon)^\frac{1}{2}\abs{\mathbf{x}}\right)$ where $K_0$ denotes the modified Bessel of the second kind with parameter 0. Following~\cite[p.266-267]{olver1997asymptotics}, we can show that the function $Y_\pm$ satisfies the asymptotics
		\begin{align*}
		Y_\pm(\mathbf{x}) = \sqrt{\frac{\pi}{2 (\mu\pm\epsilon)^\frac{1}{2} \abs{\mathbf{x}}}} e^{-(\mu\pm\epsilon)^\frac{1}{2}\abs{\mathbf{x}}} \left(1+\grandO{\frac{1}{\abs{\mathbf{x}}}}\right)\pt
		\end{align*}
		By the continuity of $v_{L,\mathbf{r}}$ and the convergence \eqref{eq:convergence_uniform} in $L^\infty(\R^2)$, there exists a constant $C_\epsilon>0$ independent from $L$ and such that
		\begin{align*}
		\forall \mathbf{x} \in\partial\Omega_\epsilon, \quad \frac{1}{C_\epsilon} Y_+(\mathbf{x} - L\mathbf{r}) \leq v_{L,\mathbf{r}}(\mathbf{x}) \leq C_\epsilon Y_-(\mathbf{x}-L\mathbf{r}) \pt
		\end{align*}
		Using the right side of \eqref{eq: closedness_modified_original_operator} and Assumption~\ref{hypo_4}, we show that $v_{L,\mathbf{r}}(\cdot + L\mathbf{r})$ (resp. $Y_-(\cdot - L\mathbf{r})$) is a supersolution of the operator $-\Delta + \mu + \epsilon$ (resp. $-\Delta + V_{L,\mathbf{r}}+ \mu_L$) for $L$ large enough, depending on $\epsilon$. We already know that $v_{L,\mathbf{r}}$ goes to zero at infinity. Then~\cite[Theorem 1.1]{hoffmann1980comparison} implies (notice that $V_{L,\mathbf{r}} + \mu_L \geq 0$ on $L\mathbf{r}+\Omega_\epsilon$ for $L$ large enough)
		\[
		\forall \mathbf{x} \in \Omega_\epsilon, \quad \frac{1}{C_\epsilon} Y_+(\mathbf{x} - L\mathbf{r}) \leq v_{L,\mathbf{r}}(\mathbf{x}) \leq C_\epsilon Y_-(\mathbf{x}-L\mathbf{r}) \pt
		\]
		We extend these estimates to $\R^2$ using $\norme{v_{L,\mathbf{r}} - v}_{L^\infty(\R^2)} = \petito{1}$ and we absorb the polynomial terms by slightly modifying $\epsilon$. The integral bounds \eqref{eq: exponential_decay_v_tilde_integrale} are shown as in the proof of Proposition~\ref{prop: exponential_decay_v}.
	\end{proof}
	
	\subsection{Orthonormalization procedure}\label{sec:orthomalization-procedure}
	
	In this section, we use the Gram-Schmidt process to construct an orthonormal family $\{w_{L,\mathbf{r}}\}_{\mathbf{r} \in \lscr^\mathbf{R}}$ from the almost orthonormal family $\{v_{L,\mathbf{r}}\}_{\mathbf{r} \in \lscr^\mathbf{R}}$. 
	
	First, we precise our notation concerning infinite matrices. Let $\Lambda$ be an infinite countable set. The Hilbert space of square summable and complex valued sequences labeled by $\Lambda$ is denoted by $\ell^2(\Lambda)$. A bounded operator $A \in \mathcal{B}(\ell^2(\Lambda))$ can be represented by an infinite matrix $(A(\mathbf{u,u'}))_{(\mathbf{u,u'}) \in \Lambda\times \Lambda}$ in the orthonormal basis $\{\ket{\mathbf{u}}\}_{\mathbf{u} \in \Lambda}$ where $\ket{\mathbf{u}} = (\delta(\mathbf{u,u'}))_{\mathbf{u'} \in \Lambda}$ with $\delta(\mathbf{u,u'})$ equal 1 if $\mathbf{u'=u}$ and $0$ otherwise. The identity matrix is denoted by $I$. The operator norm of a bounded operator $A$ will be denoted $\norme{A}$. In the sequel, we consider square infinite matrices indexed by $\mathscr{L}^\mathbf{R} \times \mathscr{L}^\mathbf{R}$. Finally, if a group $G$ acts on $\Lambda$ then we denote by $g \cdot A$ the action of $g\in G$ on $A\in\mathcal{B}(\ell^2(\Lambda))$ defined by
	\begin{align}
	\label{eq: action_matrix}
	g\cdot A \coloneqq T_g^* A T_g \, ,
	\end{align}
	where $T_g$ is the permutation matrix $T_g: \ket{\mathbf{u}} \mapsto \ket{g\cdot \mathbf{u}}$.
	
	Now, we discuss the construction of $\{w_{L,\mathbf{r}}\}_{\mathbf{r} \in \lscr^\mathbf{R}}$. In this setting, the tunneling coefficient is
	\begin{align*}
	\boxed{T_L \coloneqq \exp\left(-\sqrt{\mu} L\right)\pt}
	\end{align*}
	It measures the magnitude of the tunneling effect when $L \to \infty$. We construct the infinite Gram matrix for the family $\{v_{L,\mathbf{r}}\}_{\mathbf{r} \in \lscr^\mathbf{R}}$ as
	\begin{align*}
	Q_L (\mathbf{r},\mathbf{r'}) \coloneqq \pdtsc{v_{L,\mathbf{r}}}{v_{L,\mathbf{r'}}}_{L^2(\R^2)} >0 ,\quad \forall (\mathbf{r},\mathbf{r}') \in \lscr^\mathbf{R} \times \lscr^\mathbf{R} \pt
	\end{align*}
	Using the symmetry relations \eqref{eq: action_v}, the matrix $Q_L$ is $G$-invariant:
	\begin{align}
	\label{eq: QL_invariant}
	\forall g \in G,\quad g \cdot Q_L = Q_L \pt
	\end{align}
	In order to estimate the Gram matrix, we first need a lemma about convolutions of exponentials.
	\begin{lemma}
		\label{lemma: convolution_estimate}
		Let $\nu >0$. Then, for all $\epsilon \in \intoo{0}{1}$, we have
		\begin{align*}
		\forall \mathbf{x} \in \R^2,\quad \frac{1}{\nu^2}\left(1+\nu\abs{\mathbf{x}}\right)e^{-\nu\abs{\mathbf{x}}} \leq \left(e^{-\nu\abs{\cdot}} \ast e^{-\nu \abs{\cdot}}\right)(\mathbf{x}) \leq \frac{\pi}{2\nu^2}\left(1+\nu\abs{\mathbf{x}}\right)e^{-\nu\abs{\mathbf{x}}} \pt
		\end{align*}
	\end{lemma}
	\begin{proof}
		Let $\mathcal{E}_a$ is the ellipse defined by the equation $\absL{\mathbf{y-x}}+\absL{\mathbf{y}} = 2a$ where $a\geq\frac{\abs{\mathbf{x}}}{2}$ is the semi-major axis. Using the formula $\absL{\mathcal{E}_a} = 4aE(\abs{\mathbf{x}}/2a) $ where $E$ denotes the complete elliptic integral of second kind, we have
		\begin{align*}
		\left(e^{-\nu\abs{\cdot}} \ast e^{-\nu \abs{\cdot}}\right)(\mathbf{x}) = \int_{\frac{\abs{\mathbf{x}}}{2}}^\infty e^{-2\nu a} \absL{\mathcal{E}_a} \diff a = \int_{\abs{\mathbf{x}}}^\infty a e^{-\nu a}  E\left(\abs{\mathbf{x}}/a\right) \diff a \pt
		\end{align*}
		We use the estimate $1\leq E(e) \leq \frac{\pi}{2}$, valid for all $e\in\intff{0}{1}$, to conclude the proof of Lemma~\ref{lemma: convolution_estimate}.
	\end{proof}
	
	We recall that $v_{L,\mathbf{r}}$ admits upper and lower pointwise exponential bounds (see Proposition~\ref{prop: exponential_decay_v_tilde}). Then Lemma~\ref{lemma: convolution_estimate} implies that for all $\epsilon>0$ there exists $C_\epsilon>0$ such that
	\begin{align}
	\label{eq: matrix_element_Q_L}
	Q_L(0,0) = 1 \et \frac{1}{C_\epsilon} T_L^{(1+\epsilon)\lvert \mathbf{r-r'}\rvert} \leq Q_L (\mathbf{r},\mathbf{r'}) \leq C_\epsilon T_L^{(1-\epsilon)\lvert \mathbf{r-r'}\rvert} \pt
	\end{align}
	The polynomial terms have been absorbed by slightly modifying the $\epsilon$ in Proposition~\ref{prop: exponential_decay_v_tilde}.
	
	We recall that the group $G$ acts on the set of pairs of nearest neighbors $\mathscr{P}^\mathbf{R}$	introduced in \eqref{eq: definition_NN} and that the orbits of this action are denoted by $\mathcal{O}_1,\dots,\mathcal{O}_m$. For all $k \in \{1,\dots,m\}$, we consider $\mathbf{p}_k = (\mathbf{r}_k,\mathbf{r}_k') \in \mathcal{O}_k$ and we introduce the interaction coefficient 
	\begin{align*}
	\zeta_{L,k} \coloneqq \langle v_{L,\mathbf{r}_k}, v_{L,\mathbf{r}_k'} \rangle_{L^2(\R^2)} \, ,
	\end{align*}
	which does not depend on the choice of the representative in $\mathcal{O}_k$. Using the estimate~\eqref{eq: matrix_element_Q_L}, we immediately have
	\begin{align}
	\label{eq: estimation_zeta}
	\frac{1}{C_\epsilon}T_L^{(1+\epsilon)d_0} \leq \zeta_{L,k}  \leq C_\epsilon T_L^{(1-\epsilon)d_0} \, ,
	\end{align}
	for all $\epsilon >0$. We also introduce the adjacency matrix $J_k$ associated with the orbit $\mathcal{O}_k$, defined as
	\begin{align}
	\label{eq: defintion_J}
	J_k(\mathbf{r,r'}) = 
	\begin{cases}
	1 & \text{if}~ (\mathbf{r,r'}) \in \mathcal{O}_k \, ,\\
	0 & \text{otherwise} \pt
	\end{cases}
	\end{align}
	\begin{lemma}
		\label{lemma: expansion_Q_L}
		We have the expansion
		\begin{align*}
		\normeL{Q_L - I - \sum_{k=1}^m \zeta_{L,k} J_k}_{\ell^2(\lscr^\mathbf{R}) \to \ell^2(\lscr^\mathbf{R})} = \grandO{T_L^{d_1-}} \pt
		\end{align*}
	\end{lemma}
	We recall that the notation $\grandO{T_L^{d_1-}}$ is defined in Section~\ref{sec:notations-and-conventions}.
	\begin{proof}
		To show this expansion, first notice that the matrix element $(\mathbf{r,r'})$ of $Q_L-I- \sum_{k=1}^m \zeta_{L,k} J_k$ is equal to zero if $\mathbf{r=r'}$ or $(\mathbf{r,r'}) \in \mathscr{P}^\mathbf{R}$. It remains to call on Schur's test~\cite[Appendix A.1]{grafakos2014modern} and the estimates \eqref{eq: matrix_element_Q_L}. 
	\end{proof}
	Let $K_L \coloneqq Q_L - I$. As a consequence of Lemma~\ref{lemma: expansion_Q_L} and \eqref{eq: estimation_zeta}, we have
	\begin{align}
	\label{eq: norm_K_L}
	\norme{K_L} = \grandO{T_L^{d_0-}} \, ,
	\end{align}
	and by the (power series) functional calculus, we can write for $L$ large enough
	\begin{align}
	\label{eq: gram_schmidt}
	Q_L^{-1} = I - K_L + \grandO{T_L^{2d_0-}}  \et
	Q_L^{-\frac{1}{2}} = I - \frac{1}{2}K_L +\grandO{T_L^{2d_0-}} \pt
	\end{align}
	The infinite matrices $Q_L^{-1}$ and $Q_L^{-\frac{1}{2}}$ inherit the same periodicity property as $Q_L$, see \eqref{eq: QL_invariant}:
	\begin{align}
	\label{eq: QL_invariant_bis}
	\forall g \in G,\quad g\cdot Q_L^{-1} = Q_L^{-1} \et g\cdot Q_L^{-\frac{1}{2}} = Q_L^{-\frac{1}{2}} \pt
	\end{align}
	The next lemma states that $Q_L^{-1}$ and $Q_L^{-\frac{1}{2}}$ are well-localized in the sense that their off-diagonal coefficients decays exponentially fast with $\abs{\mathbf{r-r'}}$.
	\begin{lemma}
		\label{lemma: localization_inverse_G}
		Let $\epsilon \in \intoo{0}{1}$. There exists $L_\epsilon \geq 1$ such that for $L\geq L_\epsilon$ and for all $(\mathbf{r,r'}) \in (\lscr^\mathbf{R})^2$, we have
		\begin{align*}
		\abs{Q^{-1}_L(\mathbf{r,r'})} \leq 2 T_L^{(1-\epsilon)\abs{\mathbf{r-r'}}} \et \abs{Q^{-\frac{1}{2}}_L(\mathbf{r,r'})} \leq 2 T_L^{(1-\epsilon)\abs{\mathbf{r-r'}}} \pt
		\end{align*}
	\end{lemma}
	The proof of Lemma~\ref{lemma: localization_inverse_G} is an adaptation of the one of~\cite[Proposition 2]{jaffard1990proprietes} hence we do not write it and we refer the reader to~\cite{cazalisthesis} for the details.
	
	For $\mathbf{r} \in \lscr^\mathbf{R}$, we set
	\begin{align*}
	\boxed{w_{L,\mathbf{r}}\coloneqq \sum_{\mathbf{r'} \in \lscr^\mathbf{R}}Q_L^{-\frac{1}{2}}(\mathbf{r},\mathbf{r'} )v_{L,\mathbf{r'}} \pt}
	\end{align*}
	We recall that $s_L$ denotes the dilatation operator, defined by $s_L\mathbf{x} = L\mathbf{x}$.
	\begin{prop}[Properties of $w_{L,\mathbf{r}}$]
		\label{prop: properties_wannier}
		The series defining $w_{L,\mathbf{r}}$ converges in $H^{\min(p,2)}(\R^2)$. The family $\{w_{L,\mathbf{r}}\}_{\mathbf{r} \in \lscr^\mathbf{R}}$ forms an orthonormal family of $L^2(\R^2)$, satisfies the same symmetry relations \eqref{eq: action_v} as $\{v_{L,\mathbf{r}}\}_{\mathbf{r}\in\lscr^\mathbf{R}}$ 
		\begin{align}
		\label{eq: properties_wannier1}
		\forall g \in G,\quad \forall \mathbf{r} \in \lscr^\mathbf{R},\quad (s_Lgs_L^{-1})\cdot w_{L,\mathbf{r}} =  w_{L,g \cdot\mathbf{r}} \, ,
		\end{align}
		and the following pointwise exponential bounds: for all $\epsilon \in \intoo{0}{1}$ there exists $C_\epsilon >0 $ such that
		\begin{align}
		\label{eq: properties_wannier2}
		\forall \mathbf{r} \in\lscr^\mathbf{R}, \quad\forall \mathbf{x} \in \R^2,\quad \quad\abs{w_{L,\mathbf{r}}(\mathbf{x})} \leq C_\epsilon e^{-(1-\epsilon)\sqrt{\mu}\abs{\mathbf{x}-L\mathbf{r}}} \pt
		\end{align}
	\end{prop}
	\begin{proof}
		By Proposition~\ref{prop: regularity} and the fact that $(Q_L^{-\frac{1}{2}}(\mathbf{r},\mathbf{r'} ))_{\mathbf{r'} \in \lscr^\mathbf{R}}$ is summable (see Lemma~\ref{lemma: localization_inverse_G}), we have
		\begin{align*}
		\sum_{\mathbf{r'} \in \lscr^\mathbf{R}}\abs{Q_L^{-\frac{1}{2}}(\mathbf{r},\mathbf{r'} )}\norme{v_{L,\mathbf{r'}}}_{H^{\min(p,2)}(\R^2)} 
		&\leq \sum_{\mathbf{r'} \in \lscr^\mathbf{R}}\abs{Q_L^{-\frac{1}{2}}(\mathbf{r},\mathbf{r'} )} \left(\norme{H_{L,\mathbf{r'}} v_{L,\mathbf{r'}}}_{L^2(\R^2)} + C\norme{v_{L,\mathbf{r'}}}_{L^2(\R^2)}\right) \\
		&\leq\left( \abs{\mu_L} + C\right) \sum_{\mathbf{r'} \in \lscr^\mathbf{R}}\abs{Q_L^{-\frac{1}{2}}(\mathbf{r},\mathbf{r'} )} < \infty \pt
		\end{align*}
		Notice that we can take the same constant $C$ for each $\mathbf{r'} \in \lscr^\mathbf{R}$ in Proposition~\ref{prop: regularity} since we have $\norme{V_{L,\mathbf{r}}}_{L^p(\R^2)+L^\infty(\R^2)} = \norme{V_{L,\mathbf{r'}}}_{L^p(\R^2)+L^\infty(\R^2)}$ for all $(\mathbf{r,r'}) \in \lscr^\mathbf{R}\times \lscr^\mathbf{R}$, see \eqref{eq:regularity_constant}. This shows the first assertion of Proposition~\ref{prop: properties_wannier}. In addition, the family $w_{L,\mathbf{r}}$ is orthonormal by construction.
		
		To show \eqref{eq: properties_wannier1}, we use the $G$-invariance of $Q^{-\frac{1}{2}}_L$ (see the right side of \eqref{eq: QL_invariant_bis})
		\begin{multline*}
		(s_Lgs_L^{-1})\cdot w_{L,\mathbf{r}} 
		= \sum_{\mathbf{r'} \in \lscr^\mathbf{R}} Q_L^{-\frac{1}{2}}(\mathbf{r,r'}) v_{L,g\cdot \mathbf{r'}}
		= \sum_{\mathbf{r'} \in \lscr^\mathbf{R}} Q_L^{-\frac{1}{2}}(\mathbf{r},g^{-1}\cdot\mathbf{r'}) v_{L,\mathbf{r'}}\\
		= \sum_{\mathbf{r'} \in \lscr^\mathbf{R}} Q_L^{-\frac{1}{2}}(g\cdot\mathbf{r},\mathbf{r'}) v_{L,\mathbf{r'}}
		= w_{L,g\cdot\mathbf{r}} \pt
		\end{multline*}
		Finally, the exponential bound \eqref{eq: properties_wannier2} is shown by using Proposition~\ref{prop: exponential_decay_v_tilde} and Lemma~\ref{lemma: localization_inverse_G}.
	\end{proof}
	In particular, the family $\{w_{L,\mathbf{r}}\}_{\mathbf{r} \in \lscr^\mathbf{R}}$ is generated by the shifts of the functions $\{w_{L,\mathbf{r}}\}_{\mathbf{r} \in \mathbf{R}}$ with respect to $\lscr_L$. We introduce the Bloch-Floquet transform of these functions
	\begin{align*}
	u_{L,\mathbf{r}}(\mathbf{k},\mathbf{x}) \coloneqq  \mathcal{U}_\mathrm{BF}(w_{L,\mathbf{r}})(\mathbf{k,x}) = \sum_{\mathbf{u} \in \lscr} e^{i\mathbf{k\cdot}L\mathbf{u}} w_{L,\mathbf{u+r}}(\mathbf{x}) \pt
	\end{align*}
	We have the following pseudo-periodicity relations: for all $\mathbf{k} \in \Gamma_L^*$, for all $\mathbf{x} \in \Gamma_L$, 
	\begin{align*}
	u_{L,\mathbf{r}}(\mathbf{k},\cdot) \in \lk^2(\Gamma_L) \et u_{L,\mathbf{r}}(\cdot,\mathbf{x}) \in \lper^2(\Gamma_L^*) \pt
	\end{align*}
	Hence, we can extend the definition of $u_{L,\mathbf{r}}$ for any element of $\mathbf{r\in}\lscr^\mathbf{R}$ with
	\begin{align}
	\label{eq: periodicity_relations}
	\forall \mathbf{k} \in\Gamma_L^*,\quad\forall \mathbf{u} \in \lscr,\quad\forall \mathbf{r} \in \mathbf{R},\quad \mathcal{U}_\mathrm{BF}(w_{L,\mathbf{u+r}})(\mathbf{k,\cdot}) =  e^{-i\mathbf{k}\cdot L\mathbf{u}} u_{L,\mathbf{r}}(\mathbf{k},\cdot) \pt
	\end{align}
	\begin{prop}
		\label{prop: properties_u}
		For all $\mathbf{k}\in\Gamma_L^*$ and for all $(\mathbf{r,r'}) \in \mathbf{R}^2$, we have the relations
		\begin{gather*}
		\pdtsc{u_{L,\mathbf{r}}(\mathbf{k},\cdot)}{u_{L,\mathbf{r'}}(\mathbf{k},\cdot)}_{L^2(\Gamma_L)} = \delta_{\mathbf{rr'}} \, ,\\
		\notag \pdtsc{\mathcal{U}_\mathrm{BF}(v_{L,\mathbf{r}})(\mathbf{k},\cdot)}{\mathcal{U}_\mathrm{BF}(v_{L,\mathbf{r'}})(\mathbf{k},\cdot)}_{L^2(\Gamma_L)} = \delta_{\mathbf{rr'}} + \grandO{T_L^{d_0-}} \, , \\
		\spn\left\{ u_{L,\mathbf{r}}(\mathbf{k},\cdot)\right\}_{\mathbf{r} \in \mathbf{R}} = \spn \left\{\mathcal{U}_{\mathrm{BF}}(v_{L,\mathbf{r}})(\mathbf{k},\cdot)\right\}_{\mathbf{r} \in \mathbf{R}} \pt
		\end{gather*}
	\end{prop}
	\begin{proof}
		We have
		\begin{gather*}
		\pdtsc{u_{L,\mathbf{r}}(\mathbf{k},\cdot)}{u_{L,\mathbf{r'}}(\mathbf{k},\cdot)}_{L^2(\Gamma_L)} = \sum_{\mathbf{u} \in \lscr} e^{i\mathbf{k}\cdot L\mathbf{u}} \pdtsc{w_{L,\mathbf{r}}}{w_{L,\mathbf{u+r'}}}_{L^2(\R^2)}  = \delta_{\mathbf{rr'}}\, , \\
		\pdtsc{\mathcal{U}_\mathrm{BF}(v_{L,\mathbf{r}})(\mathbf{k},\cdot)}{\mathcal{U}_\mathrm{BF}(v_{L,\mathbf{r'}})(\mathbf{k},\cdot)}_{L^2(\Gamma_L)} =\sum_{\mathbf{u} \in \lscr} e^{i\mathbf{k}\cdot L\mathbf{u}} \pdtsc{v_{L,\mathbf{r}}}{v_{L,\mathbf{u+r'}}}_{L^2(\R^2)} = \delta_{\mathbf{rr'}} + \grandO{T_L^{d_0-}} \pt
		\end{gather*}
		We have used that $\{w_{L,\mathbf{r}}\}_{\mathbf{r}\in\lscr^\mathbf{R}}$ is an orthonormal family (see Proposition~\ref{prop: properties_wannier}) in the first line and the estimate \eqref{eq: matrix_element_Q_L} in the second one. Note that the computation is justified since the sum converges absolutely. Indeed, we have
		\begin{align*}
		\int_{\Gamma_L} \sum_{\mathbf{u,u'} \in \lscr} \abs{w_{L,\mathbf{u+r}}w_{L,\mathbf{u'+r'}}} = \sum_{\mathbf{u} \in \lscr} \int_{\R^2} \abs{w_{L,\mathbf{r}}w_{L,\mathbf{u+r'}}} = \sum_{\mathbf{u} \in \lscr} (\abs{w_{L,\mathbf{r}}} \ast \abs{w_{L,\mathbf{r'}}})(\mathbf{u}) \, ,
		\end{align*}
		and we call on Proposition~\ref{prop: properties_wannier} and Lemma~\ref{lemma: convolution_estimate} to show that this sum is finite. Same argument applies to the second sum. The family $\{u_{L,\mathbf{r}}(\mathbf{k},\cdot)\}_{\mathbf{r} \in \mathbf{R}}$ forming an orthonormal system, it spans a subspace of dimension $N = \abs{\mathbf{R}}$. Moreover, using the periodicity properties of $Q_L^{-\frac{1}{2}}$ (see \eqref{eq: QL_invariant_bis}), the exponential bounds on $v_{L,\mathbf{r'}}$ (Proposition~\ref{prop: exponential_decay_v_tilde}) and the fact that $Q_L^{-\frac{1}{2}}$ is well-localized (Lemma~\ref{lemma: localization_inverse_G}), one can show 
		\begin{align}
		\label{eq: explicit_decomposition}
		u_{L,\mathbf{r}}(\mathbf{k},\cdot) = \sum_{\mathbf{r'} \in \mathbf{R}} \left(\sum_{\mathbf{u} \in \lscr} Q^{-\frac{1}{2}}_L(\mathbf{u + r},\mathbf{r'}) e^{i\mathbf{k} \cdot L\mathbf{u}}\right) \mathcal{U}_\mathrm{BF}(v_{L,\mathbf{r'}})(\mathbf{k},\cdot)\pt
		\end{align}
		This concludes the proof of Proposition~\ref{prop: properties_u}.
	\end{proof}
	
	\subsection{Interaction matrix}
	
	We denote by $E_L \coloneqq \spn\{w_{L,\mathbf{r}}\}_{\mathbf{r}\in\lscr^\mathbf{R}}$ the subspace spanned by the family $\{w_{L,\mathbf{r}}\}_{\mathbf{r\in}\lscr^\mathbf{R}}$ and by $P_L$ the orthogonal projector on $E_L$. In this section, we compute the matrix $A_L$ of $P_L H_L P_{L{\restreinta E_L}}$ in the orthonormal basis $\{w_{L,\mathbf{r}}\}_{\mathbf{r} \in \lscr^\mathbf{R}}$. 
	We recall the definition \eqref{eq: interaction_coefficient} of the interaction coefficient
	\begin{align*}
	\theta_{L,k}=  \langle v_{L,\mathbf{r}_k}, V_{L,\mathbf{r}_k}(1-\chi_{L,\mathbf{r}_k'})v_{L,\mathbf{r}_k'}\rangle_{L^2(\R^2)} \, ,
	\end{align*}
	where $\mathbf{c}_k = (\mathbf{r}_k,\mathbf{r}_k') \in \mathcal{O}_k$. In Section~\ref{sec:properties-of-the-mono-atomic-operators}, we prove that $\theta_{L,k}$ does not depend on the choice of the pair $\mathbf{c}_k \in \mathcal{O}_k$. Also, using the relations~\eqref{eq: properties_wannier1} and Assumption~\ref{hypo_2}, we deduce that $A_L$ has the same symmetries as $V_L$:
	\begin{align}
	\label{eq: A_invariance}
	\forall g \in G,\quad g\cdot A_L = A_L \, ,
	\end{align}
	where the action of $G$ on $\mathcal{B}(\ell^2(\lscr^\mathbf{R}))$ is defined in \eqref{eq: action_matrix}. 
	
	The next proposition provides second order estimates on the coefficients of $A_L$. We recall that $-\mu_L$ denotes the common first eigenvalue of the mono-atomic operators $H_{L,\mathbf{R}}=-\Delta+V_{L,\mathbf{r}}$, see Section~\ref{sec:properties-of-the-mono-atomic-operators}.
	\begin{prop}
		\label{prop: interaction_matrix}
		We have the following expansion
		\begin{align*}
		\boxed{A_L = -\mu_L I + \sum_{k=1}^m \theta_{L,k} J_k + \grandO{T_L^{(1+\delta)d_0-} + T_L^{d_1-}} \et \forall k\in\{1,\dots,m\},\quad \theta_{L,k} = \grandO{T_L^{d_0-}} \, ,}
		\end{align*}
		where $J_k$, defined in \eqref{eq: defintion_J}, is the adjacency matrix associated with the orbit $\mathcal{O}_k$ and where $d_1 > d_0$ denotes the second nearest neighbor distance in $\lscr^\mathbf{R}$.
	\end{prop}
	\begin{proof}
		For clarity, we decompose the proof into several steps.
		\paragraph{First step.} We give an explicit expression of the matrix elements of $A_L$. Let $(\mathbf{r,r'}) \in (\lscr^\mathbf{R})^2$. First, notice that the quadratic form associated with $H_L$ is continuous on $H^1(\R^2)$ by Proposition~\ref{lemma: kato_type_inequality}. Moreover, the series defining $w_{L,\mathbf{r}}$ converges in $H^1(\R^2)$ by Proposition~\ref{prop: properties_wannier}. Therefore, we have
		\begin{align*}
		\pdtsc{w_{L,\mathbf{r}}}{H_L w_{L,\mathbf{r'}}}_{L^2(\R^2)} 
		& =  \sum_{\mathbf{s,s'} \in \lscr^\mathbf{R}} Q_L^{-\frac{1}{2}}(\mathbf{r,s})Q_L^{-\frac{1}{2}}(\mathbf{r',s'}) \pdtsc{v_{L,\mathbf{s}}}{(-\Delta + V_L)v_{L,\mathbf{s'}}}_{L^2(\R^2)} \\
		&= - \mu_L \delta(\mathbf{r,r'}) + (Q_L^{-\frac{1}{2}} D_L Q_L^{-\frac{1}{2}})(\mathbf{r,r'}) \, ,
		\end{align*}
		where $D_L$ is the infinite matrix defined by
		\[
		\forall (\mathbf{r,r'}) \in \lscr^\mathbf{R} \times \lscr^\mathbf{R},\quad D_L(\mathbf{r},\mathbf{r'}) =  \pdtsc{v_{L,\mathbf{r}}}{V_L\left(1-\chi_{L,\mathbf{r'}} \right) v_{L,\mathbf{r'}}}_{L^2(\R^2)} \pt
		\]
		Using Assumption~\ref{hypo_2} and the relations~\eqref{eq: action_v}, we have
		\begin{align}
		\label{eq: D_L_invariance}
		\forall g \in G,\quad g\cdot D_L = D_L \pt
		\end{align}
		\paragraph{Second step.} Now, we show that $D_L$ is well-localized in the following sense: for all $\epsilon >0$ there exists $C_\epsilon>0$ such that for all $(\mathbf{r,r'}) \in( \lscr^\mathbf{R})^2$ we have
		\begin{align}
		\label{eq: estimate_D_L}
		\abs{D_L(\mathbf{r,r'})} \leq 
		\begin{cases}
		C_\epsilon T_L^{(1-\epsilon)(1+\delta)d_0} &\text{if } \mathbf{r=r'}\, ,\\
		C_\epsilon T_L^{(1-\epsilon)\lvert\mathbf{r-r'}\rvert}&\text{if } \mathbf{r\neq r'} \pt
		\end{cases}
		\end{align}
		To lighten the notation, we introduce $\rho_{\mathbf{rr'}} = v_{L,\mathbf{r}}(1-\chi_{L,\mathbf{r'}}) v_{L,\mathbf{r'}}$.	
		
		If $\mathbf{r=r'}$ then Proposition~\ref{lemma: kato_type_inequality} shows that (recall that $\sqrt{1-\chi_{L,\mathbf{r}}}$ is smooth enough, see \eqref{eq: chi_regularity})
		\begin{align*}
		\abs{D_L(\mathbf{r,r})} \leq \norme{\nabla \sqrt{\rho_{\mathbf{rr}}} }^2_{L^2(\R^2)} + C(1,L) \norme{\sqrt{\rho_{\mathbf{rr}}}}^2_{L^2(\R^2)} \, ,
		\end{align*}
		where $C(1,L)$ is polynomial in $L$. We obtain $\abs{D_L(\mathbf{r,r'})} = \grandO{T_L^{(1+\delta)d_0-}}$ with Proposition~\ref{prop: exponential_decay_v_tilde}.
		
		Now, assume $\mathbf{r\neq r'}$.
		Using the periodicity of $V_L$ and Hölder's inequality, we have
		\begin{align*}
		\abs{D_L(\mathbf{r,r'})} \leq \norme{V_L}_{\lper^p} \sum_{\mathbf{u} \in \lscr} \norme{\rho_{\mathbf{rr'}}}_{L^{p'}(\Gamma_L+L\mathbf{u})} \, ,
		\end{align*}
		where $p'\in\intoo{1}{\infty}$ is the conjugate exponent of $p$. Let $\mathcal{B}$ be the ball $B\left(\frac{\mathbf{r+r'}}{2},\left(2 + \frac{d}{d_0}\right)\abs{\mathbf{r-r'}}\right)$ where $d_0>0$ is defined in \eqref{eq: min_distance_lattice} and where $d \coloneqq \max_{\mathbf{x}\in\Gamma} \abs{\mathbf{x}}$. We split the sum as follows
		\begin{align*}
		\sum_{\mathbf{u} \in \lscr} \norme{\rho_{\mathbf{rr'}}}_{L^{p'}(\Gamma_L+L\mathbf{u})} 
		= \sum_{\mathbf{u} \in  \lscr\cap \mathcal{B}} \norme{\rho_{\mathbf{rr'}}}_{L^{p'}(\Gamma_L+L\mathbf{u})}  +  \sum_{\mathbf{u} \in \lscr \cap \mathcal{B}^c} \norme{\rho_{\mathbf{rr'}}}_{L^{p'}(\Gamma_L+L\mathbf{u})} = I_1+I_2\pt
		\end{align*}
		The number of vertices of $\lscr$ inside the ball $\mathcal{B}$ is estimated by a $\grandO{\abs{\mathbf{r-r'}}^2}$. Thus, using the concavity inequality $\sum_{i=1}^n \abs{a_i}^\theta \leq n^{1-\theta} \left(\sum_{i=1}^n \abs{a_i}\right)^\theta$ valid for all $\theta \in \intff{0}{1}$, all $n\in\N$ and all $(a_1,\dots,a_n)\in\C^n$, we can bound $I_1$ by $\grandO{\abs{\mathbf{r-r'}}^{2/p} \norme{\rho_{\mathbf{rr'}}}_{L^{p'}(\R^2)}}$. Then, we estimate $\norme{\rho_{\mathbf{rr'}}}_{L^{p'}(\R^2)}$ using Lemma~\ref{lemma: convolution_estimate} and Proposition~\ref{prop: exponential_decay_v_tilde}: for all $\epsilon>0$ there exists $C_\epsilon,C_\epsilon'>0$ such that
		\begin{align}
		\label{eq: interaction_12}
		I_1 \leq C_\epsilon \frac{\abs{\mathbf{r-r'}}^{2/p}}{1+(L\abs{\mathbf{r-r'}})^{\frac{1}{2} - \frac{1}{p'}}} T_L^{(1-\epsilon/2) \abs{\mathbf{r-r'}}} \leq C'_\epsilon  T_L^{(1-\epsilon) \abs{\mathbf{r-r'}}} \pt
		\end{align}
		Now, we estimate $I_2$. First, we notice that $ \mathcal{B}^c \subset \mathcal{B}_1^c \cap \mathcal{B}_2^c$ where $\mathcal{B}_1 \coloneqq B\left(\mathbf{r},\left(\frac{3}{2} + \frac{d}{d_0}\right)\abs{\mathbf{r-r'}}\right)$ and $\mathcal{B}_2 \coloneqq B\left(\mathbf{r'},\left(\frac{3}{2} + \frac{d}{d_0}\right)\abs{\mathbf{r-r'}}\right)$. If $(\mathbf{r,r'}) \in \lscr^\mathbf{R} \times \lscr^\mathbf{R}$ are such that $\abs{\mathbf{r-r'}} \geq d_0 $, we have
		\begin{align*}
		\bigcup_{\mathbf{u} \in \lscr \cap \mathcal{B}_1^c} \left(\Gamma_L + L\mathbf{u}\right)
		\subset L (\mathcal{B}_1^c+\Gamma) 
		\subset L B\left(\mathbf{r}, \left(\frac{3}{2} + \frac{d}{d_0}\right) \abs{\mathbf{r-r'}} - d\right)^c  
		\subset L B\left(\mathbf{r}, \frac{3}{2} \abs{\mathbf{r-r'}}\right)^c \, ,
		\end{align*}
		and similar ones for $\mathcal{B}_2$. Then, using the Cauchy-Schwarz inequality, the Sobolev embedding $ H^1(\Gamma_L+L\mathbf{u}) \subset L^{2p'}(\Gamma_L+L\mathbf{u})$ (where the embedding constant does not depend on $L$ nor on $\mathbf{u}$, see~\cite[Theorem 2.28]{aubin1998somenonlinear}) and the inclusions stated above, we have 
		\begin{align*}
		I_2
		&\leq \sum_{\mathbf{u} \in \lscr \cap \mathcal{B}_1^c \cap \mathcal{B}_2^c} \norme{v_{L,\mathbf{r}}}_{L^{2p'}(\Gamma_L+L\mathbf{u})}\norme{v_{L,\mathbf{r'}}}_{L^{2p'}(\Gamma_L+L\mathbf{u})} \\
		&\leq \frac{1}{2} \left(\sum_{\mathbf{u} \in \lscr \cap \mathcal{B}_1^c} \norme{v_{L,\mathbf{r}}}_{L^{2p'}(\Gamma_L+L\mathbf{u})}^2 + \sum_{\mathbf{u} \in \lscr \cap \mathcal{B}_2^c} \norme{v_{L,\mathbf{r'}}}_{L^{2p'}(\Gamma_L+L\mathbf{u})}^2\right) \\
		&\lesssim \norme{v_{L,\mathbf{r}}}_{H^1(B(L\mathbf{r},\frac{3}{2} L \abs{\mathbf{r-r'}})^c)}^2 + \norme{v_{L,\mathbf{r'}}}_{H^1(B(L\mathbf{r'},\frac{3}{2} L\abs{\mathbf{r-r'}})^c)}^2 \pt
		\end{align*}
		Then, using Proposition~\ref{prop: exponential_decay_v_tilde}, we deduce that for $L$ large enough we have
		\begin{align}
		\label{eq: interaction_13}
		I_2 \lesssim \norme{v_{L,\mathbf{r}}}_{H^1(B(L\mathbf{r},\frac{3}{2} L\abs{\mathbf{r-r'}})^c)}^2
		\lesssim  T_L^{2\abs{\mathbf{r-r'}}}  \pt
		\end{align}
		Combining \eqref{eq: interaction_13} with \eqref{eq: interaction_12} shows \eqref{eq: estimate_D_L}.
		\paragraph{Third step.}
		Using the symmetry relations \eqref{eq: D_L_invariance}, we can write $D_L$ as
		\[
		D_L = \alpha_L I + \sum_{k=1}^m\beta_{L,k} J_k + R_L \, ,
		\]
		for some constants $\alpha_L,\beta_{L,1},\dots,\beta_{L,m}$ such that the diagonal and the nearest neighbors matrix elements of $R_L$ are equal to zero.
		The operator norm of $R_L$ is estimated thanks to \eqref{eq: estimate_D_L} and by employing the same method used for the estimation of $\norme{Q_L-I- \sum_{k=1}^m \zeta_{L,k} J_k}$ in Section~\ref{sec:orthomalization-procedure}: $\norme{R_L} = \grandO{T_L^{d_1-}}$ where $d_1 > d_0$ is the second nearest neighbor distance. Now, we estimate the coefficients $\alpha_L$ and $\beta_{L,k}$.
		
		Let $\mathbf{r} \in \lscr^\mathbf{R}$. Using that $\sqrt{1-\chi_{L,\mathbf{r}}} \in \mathcal{C}^1(\R^2)$ (see the assumption \eqref{eq: chi_regularity}), we have $\sqrt{1-\chi_{L,\mathbf{r}}}v_{L,\mathbf{r}} \in H^1(\R^2)$. Then, adapting the proof of inequality \eqref{eq: kato_type_inequality1} from Proposition~\ref{lemma: kato_type_inequality}, we can show
		\begin{align*}
		\alpha_L = \pdtsc{\sqrt{1-\chi_{L,\mathbf{r}}}v_{L,\mathbf{r}}}{V_L \sqrt{1-\chi_{L,\mathbf{r}}}v_{L,\mathbf{r}}}_{L^2(\R^2)}
		\lesssim \norme{V_L}_{\lper^p} \norme{\sqrt{1-\chi_{L,\mathbf{r}}}v_{L,\mathbf{r}}}^2_{H^1(\R^2)} \pt
		\end{align*}
		Then, using that $\supp \sqrt{1-\chi_{L,\mathbf{r}}} \subset B\left(L\mathbf{r}, \frac{1+\delta}{2}Ld_0\right)^c$, $\norme{\nabla \chi_{L,\mathbf{r}}}_{L^\infty} = \grandO{L^{-1}}$ (see Section~\ref{sec:effective-mono-atomic-operator} for these two facts), the estimate~\eqref{rem: grow_polynomial} and the estimate \eqref{eq: exponential_decay_v_tilde_integrale} from Proposition~\ref{prop: exponential_decay_v_tilde}, we show that
		\begin{align*}
		\alpha_L = \grandO{T_L^{(1+\delta)d_0-}} \pt
		\end{align*}
		
		Let $(\mathbf{r,r'})\in \mathscr{P}^\mathbf{R}$. Using the Cauchy-Schwarz inequality and the previous estimates twice, we obtain
		\begin{align*}
		\absL{\pdtsc{(1-\chi_{L,\mathbf{r}})v_{L,\mathbf{r}}}{V_L(1-\chi_{L,\mathbf{r'}})v_{L,\mathbf{r'}}}_{L^2(\R^2)}} = \grandO{T_L^{(1+\delta)d_0-}} \pt
		\end{align*}
		We recall that the nearest neighbor interaction coefficient $\theta_{L,k}$ defined in \eqref{eq: interaction_coefficient} does not depend on the choice of $(\mathbf{r,r'}) \in \mathcal{O}_k$. Hence, we have shown
		\begin{align*}
		\beta_{L,k} = \theta_{L,k} + \grandO{T_L^{(1+\delta)d_0-}} \et D_L = \sum_{k=1}^{m}\theta_{L,k} J_k + \grandO{T_L^{(1+\delta)d_0-} + T_L^{d_1-}} \pt
		\end{align*}
		\paragraph{Fourth step.}
		From \eqref{eq: estimate_D_L}, we deduce $\beta_{L,k} = \grandO{T_L^{d_0-}}$ which implies $\theta_{L,k} = \grandO{T_L^{d_0-}}$.
		Then, using~\eqref{eq: gram_schmidt}, the norm estimation \eqref{eq: norm_K_L} and the first step of the proof, we have
		\[
		A_L = -\mu_L I + Q_L^{-\frac{1}{2}} D_L Q_L^{-\frac{1}{2}} = -\mu_L I +  \sum_{k=1}^{m}\theta_{L,k} J_k + \grandO{T_L^{(1+\delta)d_0-}+ T_L^{d_1-}} \pt \qedhere
		\]
		This concludes the proof of Proposition~\ref{prop: interaction_matrix}.
	\end{proof}
	We recall that $H_L$ admits the decomposition in fibers $H_L = \fint_{\Gamma_L^*}^\oplus H_L(\mathbf{k}) \diff \mathbf{k}$ where $H_L(\mathbf{k}) = -\Delta + V_L$ on the subspace $\lk^2(\Gamma_L)$. We denote by $A_L(\mathbf{k})$ the matrix of $H_L(\mathbf{k})$ restricted to the $N$-dimensional subspace $E_L(\mathbf{k}) \coloneqq \spn\left\{u_{L,\mathbf{r}}(\mathbf{k},\cdot)\right\}_{\mathbf{r} \in \mathbf{R}}$.
	\begin{cor}
		\label{cor: expansion_A_k}
		We have the expansion 
		\begin{align*}
		\boxed{A_L(\mathbf{k}) = -\mu_L I + \sum_{k=1}^{m}\theta_{L,k} B_k(L\mathbf{k}) + \grandO{T_L^{(1+\delta)d_0-}+T_L^{d_1-}} \, ,}
		\end{align*}
		where the $O$ makes sense in $\mathcal{C}^\infty_\mathrm{per}(\Gamma_L^*)$, where $d_1>d_0$ denotes the second nearest neighbor distance in $\lscr^\mathbf{R}$ and where the matrix elements of $B_k(\mathbf{k})$ are given by
		\begin{align*}
		\forall \mathbf{k} \in \Gamma^*,\quad\forall (\mathbf{r,r'}) \in \mathbf{R}^2,\quad B_k(\mathbf{k}; \mathbf{r,r'}) = \sum_{\substack{\mathbf{u} \in \lscr \\ (\mathbf{\mathbf{r},\mathbf{u+r'}}) \in \mathcal{O}_k}} e^{i\mathbf{k}\cdot \mathbf{u}} \pt
		\end{align*}
	\end{cor}
	\begin{proof}
		Let $(\mathbf{r,r'}) \in \mathbf{R}^2$. Using that $\mathcal{U}_\mathrm{BF}$ defines an isometry from $L^2(\R^2)$ to $L^2(\Gamma_L^*,L^2(\Gamma_L))$ and the relations \eqref{eq: periodicity_relations}, we have
		\begin{align*}
		A_L(\mathbf{r,u+r'}) 
		& = \pdtsc{w_{L,\mathbf{r}}}{H_L w_{L,\mathbf{u+r'}}}_{L^2(\R^2)} \\
		& = \fint_{\Gamma_L^*}  e^{i\mathbf{k}\cdot L\mathbf{u}} \pdtsc{u_{L,\mathbf{r}}(\mathbf{k},\cdot)}{H_L(\mathbf{k}) u_{L,\mathbf{r'}}(\mathbf{k},\cdot)}_{L^2(\Gamma_L)}\diff \mathbf{k} = \fint_{\Gamma_L^*}  e^{i\mathbf{k}\cdot L\mathbf{u}} A_L(\mathbf{k};\mathbf{r,r'}) \diff \mathbf{k}\pt
		\end{align*}
		Te first step of the proof of Proposition~\ref{prop: interaction_matrix} implies that the off-diagonal coefficients of $A_L$ are given by $Q^{-\frac{1}{2}}_L D_L Q^{-\frac{1}{2}}_L$. By Lemma~\ref{lemma: localization_inverse_G} and the estimates \eqref{eq: estimate_D_L}, the off-diagonal coefficients of $D_L$ and $Q_L^{-\frac{1}{2}}$ are exponentially decaying with the same rate. One can show the similar property for $Q^{-\frac{1}{2}}_L D_L Q^{-\frac{1}{2}}_L$, that is, for all $\epsilon \in \intoo{0}{1}$ there exists $C_\epsilon >0$ such that
		\begin{align}
		\label{eq: estimate_GDG}
		\forall (\mathbf{r,r'}) \in (\lscr^\mathbf{R})^2,\quad \absL{\left(Q^{-\frac{1}{2}}_L D_L Q^{-\frac{1}{2}}_L\right)(\mathbf{r,r'})} \leq C_\epsilon T_L^{(1-\epsilon)\abs{\mathbf{r-r'}}} \pt
		\end{align}
		Recalling the relations \eqref{eq: A_invariance}, we deduce that $A_L(\mathbf{k};\mathbf{r,r'})$ is given by the Fourier series
		\begin{multline*}
		A_L(\mathbf{k};\mathbf{r,r'}) 
		= \sum_{\mathbf{u} \in \lscr} e^{i\mathbf{k}\cdot L\mathbf{u}} A_L(\mathbf{r},\mathbf{u+r'})
		= -\mu_L \delta(\mathbf{r,r'}) + \sum_{k=1}^{m} \theta_{L,k} \left(\sum_{\mathbf{u}\in\lscr} e^{iL\mathbf{k}\cdot \mathbf{u}} \mathds{1}_{(\mathbf{\mathbf{r},\mathbf{u+r'}}) \in \mathcal{O}_k}\right)\\
		+ \sum_{\mathbf{u} \neq 0} e^{iL\mathbf{k}\cdot \mathbf{u}} (Q_L^{-\frac{1}{2}}D_LQ_L^{-\frac{1}{2}})(\mathbf{r,u+r'}) \mathds{1}_{(\mathbf{\mathbf{r},\mathbf{u+r'}}) \notin \mathscr{P}^\mathbf{R}}\, ,
		\end{multline*}
		where the last term is $\grandO{T_L^{d_1-}}$ in $\mathcal{C}^\infty_\mathrm{per}(\Gamma_L^*)$ thanks to \eqref{eq: estimate_GDG}.
	\end{proof}
	
	\subsection{An energy estimate on \texorpdfstring{$E_L^\perp(\mathbf{k})$}{ELk}}
	
	For $\mathbf{k}\in\Gamma_L^*$, we recall that the $N$-dimensional vector space $E_L(\mathbf{k})$ is given by
	\begin{align*}
	E_L(\mathbf{k}) \coloneqq \spn \left\{u_{L,\mathbf{r}}(\mathbf{k},\cdot)\right\}_{\mathbf{r\in R}} = \spn \left\{\mathcal{U}_\mathrm{BF}(v_{L,\mathbf{r}})(\mathbf{k},\cdot)\right\}_{\mathbf{r\in R}} \pt
	\end{align*}
	The equality comes from Proposition~\ref{prop: properties_u}. We denote by $P_L(\mathbf{k})$ the associated orthogonal projection and $P_L^\perp(\mathbf{k}) \coloneqq 1 - P_L(\mathbf{k})$ the orthogonal projection on $E_L^\perp(\mathbf{k}) \coloneqq E_L(\mathbf{k})^\perp$. Our goal is to show that there is an uniform energy gap on $E^\perp_L(\mathbf{k})$ compared with $E_L(\mathbf{k})$.
	\begin{prop}
		\label{prop: energy_estimate}
		Let $\mathbf{k} \in \Gamma_L^*$. We have the energy estimate: for all $u \in H^1_\mathbf{k}(\Gamma_L) \cap E^\perp_L(\mathbf{k})$, we have
		\begin{align}
		\label{eq: energy_estimate}
		\pdtsc{u}{H_L(\mathbf{k})u}_{L^2(\Gamma_L)} \geq (g_L - \mu_L + \petito{1}) \norme{u}^2_{L^2(\Gamma_L)} = (g - \mu + \petito{1}) \norme{u}^2_{L^2(\Gamma_L)} \, ,
		\end{align}
		where $g_L >0$ and $\mu_L >0$ (resp. $g>0$ and $\mu >0$) denote the first spectral gap and the common lowest eigenvalue of the effective mono-atomic operators $\{H_{L,\mathbf{r}}\}_{\mathbf{r\in}\lscr^\mathbf{R}}$ (resp. reference operator $H$) introduced in \eqref{eq: effective_mono_atomic_operator} (resp. in \eqref{eq: reference_operator}).
	\end{prop}
	\begin{proof}
		First, we notice that the last equality in \eqref{eq: energy_estimate} comes from \eqref{eq: closedness_modified_original_operator}.	Now, we consider an $\lscr$-periodic partition of unity $\sum_{\mathbf{r \in R}} (\eta_\mathbf{r})^2 = 1 $ with $\eta_\mathbf{r} \in W^{1,\infty}_\mathrm{per}(\Gamma_L) \cap \mathcal{C}^\infty(\R^2)$ where 
		\[
		W^{1,\infty}_\mathrm{per}(\Gamma_L) = \enstq{u \in \lper^\infty(\Gamma_L)}{\partial_1 u \in \lper^\infty(\Gamma_L) \et \partial_2 u \in \lper^\infty(\Gamma_L)}\, ,
		\]
		and satisfying the following conditions
		\begin{gather*}
		\exists \rho >0,\quad\forall \mathbf{r} \in \mathbf{R},\quad \exists \xi_\mathbf{r} \in \mathcal{C}^\infty_c(\R^2),\quad \left(\xi_{\mathbf{r}\restreinta B(\mathbf{r},\rho)} \equiv 1 \et \eta_\mathbf{r} = \sum_{\mathbf{u} \in \lscr} \xi_\mathbf{r}(\cdot - \mathbf{u})\right) \pt
		\end{gather*}
		The parameter $\rho>0$ can be chosen as small as we want. Then, by defining $\eta_{L,\mathbf{r}} \coloneqq \eta_\mathbf{r}(L^{-1} \cdot)$ and $\xi_{L,\mathbf{r}} \coloneqq \xi_\mathbf{r}(L^{-1} \cdot)$, we obtain a $\lscr_L$-periodic partition of unity $\sum_{\mathbf{r \in R}} (\eta_{L,\mathbf{r}})^2 = 1 $.
		Let $ u \in H^1_\mathbf{k}(\Gamma_L) \cap E^\perp_L(\mathbf{k})$. Using the periodic IMS localization formula, we have
		\begin{align}
		\label{eq: energy_estimate_1}
		\pdtsc{u}{H_L(\mathbf{k})u}_{L^2(\Gamma_L)} 
		= \sum_{\mathbf{r \in R}}\left( \pdtsc{\eta_{L,\mathbf{r}} u}{H_L(\mathbf{k})\eta_{L,\mathbf{r}} u}_{L^2(\Gamma_L)} - \pdtsc{u}{\absL{\nabla \eta_{L,\mathbf{r}}}^2u}_{L^2(\Gamma_L)}\right) \pt
		\end{align}
		Because $\xi_\mathbf{r}$ is compactly supported into a ball $B(\mathbf{r},R)$ for some constant $R>0$, we have
		\begin{align*}
		\pdtsc{\eta_{L,\mathbf{r}} u}{H_L(\mathbf{k})\eta_{L,\mathbf{r}} u}_{L^2(\Gamma_L)} = \sum_{\abs{\mathbf{u}} \leq 2R} \pdtsc{\xi_{L,\mathbf{r}}(\cdot - L\mathbf{u}) u}{H_L\xi_{L,\mathbf{r}} u}_{L^2(\R^2)} \pt
		\end{align*}
		Inserting this in \eqref{eq: energy_estimate_1}, taking the real part and using $\norme{\nabla \eta_{L,\mathbf{r}}}_{L^\infty} = \grandO{L^{-1}}$, we get
		\begin{multline*}
		\pdtsc{u}{H_L(\mathbf{k})u}_{L^2(\Gamma_L)} 
		= \sum_{\mathbf{r \in R}} \pdtsc{\xi_{L,\mathbf{r}} u}{H_L\xi_{L,\mathbf{r}} u}_{L^2(\R^2)} \\
		+ \sum_{\mathbf{r \in R}}\sum_{0< \abs{\mathbf{u}} \leq 2R} \Re\left(\pdtsc{\xi_{L,\mathbf{r}}(\cdot - L\mathbf{u}) u}{H_L\xi_{L,\mathbf{r}} u}_{L^2(\R^2)}\right) + \grandO{L^{-2} \norme{u}^2_{L^2(\Gamma_L)}}	\pt
		\end{multline*}
		Recalling the following stability inequality (which results from the functional calculus)
		\begin{align*}
		\forall v \in H^1(\R^2),\quad \pdtsc{v}{H_{L,\mathbf{r}} v}_{L^2(\R^2)} \geq (-\mu_L + g_L) \norme{v}^2_{L^2(\R^2)} - g_L \abs{\langle v,v_{L,\mathbf{r}} \rangle_{L^2(\R^2)}}^2 \, ,
		\end{align*}
		we can write
		\begin{align}
		\label{eq: energy_estimate_inter}
		\pdtsc{u}{H_L(\mathbf{k})u}_{L^2(\Gamma_L)}
		\geq &~ (-\mu_L + g_L) \sum_{\mathbf{r\in R}} \norme{\xi_{L,\mathbf{r}} u}^2_{L^2(\R^2)}  - g_L \sum_{\mathbf{r \in R}} \abs{\langle \xi_{L,\mathbf{r}} u,v_{L,\mathbf{r}} \rangle_{L^2(\R^2)}}^2 \\ 
		\notag
		&+ \sum_{\mathbf{r} \in \mathbf{R}} \pdtsc{\xi_{L,\mathbf{r}} u}{V_L(1- \chi_{L,\mathbf{r}})\xi_{L,\mathbf{r}} u}_{L^2(\R^2)} \\
		\notag
		&+\sum_{\mathbf{r \in R}}\sum_{0< \abs{\mathbf{u}} \leq 2R} \Re\left(\pdtsc{\xi_{L,\mathbf{r}}(\cdot - L\mathbf{u}) u}{H_L\xi_{L,\mathbf{r}} u}_{L^2(\R^2)}\right)+ \petito{\norme{u}^2_{L^2(\Gamma_L)}} \pt
		\end{align}
		Expanding the identity $\norme{u}^2_{L^2(\Gamma_L)} = \sum_{\mathbf{r\in R}} \norme{\eta_{L,\mathbf{r}} u}^2_{L^2(\Gamma_L)}  $, we obtain
		\begin{align*}
		\norme{u}^2_{L^2(\Gamma_L)}= \sum_{\mathbf{r\in R}} \norme{\xi_{L,\mathbf{r}} u}^2_{L^2(\R^2)} + \sum_{\mathbf{r\in R}} \sum_{\mathbf{u}\in\lscr\setminus\{0\}} \pdtsc{\xi_{L,\mathbf{r}}(\cdot - L\mathbf{u})}{\xi_{L,\mathbf{r}}u}_{L^2(\R^2)} \geq  \sum_{\mathbf{r\in R}} \norme{\xi_{L,\mathbf{r}} u}^2_{L^2(\R^2)} \pt
		\end{align*}
		Since $-\mu_L + g_L \leq 0$ (recall that $\sigma_\mathrm{ess}(H_{L,\mathbf{r}}) = \intfo{0}{\infty}$), we have
		\begin{align}
		\label{eq: orthogonality_condition_11}
		(-\mu_L + g_L) \sum_{\mathbf{r\in R}} \norme{\xi_{L,\mathbf{r}} u}^2_{L^2(\R^2)} \geq (-\mu_L + g_L )\norme{u}^2_{L^2(\Gamma_L)} \pt
		\end{align}
		Moreover, using that $E_L(\mathbf{k})$ is spanned by the family $\{\mathcal{U}_\mathrm{BF}(v_{L,\mathbf{r}})(\mathbf{k},\cdot)\}_{\mathbf{r\in R}}$ (see Proposition~\ref{prop: properties_u}), the Cauchy-Schwarz inequality and Proposition~\ref{prop: exponential_decay_v_tilde}, we have
		\begin{align*}
		0 
		& = \pdtsc{u}{\mathcal{U}_\mathrm{BF}(v_{L,\mathbf{r}})(\mathbf{k},\cdot)}_{L^2(\Gamma_L)} \\
		& = \pdtsc{ u}{\xi_{L,\mathbf{r}} v_{L,\mathbf{r}}}_{L^2(\Gamma_L)} + \pdtsc{ u}{(1-\xi_{L,\mathbf{r}}) v_{L,\mathbf{r}}}_{L^2(\Gamma_L)} + \sum_{\mathbf{u} \in \lscr \setminus \{0\}} \pdtsc{ u}{v_{L,\mathbf{r}}}_{L^2(\Gamma_L+L\mathbf{u})} \\
		& = \pdtsc{ u}{\xi_{L,\mathbf{r}} v_{L,\mathbf{r}}}_{L^2(\Gamma_L)} + \petito{\norme{u}_{L^2(\Gamma_L)}}\pt
		\end{align*}
		Thus, we have shown
		\begin{align}
		\label{eq: orthogonality_condition_12}
		\pdtsc{ u}{\xi_{L,\mathbf{r}} v_{L,\mathbf{r}}}_{L^2(\Gamma_L)} = \petito{\norme{u}_{L^2(\Gamma_L)}}\pt
		\end{align}
		Notice that $\supp \xi_{L,\mathbf{r}} \subset \left[\cup_{\mathbf{r' \in \lscr^\mathbf{R} \setminus\{\mathbf{r}\}}} B(L\mathbf{r'},L\rho)\right]^c \cap B(L\mathbf{r},LR)$. Then, if we choose $\rho$ small enough in order to have $B(L\mathbf{r},L\rho) \subset \{\chi_{L,\mathbf{r}} \equiv 1\}$, we have
		\begin{align*}
		\absL{\pdtsc{\xi_{L,\mathbf{r}} u}{V_L(1- \chi_{L,\mathbf{r}})\xi_{L,\mathbf{r}} u}_{L^2(\R^2)}}
		\leq \int_{\R^2} \abs{V_L} \abs{u}^2 \mathds{1}_{\supp \xi_{L,\mathbf{r}}\cap \{\chi_{L,\mathbf{r}} \equiv 1\}^c} 
		\lesssim \norme{V_L \mathds{1}_{\Lambda_L}}_{\lper^\infty(\Gamma_L)} \norme{u}^2_{L^2(\Gamma_L)} \, ,
		\end{align*}
		where $\Lambda_L = \left(\cup_{\mathbf{r'} \in \lscr^\mathbf{R}} B(L\mathbf{r'},L\rho)\right)^c$. Assumption~\ref{hypo_04} says that $\norme{V_L \mathds{1}_{\Lambda_L}}_{\lper^\infty(\Gamma_L)} = \petito{1}$. Thus
		\begin{align}
		\label{eq: orthogonality_condition_13}
		\pdtsc{\xi_{L,\mathbf{r}} u}{V_L(1- \chi_{L,\mathbf{r}})\xi_{L,\mathbf{r}} u}_{L^2(\R^2)} = \petito{\norme{u}^2_{L^2(\Gamma_L)}} \pt
		\end{align}
		With the same support arguments, we can show that 
		\begin{align}
		\label{eq: orthogonality_condition_14}
		\forall \mathbf{u} \in \lscr\setminus \{0\},\quad\pdtsc{\xi_{L,\mathbf{r}}(\cdot - L\mathbf{u}) u}{V_L\xi_{L,\mathbf{r}} u}_{L^2(\R^2)} = \petito{\norme{u}^2_{L^2(\Gamma_L)}} \pt
		\end{align}
		To treat the kinetic part $\Re\left(\pdtsc{\xi_{L,\mathbf{r}}(\cdot - L\mathbf{u}) u}{-\Delta\xi_{L,\mathbf{r}} u}_{L^2(\R^2)}\right)$, we use the following lemma.
		\begin{lemma}
			\label{lemma: technical_and_convenient}
			Let $\xi,\xi' \in \mathcal{C}^\infty_c(\R^2)$ smooth, positive-valued and compactly supported functions. Then, for all $u \in H^1_\mathrm{loc}(\R^2,\C)$, we have
			\begin{align*}
			\Re\left(\pdtsc{\xi u}{-\Delta (\xi' u)}_{L^2(\R^2)}\right) \geq -\frac{1}{2} \int_{\R^2} \abs{u}^2 \left(\xi\Delta \xi' + \xi' \Delta \xi \right)\pt
			\end{align*}
		\end{lemma}
		For clarity, we show this lemma after completing the proof of Proposition~\ref{prop: energy_estimate}. By Lemma~\ref{lemma: technical_and_convenient} and the properties of $\xi_{L,\mathbf{r}}$, we can show
		\begin{align}
		\label{eq: orthogonality_condition_15}
		\forall \mathbf{u} \in \lscr\setminus \{0\},\quad\pdtsc{\xi_{L,\mathbf{r}}(\cdot - L\mathbf{u}) u}{-\Delta \xi_{L,\mathbf{r}} u}_{L^2(\R^2)} = \grandO{L^{-2}\norme{u}^2_{L^2(\Gamma_L)}} \pt
		\end{align}	
		We obtain \eqref{eq: energy_estimate} by inserting \eqref{eq: orthogonality_condition_11}, \eqref{eq: orthogonality_condition_12}, \eqref{eq: orthogonality_condition_13}, \eqref{eq: orthogonality_condition_14} and \eqref{eq: orthogonality_condition_15} in \eqref{eq: energy_estimate_inter} and using that $g_L$ is uniformly bounded for $L \geq 1$ (see \eqref{eq: closedness_modified_original_operator}).
	\end{proof}
	\begin{proof}[Proof of Lemma~\ref{lemma: technical_and_convenient}]
		Integrating by parts several times, we get
		\begin{align*}
		\Re\left(\pdtsc{\xi u}{-\Delta (\xi' u)}_{L^2(\R^2)}\right) 
		&= \int \left(\abs{u}^2 \nabla \xi \cdot \nabla \xi' + \Re\left(u \nabla \overline{u}\right) \cdot \xi \nabla \xi' + \Re\left(\overline{u} \nabla u\right) \cdot \xi' \nabla \xi + \xi \xi' \abs{\nabla u}^2\right) \\
		&\geq  \int \abs{u}^2  \nabla \xi \cdot \nabla \xi' + \frac{1}{2} \nabla \abs{u}^2 \cdot \nabla (\xi \xi') \\
		&\geq -\frac{1}{2} \int \abs{u}^2 (\Delta(\xi \xi') - 2 \nabla \xi \cdot \nabla \xi')\\
		&\geq -\frac{1}{2} \int_{\R^2} \abs{u}^2 \left(\xi\Delta \xi' + \xi' \Delta \xi \right)\pt \qedhere
		\end{align*}
	\end{proof}
	
	\subsection{Estimate of the residual term in the Feshbach-Schur method}\label{sec:estimate-of-the-residual-term-in-the-feshbach-schur-method}
	
	In this section, we use the Feshbach-Schur method to express the $N$ lowest eigenvalues of $H_L(\mathbf{k})$ as a perturbation of the spectrum of the $N\times N$ matrix $A_L(\mathbf{k})$.
	
	As a consequence of Proposition~\ref{prop: energy_estimate}, we have for $L$ large enough 
	\begin{align}
	\label{eq: energy_estimate_away}
	\boxed{P_L^\perp(\mathbf{k}) (H_L(\mathbf{k})+ \mu_L) P_L^\perp(\mathbf{k}) \geq  \frac{g}{2} P_L^\perp(\mathbf{k}) \pt}
	\end{align}
	We decompose $\lk^2(\Gamma_L) = E_L(\mathbf{k}) \overset{\perp}{\oplus} E_L^\perp(\mathbf{k})$ and we write
	\[
	H_L(\mathbf{k}) = \begin{pmatrix}
	A_L(\mathbf{k}) & C_L(\mathbf{k})^* \\
	C_L(\mathbf{k}) & B_L(\mathbf{k})
	\end{pmatrix} \, ,
	\]
	with $ A_L(\mathbf{k}) = P_L(\mathbf{k}) H_L(\mathbf{k}) P_L(\mathbf{k})$, $B_L(\mathbf{k}) = P_L^\perp(\mathbf{k}) H_L(\mathbf{k}) P_L^\perp(\mathbf{k})$ and $C_L(\mathbf{k}) = P_L^\perp(\mathbf{k}) H_L(\mathbf{k}) P_L(\mathbf{k})$. The Feshbach-Schur method~\cite[Chapter 11]{gustafson2020mathematical} says that for $\lambda\leq -\mu_L+g/3$, we have
	\begin{align}
	\label{eq: FS_method}
	\boxed{\lambda \in \sigma(H_L(\mathbf{k}))
		\Longleftrightarrow \lambda \in \sigma\left(A_L(\mathbf{k}) - C_L(\mathbf{k})^*(B_L(\mathbf{k})-\lambda)^{-1} C_L(\mathbf{k}) \right)\pt}
	\end{align}
	In the right hand side appears an $N \times N$ hermitian matrix. Therefore, \eqref{eq: FS_method} implies that $H_L(\mathbf{k})$ has exactly $N$ eigenvalues (counted with multiplicity) in the interval $\intoo{-\infty}{-\mu_L + g/3}$. Corollary~\ref{cor: expansion_A_k} gives the expansion for $A_L(\mathbf{k})$ when $L \to \infty$. The next proposition bounds the residual term $C_L(\mathbf{k})^*(B_L(\mathbf{k})-\lambda)^{-1} C_L(\mathbf{k})$ in operator norm.
	\begin{prop}
		\label{prop: estimation_key}
		For $L$ large enough, for all $\mathbf{k} \in \Gamma_L^*$ and for all $\lambda \in \intoo{-\infty}{-\mu_L+ g/3}$, we have the estimate
		\begin{align*}
		\boxed{\normeL{C_L(\mathbf{k})^*(B_L(\mathbf{k})-\lambda)^{-1} C_L(\mathbf{k})} = \grandO{T_L^{(1+\delta)d_0-}}\, ,}
		\end{align*}
		where the $O$ is independent from $\mathbf{k}$.
	\end{prop}
	
	We divide the proof of Proposition~\ref{prop: estimation_key} in several lemmas.
	\begin{lemma} 
		\label{lemma: estimation_key1}
		For $L$ large enough, for all $\mathbf{k} \in \Gamma_L^*$, for all $\lambda \in \intoo{-\infty}{-\mu_L+ g/3}$, for all $\varphi \in H^1_\mathbf{k}(\Gamma_L) \cap E_L^\perp(\mathbf{k})$, we have
		\begin{align}
		\label{eq: technical_boundedness}
		\normeL{ \frac{1}{\sqrt{B_L(\mathbf{k}) - \lambda}}\sqrt{P_L^\perp(\mathbf{k}) (-\Delta + 1)  P_L^\perp(\mathbf{k})} \varphi}_{L^2(\Gamma_L)} \leq \sqrt{\frac{12}{g} \max\left(1,2C_L+\frac{g}{3}\right)} \norme{\varphi}_{L^2(\Gamma_L)} \, ,
		\end{align}
		where $C_L = C(1/2,L)$ is the constant appearing in the estimate \eqref{eq: kato_type_inequality2} of Proposition~\ref{lemma: kato_type_inequality}.
	\end{lemma}
	\begin{proof}
		As a consequence of the estimate \eqref{eq: kato_type_inequality2} of Proposition~\ref{lemma: kato_type_inequality}, there exists $C_L >0$ such that
		\begin{align}
		\notag
		B_L(\mathbf{k})-\lambda=P_L^\perp(\mathbf{k})(H_L(\mathbf{k}) - \lambda)P_L^\perp(\mathbf{k}) 
		& \geq P_L^\perp(\mathbf{k})\left[-\frac{1}{2} \Delta - (C_L+\lambda) \right]P_L^\perp(\mathbf{k})\\
		\label{eq: technical_boundedness_1}
		&\geq P_L^\perp(\mathbf{k})\left[- \frac{1}{2}\Delta - C_L \right]P_L^\perp(\mathbf{k}) \pt
		\end{align}
		The second inequality comes from the fact that for $L$ large enough, we have $\intoo{-\infty}{-\mu_L + g/3} \subset \intoo{-\infty}{0}$. By \eqref{eq: energy_estimate_away}, we have 
		\begin{align}
		\label{eq: technical_boundedness_2}
		B_L(\mathbf{k})-\lambda \geq  \frac{g}{6} P_L^\perp(\mathbf{k}) \, ,
		\end{align}
		for all $\lambda < -\mu_L + g/3$. By taking the convex combination $(1-t)\times \eqref{eq: technical_boundedness_1}+ t\times\eqref{eq: technical_boundedness_2}$ with $t=\frac{g}{2g+12C_L}$, we obtain
		\begin{align*}
		B_L(\mathbf{k}) - \lambda \geq \frac{g}{12} \min\left(1, \frac{1}{2C_L+g/3}\right) P_L^\perp(\mathbf{k}) (-\Delta + 1)  P_L^\perp(\mathbf{k}) \pt
		\end{align*}
		From this previous estimate and~\cite[Corollary 10.12]{schmugden2012unbounded}, we have 
		\begin{align*}
		(B_L(\mathbf{k}) - \lambda)^{-1} \leq \frac{12}{g} \max\left(1,2C_L+\frac{g}{3}\right) P_L^\perp(\mathbf{k})(-\Delta + 1)^{-1}P_L^\perp(\mathbf{k}) \pt
		\end{align*}
		The estimate \eqref{eq: technical_boundedness} results from this inequality.
	\end{proof}
	The following technical lemma is convenient to estimate convolutions when only integral bounds are at our disposal.
	\begin{lemma}
		\label{lemma: estimation_key2}
		Let $\varphi,\psi\in L^2(\R^2)$ such that $\norme{\varphi \mathds{1}_{\abs{x}\geq R}}_{L^2(\R^2)} \leq C e^{-\alpha R}$ and $\norme{\psi \mathds{1}_{\abs{x}\geq R}}_{L^2(\R^2)} \leq C e^{-\beta R}$ for some constants $C,\alpha,\beta \geq  0$ and for all $R\geq0$. Then, we have
		\begin{align*}
		\forall \mathbf{x} \in \R^2,\quad \absL{\left(\varphi \ast \psi \right)(\mathbf{x})} \leq C\left( \norme{\varphi}_{L^2(\R^2)} e^{-\frac{\beta}{2}\abs{\mathbf{x}}}+  \norme{\psi}_{L^2(\R^2)} e^{-\frac{\alpha}{2}\abs{\mathbf{x}}}\right) \pt
		\end{align*}
	\end{lemma}
	\begin{proof}
		We cannot have $\abs{\mathbf{y}} < \abs{\mathbf{x}}/2$ and $\abs{\mathbf{y-x}} < \abs{\mathbf{x}}/2$ simultaneously since it would imply $\abs{\mathbf{x}} < \abs{\mathbf{x}}$.
		Hence, we can write
		\begin{align*}
		\absL{\left(\varphi \ast \psi \right)(\mathbf{x})} 
		\leq \int_{\abs{\mathbf{x-y}} \geq \frac{\abs{\mathbf{x}}}{2}} \absL{\varphi(\mathbf{x-y})} \absL{\psi(\mathbf{y})} \diff \mathbf{y} + \int_{\abs{\mathbf{y}} \geq \frac{\abs{\mathbf{x}}}{2}} \absL{\varphi(\mathbf{x-y})} \absL{\psi(\mathbf{y})} \diff \mathbf{y} \, ,
		\end{align*}
		and we conclude with the Cauchy-Schwarz inequality.
	\end{proof}

	For $\mathbf{v}\in\lscr_L^*$ and $\mathbf{k}\in\Gamma_L^*$, we denote $e_{\mathbf{v+k}}(\mathbf{x}) \coloneqq \abs{\Gamma_L}^{-\frac{1}{2}} e^{i \mathbf{x \cdot (v+k)}} \in \lk^2(\Gamma_L)$. We recall that $e_{\mathbf{v+k}}$ is a normalized eigenfunction of the Laplace operator $-\Delta$ defined on $\lk^2(\Gamma_L)$ associated with the eigenvalue $\abs{\mathbf{v+k}}^2$ and that the set $(e_{\mathbf{v+k}})_{\mathbf{v} \in \lscr_L^*}$ forms an orthonormal family of $\lk^2(\Gamma_L)$.
	\begin{lemma}
		\label{lemma: estimation_key3}
		For all $\mathbf{r\in R}$, we have
		\begin{align}
		\label{eq: uniform_boundedness_H1_UBF}
		\sup_{L\geq 1} \sup_{\mathbf{k}\in\Gamma_L^*}\normeL{\sqrt{-\Delta+1}~\mathcal{U}_{\mathrm{BF}}(v_{L,\mathbf{r}})(\mathbf{k},\cdot)}_{L^2(\Gamma_L)} < \infty \pt
		\end{align}
	\end{lemma}
	\begin{proof}
		The parity operator $\mathcal{P}$ is defined by $(\mathcal{P}v)(\mathbf{x}) = v(-\mathbf{x})$. We have
		\begin{align*}
		\normeL{\sqrt{-\Delta+1}~\mathcal{U}_{\mathrm{BF}}(v_{L,\mathbf{r}})(\mathbf{k},\cdot)}_{L^2(\Gamma_L)}^2 
		& = \sum_{\mathbf{v} \in \lscr^*_L} \left(\abs{\mathbf{v+k}}^2+1\right) \absL{\pdtsc{e_{\mathbf{v+k}}}{\mathcal{U}_{\mathrm{BF}}(v_{L,\mathbf{r}})(\mathbf{k},\cdot)}_{L^2(\Gamma_L)}}^2 \\
		& = \frac{(2\pi)^2}{\abs{\Gamma_L}}\sum_{\mathbf{v} \in \lscr^*_L} \left(\abs{\mathbf{v+k}}^2+1\right) \absL{\mathcal{F}(v_{L,\mathbf{r}})(\mathbf{v+k})}^2 \\
		& = \frac{2\pi}{\abs{\Gamma_L}}\sum_{\mathbf{v} \in \lscr^*_L} \mathcal{F}\left[\left(\sqrt{-\Delta+1}\mathcal{P}v_{L,\mathbf{r}}\right)\ast \left(\sqrt{-\Delta+1}v_{L,\mathbf{r}}\right)\right] (\mathbf{v+k})\\
		& = \sum_{\mathbf{u} \in \lscr_L} e^{-i\mathbf{u\cdot k}} \left[\left(\sqrt{-\Delta+1}\mathcal{P}v_{L,\mathbf{r}}\right)\ast \left(\sqrt{-\Delta+1}v_{L,\mathbf{r}}\right)\right](\mathbf{u})  \pt
		\end{align*}
		The equality $\pdtsc{e_{\mathbf{v+k}}}{\mathcal{U}_{\mathrm{BF}}(v_{L,\mathbf{r}})(\mathbf{k},\cdot)}_{L^2(\Gamma_L)} = 2\pi\abs{\Gamma_L}^{-1/2}\mathcal{F}(v_{L,\mathbf{r}})(\mathbf{v+k})$ holds because $v_{L,\mathbf{r}}$ decays exponentially fast by Proposition~\ref{prop: exponential_decay_v_tilde}. We recall that if $f,g \in L^2(\R^d)$ satisfy $f \ast g \in L^1(\R^d)$ then we have $ \mathcal{F}(f \ast g) = (2\pi)^{d/2}\widehat{f} \widehat{g}$. By Proposition~\ref{prop: exponential_decay_v_tilde} and Lemma~\ref{lemma: estimation_key2}, we have $\left(\sqrt{-\Delta+1}v_{L,\mathbf{r}}\right)\ast \left(\sqrt{-\Delta+1}v_{L,\mathbf{r}}\right) \in L^1(\R^2)$ which justifies the third equality in the computation above. The last one is justified with Lemma~\ref{lemma: estimation_key2} and Poisson's summation formula, which we recall (see for instance~\cite[Theorem 3.2.8]{grafakos2014classical}): for $f\in\mathcal{C}(\R^2)$ which satisfies, for some $C,\delta >0$, the conditions $\abs{f(\mathbf{x})} \leq C(1+\abs{\mathbf{x}})^{-2-\delta}$ and $\sum_{\mathbf{v} \in \lscr^*_L} |\widehat{f}(\mathbf{v})| < \infty$ then we have
		\begin{align}
		\label{prop: poisson_summation_formula}
		\forall \mathbf{x} \in \R^2,\quad \frac{2\pi}{\abs{\Gamma_L}}\sum_{\mathbf{v} \in \lscr^*_L} \widehat{f}(\mathbf{v}) e^{i\mathbf{x \cdot v}} = \sum_{\mathbf{u} \in \lscr_L} f(\mathbf{u+x}) \pt
		\end{align}
		Then, using Lemma~\ref{lemma: estimation_key2} and the estimate \eqref{eq: exponential_decay_v_tilde_integrale} from Proposition~\ref{prop: exponential_decay_v_tilde}, we can bound
		\begin{align*}
		\normeL{\sqrt{-\Delta+1}~\mathcal{U}_{\mathrm{BF}}(v_{L,\mathbf{r}})(\mathbf{k},\cdot)}_{L^2(\Gamma_L)}^2 
		\lesssim \norme{v_{L,\mathbf{r}}}_{H^1(\R^2)} \pt
		\end{align*}
		Finally, we show the estimate \eqref{eq: uniform_boundedness_H1_UBF} by recalling that $v_{L,\mathbf{r}}$ is uniformly bounded in $H^1(\R^2)$ (see for instance the right side of \eqref{eq: closedness_modified_original_operator}).
	\end{proof}
	\begin{lemma}
		\label{lemma: estimation_key4}
		We have
		\[
		\sup_{L\geq 1} \sup_{\mathbf{k}\in\Gamma_L^*} \normeL{\frac{1}{\sqrt{-\Delta+\mu_L}} P_L^\perp(\mathbf{k}) \sqrt{-\Delta+\mu_L}}_{\mathcal{B}(\lk^2(\Gamma_L))} < \infty \pt
		\]
	\end{lemma}
	\begin{proof}
		Using the identity $P_L^\perp(\mathbf{k}) = \id-P_L(\mathbf{k})$ and $\sup_{L\geq 1} \mu_L < \infty$, it is sufficient to show that the operator $P_L(\mathbf{k}) \sqrt{-\Delta+1}$ is bounded uniformly with respect to $L\geq 1$ and $\mathbf{k}\in\Gamma_L^*$. Writing $P_L(\mathbf{k}) = \sum_{\mathbf{r \in R}} \ket{u_{L,\mathbf{r}}(\mathbf{k},\cdot)} \bra{u_{L,\mathbf{r}}(\mathbf{k},\cdot)}$ (see Proposition~\ref{prop: properties_u}), we can show that for all $\varphi \in H^1_\mathbf{k}(\Gamma_L)$ we have
		\begin{align*}
		\normeL{P_L(\mathbf{k}) \sqrt{-\Delta+1} \varphi}_{L^2(\Gamma_L)} \leq \left(\sum_{\mathbf{r \in R}} \normeL{\sqrt{-\Delta+1} u_{L,\mathbf{r}}(\mathbf{k},\cdot)}_{L^2(\Gamma_L)}\right) \norme{\varphi}_{L^2(\Gamma_L)} \pt
		\end{align*}
		Then, using identity \eqref{eq: explicit_decomposition}, we can bound
		\begin{align*}
		\sum_{\mathbf{r \in R}}\normeL{\sqrt{-\Delta+1} u_{L,\mathbf{r}}(\mathbf{k},\cdot)}_{L^2(\Gamma_L)} 
		\leq \left(\sup_\mathbf{r \in R} \sum_{\mathbf{u} \in \lscr^\mathbf{R}} \absL{G^{-\frac{1}{2}}_L(\mathbf{u,r})}\right)\sum_{\mathbf{r \in R}} \normeL{\sqrt{-\Delta+1}~\mathcal{U}_{\mathrm{BF}}(v_{L,\mathbf{r}})(\mathbf{k},\cdot)}_{L^2(\Gamma_L)} \pt
		\end{align*}
		We conclude thanks to Lemma~\ref{lemma: localization_inverse_G} and Lemma~\ref{lemma: estimation_key3}.
	\end{proof}
	\begin{lemma}
		\label{lemma: estimation_key7}
		For $L$ large enough, we have: for all $\mathbf{r \in R}$, for all $\mathbf{k} \in \Gamma_L^*$,
		\begin{align*}
		\normeL{\frac{1}{\sqrt{-\Delta+\mu_L}} \left(H_L(\mathbf{k})+\mu_L\right) \mathcal{U}_\mathrm{BF}(v_{L,\mathbf{r}})(\mathbf{k},\cdot)}_{L^2(\Gamma_L)}
		\lesssim L^{M} \normeL{\sqrt{-\Delta+1}~\mathcal{U}_\mathrm{BF}\left((1-\chi_{L,\mathbf{r}})v_{L,\mathbf{r}}\right)(\mathbf{k},\cdot)}_{L^2(\Gamma_L)} \, ,
		\end{align*}
		where $M$ is the constant appearing in~\eqref{rem: grow_polynomial}.
	\end{lemma}
	\begin{proof}
		We have
		\begin{multline*}
		\normeL{\frac{1}{\sqrt{-\Delta+\mu_L}}\left( H_L(\mathbf{k})+\mu_L\right) \mathcal{U}_\mathrm{BF}(v_{L,\mathbf{r}})(\mathbf{k},\cdot)}^2_{L^2(\Gamma_L)}\\
		= \normeL{\left(\sqrt{-\Delta+\mu_L} + \frac{1}{\sqrt{-\Delta+\mu_L}}V_L\right)\mathcal{U}_\mathrm{BF}(v_{L,\mathbf{r}})(\mathbf{k},\cdot)}^2_{L^2(\Gamma_L)} \pt
		\end{multline*}
		The function $\mathcal{U}_\mathrm{BF}(v_{L,\mathbf{r}})(\mathbf{k},\cdot)$ is in $H^1_\mathbf{k}(\Gamma_L)$ by Lemma~\ref{lemma: estimation_key3} and the operator $\frac{1}{\sqrt{-\Delta+1}}V_L\frac{1}{\sqrt{-\Delta+1}}$ is bounded on $\lk^2(\Gamma_L)$ by Corollary~\ref{cor: kato_type_inequality_consequence}, so all the terms in the previous equality make sense in $\lk^2(\Gamma_L)$. Let $\mathbf{v} \in \lscr_L^*$. Because $v_{L,\mathbf{r}}$ is exponentially decaying, we have
		\begin{align*}
		\pdtsc{e_{\mathbf{v+k}}}{\sqrt{-\Delta+\mu_L}~ \mathcal{U}_\mathrm{BF}(v_{L,\mathbf{r}})(\mathbf{k},\cdot)}_{L^2(\Gamma_L)}
		&=\frac{2\pi}{\sqrt{\abs{\Gamma_L}}}\sqrt{\abs{\mathbf{v+k}}^2+\mu_L} \mathcal{F}(v_{L,\mathbf{r}})(\mathbf{v+k}) \\
		&=\frac{2\pi}{\sqrt{\abs{\Gamma_L}}} \mathcal{F}\left(\sqrt{-\Delta+\mu_L}v_{L,\mathbf{r}}\right)(\mathbf{v+k})  \\
		&=-\mathcal{F}\left(\frac{1}{\sqrt{-\Delta+\mu_L}} V_{L,\mathbf{r}} v_{L,\mathbf{r}}\right)(\mathbf{v+k}) \pt
		\end{align*}
		From the inequality
		\begin{align*}
		\forall \varphi\in\mathcal{S}(\R^2),\quad \normeL{\frac{1}{\sqrt{-\Delta+\mu_L}} \varphi}_{L^p(\R^2)} \leq \frac{2\pi}{\sqrt{\mu_L}} \norme{\varphi}_{L^p(\R^2)} \, ,
		\end{align*}	
		we see that the operator $(-\Delta+\mu_L)^{-1/2}$ is a $L^p$ Fourier multiplier.
		We have $V_{L,\mathbf{r}} v_{L,\mathbf{r}} \in L^p(\R^2)$ since $v_{L,\mathbf{r}} \in  L^\infty(\R^2)$. This implies
		\begin{align*}
		\mathcal{F}\left(\frac{1}{\sqrt{-\Delta+\mu_L}} V_{L,\mathbf{r}} v_{L,\mathbf{r}}\right)(\mathbf{v+k})
		= \frac{1}{\sqrt{\abs{\mathbf{v+k}}+\mu_L}} \mathcal{F}\left(V_{L,\mathbf{r}} v_{L,\mathbf{r}}\right)(\mathbf{v+k}) \pt
		\end{align*}
		By~\cite[Theorem 4.3.7]{grafakos2014classical}, the operator $(-\Delta+\mu_L)^{-1/2}$ is also a $L^p$ multiplier on the torus $\R^2 / \lscr_L $. Then, we have
		\begin{align*}
		\pdtsc{e_{\mathbf{v+k}}}{\frac{1}{\sqrt{-\Delta+\mu_L}}V_L \mathcal{U}_\mathrm{BF}(v_{L,\mathbf{r}})(\mathbf{k},\cdot)}_{L^2(\Gamma_L)}
		= \frac{2\pi}{\sqrt{\abs{\Gamma_L}}} \frac{1}{\sqrt{\abs{\mathbf{v+k}}^2+\mu_L}} \mathcal{F}\left(V_L v_{L,\mathbf{r}}\right)(\mathbf{v+k}) \pt
		\end{align*}
		The second equality comes from the fact that $v_{L,\mathbf{r}}$ is exponentially decaying at infinity. Recalling that $\mu_L \to \mu$ as $L\to\infty$ and using the Plancherel theorem, we show
		\begin{multline*}
		\normeL{\frac{1}{\sqrt{-\Delta+\mu_L}} \left(H_L(\mathbf{k}) +\mu_L\right)\mathcal{U}_\mathrm{BF}(v_{L,\mathbf{r}})(\mathbf{k},\cdot)}^2_{L^2(\Gamma_L)} \\
		\lesssim \frac{(2\pi)^2}{\abs{\Gamma_L}} \sum_{\mathbf{v} \in \lscr_L^*} \frac{1}{\abs{\mathbf{v+k}}^2+1} \absL{\mathcal{F}\left(V_L(1-\chi_{L,\mathbf{r}}) v_{L,\mathbf{r}}\right)(\mathbf{v+k})}^2 \pt
		\end{multline*}
		Going back up the previous calculations, we have
		\begin{multline*}
		\frac{1}{\sqrt{\abs{\Gamma_L}}}\frac{1}{\sqrt{\abs{\mathbf{v+k}}^2+1}}\mathcal{F}\left(V_L(1-\chi_{L,\mathbf{r}}) v_{L,\mathbf{r}}\right)(\mathbf{v+k})\\
		= \pdtsc{e_{\mathbf{v+k}}}{\frac{1}{\sqrt{-\Delta+1}}V_L \mathcal{U}_\mathrm{BF}\left((1- \chi_{L,\mathbf{r}})v_{L,\mathbf{r}}\right)(\mathbf{k},\cdot)}_{L^2(\Gamma_L)} \pt
		\end{multline*}
		Then, by the Plancherel theorem, we obtain 
		\begin{align*}
		\frac{(2\pi)^2}{\abs{\Gamma_L}}\sum_{\mathbf{v} \in \lscr_L^*} \frac{1}{\abs{\mathbf{v+k}}^2+1} &\absL{\mathcal{F}\left(V_L(1-\chi_{L,\mathbf{r}}) v_{L,\mathbf{r}}\right)(\mathbf{v+k})}^2\\
		&= \normeL{\frac{1}{\sqrt{-\Delta+1}}V_L \mathcal{U}_\mathrm{BF}\left((1- \chi_{L,\mathbf{r}})v_{L,\mathbf{r}}\right)(\mathbf{k},\cdot)}_{L^2(\Gamma_L)}^2\\
		&\lesssim L^{2M} \normeL{\sqrt{-\Delta+1}~\mathcal{U}_\mathrm{BF}\left((1-\chi_{L,\mathbf{r}})v_{L,\mathbf{r}}\right)(\mathbf{k},\cdot)}^2_{L^2(\Gamma_L)} \pt
		\end{align*}
		For the last inequality, we have used the boundedness of the operator $(-\Delta+1)^{-1/2}V_L(-\Delta+1)^{-1/2}$ on $\lk^2(\Gamma_L)$, see Corollary~\ref{cor: kato_type_inequality_consequence}.
	\end{proof}
	\begin{lemma}
		\label{lemma: estimation_key8}
		For $L$ large enough we have: for all $\mathbf{r\in R}$, for all $\mathbf{k} \in \Gamma_L^*$,
		\begin{align*}
		\normeL{\sqrt{-\Delta+1}~\mathcal{U}_\mathrm{BF}\left((1-\chi_{L,\mathbf{r}})v_{L,\mathbf{r}}\right)(\mathbf{k},\cdot)}^2_{L^2(\Gamma_L)}
		= \grandO{T_L^{(1+\delta)d_0-}} \, ,
		\end{align*}
		where the $O$ is independent from $\mathbf{r}$ or $\mathbf{k}$.
	\end{lemma}
	\begin{proof}
		We denote $\psi = (1-\chi_{L,\mathbf{r}}) v_{L,\mathbf{r}}$. Following the proof of Lemma~\ref{lemma: estimation_key3}, we get
		\begin{align*}
		\normeL{\sqrt{-\Delta+1}\left(\mathcal{U}_\mathrm{BF}\psi\right)(\mathbf{k},\cdot)}^2_{L^2(\Gamma_L)}
		= \sum_{\mathbf{u} \in \lscr_L} e^{-i\mathbf{u\cdot k}} \left[\left(\sqrt{-\Delta+1}\mathcal{P}\psi\right)\ast \left(\sqrt{-\Delta+1}\psi \right)\right](\mathbf{u})  \, ,
		\end{align*}
		where $\mathcal{P}$ is the parity operator.
		Clearly, $\psi$ and $\mathcal{P}\psi$ satisfy the same integral exponential bound~\eqref{eq: exponential_decay_v_tilde_integrale} from Proposition~\ref{prop: exponential_decay_v_tilde} as $v_{L,\mathbf{r}}$. Therefore, by Lemma~\ref{lemma: estimation_key2}, we have for $L$ large enough
		\begin{align*}
		\absL{\left[\left(\sqrt{-\Delta+1}\mathcal{P}\psi\right)\ast \left(\sqrt{-\Delta+1}\psi \right)\right](\mathbf{u})}
		\lesssim \norme{\psi}_{H^1(\R^2)} e^{-\frac{\mu}{4} \abs{\mathbf{u}}} \pt
		\end{align*}
		Also, by Young's inequality, we have 
		\begin{align*}
		\normeL{\left[\left(\sqrt{-\Delta+1}\mathcal{P}\psi\right)\ast \left(\sqrt{-\Delta+1}\psi \right)\right]}_{L^\infty(\R^2)} \leq \norme{\psi}^2_{H^1(\R^2)} \pt
		\end{align*}
		Hence, for all $\epsilon >0$ and for all $L\geq 1$, we obtain
		\begin{align*}
		\normeL{\sqrt{-\Delta+1}\left(\mathcal{U}_\mathrm{BF}\psi\right)(\mathbf{k},\cdot)}^2_{L^2(\Gamma_L)}
		\lesssim \norme{\psi}^{2-\epsilon}_{H^1(\R^2)} \sum_{\mathbf{u} \in \lscr} e^{-\frac{\mu \epsilon L\abs{\mathbf{u}}}{4} } 
		\lesssim \norme{\psi}^{2-\epsilon}_{H^1(\R^2)}  \underset{=C_\epsilon <\infty}{\underbrace{\sum_{\mathbf{u} \in \lscr} e^{-\frac{\mu \epsilon \abs{\mathbf{u}}}{4}}}} \pt
		\end{align*}
		We recall that $\{\chi_{L,\mathbf{r}} \equiv 1\} \subset B\left(L\mathbf{r}, \frac{1+\delta}{2} Ld_0\right)$. So, by Proposition~\ref{prop: exponential_decay_v_tilde}, we have
		\begin{align*}
		\norme{\psi}^2_{H^1(\R^2)} = \grandO{T_L^{(1+\delta)d_0-}} \pt
		\end{align*}
		This ends the proof of Lemma~\ref{lemma: estimation_key8}.
	\end{proof}
	\begin{proof}[Proof of Proposition~\ref{prop: estimation_key}]
		We notice that
		\begin{align*}
		\normeL{C_L(\mathbf{k})^*(B_L(\mathbf{k})-\lambda)^{-1} C_L(\mathbf{k})} \leq \normeL{(B_L(\mathbf{k})-\lambda)^{-1/2} C_L(\mathbf{k})}^2 \pt
		\end{align*}
		Then, we write
		\begin{align*}
		(B_L(\mathbf{k})-\lambda)^{-1/2} C_L(\mathbf{k}) = (B_L(\mathbf{k})-\lambda)^{-1/2} P_L^\perp(\mathbf{k})\left(H_L(\mathbf{k}) + \mu_L\right) P_L(\mathbf{k}) \pt
		\end{align*}
		We use successively Lemma~\ref{lemma: estimation_key1}, Lemma~\ref{lemma: estimation_key4}, Lemma~\ref{lemma: estimation_key7}, Lemma~\ref{lemma: estimation_key8} and the fact that $E_L(\mathbf{k})$ is spanned by the family $\{\mathcal{U}_\mathrm{BF}(v_{L,\mathbf{r}})(\mathbf{k,\cdot})\}_{\mathbf{r\in R}}$ by Proposition~\ref{prop: properties_u}. The polynomial terms emerging from the norm of $(-\Delta+1)^{-1/2}V_L(-\Delta+1)^{-1/2}$ are absorbed by slightly modifying the $\epsilon$ in the definition of $\grandO{T_L^{(1+\delta)d_0-}}$.
	\end{proof}
	We can conclude this section with the
	\begin{proof}[Proof of Theorem~\ref{theo: feshbach-schur}]
		Theorem~\ref{theo: feshbach-schur} is a consequence of Corollary~\ref{cor: expansion_A_k}, statement \eqref{eq: FS_method} and Proposition~\ref{prop: estimation_key}.
	\end{proof}
	
	\section{Proof of Theorem~\ref{theo: existence_dirac_cones}}\label{sec:example-of-the-honeycomb-lattice}
	
	\subsection{The triangular and honeycomb lattices}\label{sec:the-triangular-and-the-hexagonal-lattices}
	
	First, we recall some basic geometric features of the triangular and honeycomb lattices, see Figure~\ref{fig:honeycomb_lattice_BZ}. The \emph{triangular lattice }is the Bravais lattice defined as the set of discrete translations
	\begin{align*}
	\mathscr{L}\coloneqq \Z \mathbf{u}_1 \oplus \Z \mathbf{u}_2 \quad\text{where}\quad \mathbf{u}_1 = \begin{pmatrix}
	\sqrt{3}/2\\ 
	1/2
	\end{pmatrix} \et \mathbf{u}_2 = \begin{pmatrix}
	\sqrt{3}/2 \\ 
	-1/2
	\end{pmatrix} \pt
	\end{align*}
	We denote by $\Gamma$ its Wigner-Seitz cell which is a regular hexagon. The reciprocal lattice $\mathscr{L}^*$ of $\mathscr{L}$ is given by
	\begin{align*}
	\mathscr{L}^*\coloneqq \Z \mathbf{v}_1 \oplus \Z \mathbf{v}_2 \quad\text{where} \quad \mathbf{v}_1 = \frac{4\pi}{\sqrt{3}}\begin{pmatrix}
	1/2\\ 
	\sqrt{3}/2
	\end{pmatrix} \et \mathbf{v}_2 = \frac{4\pi}{\sqrt{3}}\begin{pmatrix}
	1/2 \\ 
	-\sqrt{3}/2
	\end{pmatrix} \pt
	\end{align*}
	Notice the normalizations $\abs{\mathbf{u}_i} = 1$ and $\abs{\mathbf{v}_i} = \frac{4\pi}{\sqrt{3}}$.
	The first Brillouin zone $\Gamma^*$ is also a regular hexagon. Its vertices are of two types, $\mathbf{K}$ and $\mathbf{K}'$, depending on their orbit under the rotation by $2\pi/3$ about the origin. We use the following conventions
	\begin{align*}
	\mathbf{K} = \frac{1}{3} (\mathbf{v}_1 - \mathbf{v}_2) \et \mathbf{K}' = \frac{1}{3} (\mathbf{v}_2 - \mathbf{v}_1) = - \mathbf{K} \pt
	\end{align*}	
	A generic vertex is denoted by $\mathbf{K}_\star \in \{\mathbf{K},\mathbf{K}'\}$. Notice that for any vertex $\mathbf{K}_\star$, we have $M_\mathcal{R}\mathbf{K}_\star \in \mathbf{K}_\star + \lscr^*$ where $M_\mathcal{R}$ is the rotation matrix by $2\pi/3$. The \emph{honeycomb lattice }is defined as
	\begin{align*}
	\lscr^H \coloneqq \left(\lscr + \mathbf{a}\right) \cup \left(\lscr+\mathbf{b}\right) \ou
	\mathbf{a} = \frac{1}{2\sqrt{3}} \begin{pmatrix}
	1\\0
	\end{pmatrix} \et \mathbf{b}=-\mathbf{a} \pt
	\end{align*}
	The nearest neighbor distance in $\lscr^H$ is $d_0 =\abs{\mathbf{a-b}} = 1/\sqrt{3}$ and the second nearest neighbor distance is $d_1=\abs{\mathbf{u}_i}=1$.
	As in the previous sections, for $L\geq 1$, we use the subscript $L$ to denote the objects defined from the dilated lattice $\mathscr{L}_L = L \mathscr{L}$.
	
	The symmetry group $G$ of the honeycomb lattice $\mathscr{L}^H$ belongs to the \textbf{p6m} class. In particular, $\lscr^H$ is invariant with respect to the shifts of the triangular lattice $\mathscr{L}$, under parity symmetry and horizontal reflection symmetry. Also, $\mathscr{L}_H$ is invariant with respect to the rotation by $2\pi/3$ (resp. by $\pi/3$) about $\mathbf{a}$ or $\mathbf{b}$ (resp. about $\mathbf{x}_c$ where $\mathbf{x}_c \in \R^2$ denotes the center of any hexagon).
	
	Notice that because $\mathbf{b} = \mathbf{-a}$, $\lscr^H$ satisfies Assumption~\ref{hypo_1}: $G$ acts transitively on $\lscr^H$. In addition, it can be seen that the action of $G$ on the set of nearest neighbors $\mathscr{P}^\mathbf{R}$, defined in \eqref{eq: definition_NN}, is also transitive. Hence $m=1$ in this case.
	
	\subsection{Existence of Dirac cones}
	
	In this section, we consider a potential $V_L$ satisfying Assumptions~\ref{hypo_0}, \ref{hypo_2}, \ref{hypo_3} and \ref{hypo_4} with $\lscr^\mathbf{R} = \lscr^H$, which satisfies Assumption~\ref{hypo_1}. We use the same notations as in Section~\ref{sec:two-dimensional-lattices-at-dissociation-with-coulomb-singularities}. In particular, we denote by $H_L = -\Delta+V_L$ the periodic Schrödinger operator associated with $V_L$ and by $H_L(\mathbf{k})$ the restriction of $H_L$ along the fiber $L^2_\mathbf{k}(\Gamma)$. Also the matrix of $H_L$ (resp. $H_L(\mathbf{k})$) in the subspace $E_L$ (resp. $E_L(\mathbf{k})$) spanned by the orthonormal family $\{w_{L,\mathbf{r}}\}_{\mathbf{r}\in\lscr^\mathbf{R}}$ (resp. $\{u_{L,\mathbf{r}}(\mathbf{k,\cdot})\}_{\mathbf{r}\in\lscr^\mathbf{R}}$) is denoted by $A_L$ (resp. $A_L(\mathbf{k})$).
	
	By Theorem~\ref{theo: feshbach-schur}, we know that for $L$ large enough, the spectrum $H_L(\mathbf{k})$ is given to leading order by the Wallace model:
	\begin{align*}
	\forall \mathbf{k} \in \Gamma_L^*,\quad \mu_{\pm,L}(\mathbf{k}) = - \mu_L \pm \abs{ \theta_L} \absL{ 1+e^{i\mathbf{k\cdot u}_1}+e^{i\mathbf{k\cdot u}_2}} +\grandO{T_L^{\frac{1+\delta}{\sqrt{3}}-}+T_L^{1-}} \, ,
	\end{align*}
	where $\theta_L$ is the interaction coefficient defined in \eqref{eq: interaction_coefficient}. Assume there exists $c >0$ and $\delta' \in \intfo{0}{\delta}$ small enough such that
	\begin{align}
	\label{eq: assumption_gamma_L}
	\abs{\theta_L} \geq c T_L^{\frac{1+\delta'}{\sqrt{3}}} \pt
	\end{align}
	Let $r>0$ small enough and $\mathbf{K}_\star \in \{\mathbf{K},\mathbf{K}'\}$ be a vertex of the first Brillouin zone $\Gamma^*$. Then, our goal is to show that, when $L \to \infty$, we have the expansion
	\begin{align}
	\label{eq: expansion_eigenvalues_honeycomb}
	\mu_{\pm,L}\left(\frac{\mathbf{K}_\star + \kappa}{L}\right) = -\mu_L + \petito{\abs{\theta_L}} \pm \frac{\sqrt{3}}{2} \abs{\theta_L}\abs{\kappa} \left(1+E(\kappa)\right)(1+\petito{1})\, ,
	\end{align}
	where $\abs{E(\kappa)} \leq C \abs{\kappa}$ for all $\kappa \in B(0,r)$ and where the $o$'s do not depend on $\kappa$.
	The proof is mainly an adaptation of the arguments used to show Theorem~\ref{theo: feshbach-schur} but with the additional knowledge that both the symmetry group $G$ and its action on $\lscr^\mathbf{R}$ and $\mathscr{P}^\mathbf{R}$ are entirely specified. 
	
	Throughout the proof, the $O$'s and $o$'s will not depend on the pseudo-momentum $\mathbf{k}$.
	\paragraph{First step.}
	For $\mathbf{k} \in \R^2$, we write
	\[
	A_L(\mathbf{k}/L) = \begin{pmatrix}
	A_L(\mathbf{k}/L; \mathbf{a},\mathbf{a})& A_L(\mathbf{k}/L; \mathbf{a},\mathbf{b}) \\
	A_L(\mathbf{k}/L; \mathbf{b},\mathbf{a})& A_L(\mathbf{k}/L; \mathbf{b},\mathbf{b})
	\end{pmatrix} \, ,	
	\]
	where
	\begin{gather*}
	A_L(\mathbf{k}/L; \mathbf{a},\mathbf{a}) = \overline{A_L(\mathbf{k}/L; \mathbf{a},\mathbf{a})} = \sum_{\mathbf{u\in}\lscr} e^{i\mathbf{k\cdot u}} A_L(\mathbf{a, a+u}) \, ,\\
	A_L(\mathbf{k}/L; \mathbf{b},\mathbf{b}) = \overline{A_L(\mathbf{k}/L; \mathbf{b},\mathbf{b})} = \sum_{\mathbf{u\in}\lscr} e^{i\mathbf{k\cdot u}} A_L(\mathbf{b, b+u}) \, ,\\
	A_L(\mathbf{k}/L; \mathbf{a},\mathbf{b}) = \overline{A_L(\mathbf{k}/L; \mathbf{b},\mathbf{a})} = \sum_{\mathbf{u\in}\lscr} e^{i\mathbf{k\cdot u}} A_L(\mathbf{a, b+u}) \pt
	\end{gather*}
	We recall that $s_L$ denotes the scaling operator, defined by $s_L\mathbf{x} = L\mathbf{x}$.
	Using the fact that the point group of the symmetry group $G$ of $\lscr^H$ acts on $\lscr$ (see~\cite[Theorem 25.2]{armstrong1988groups}), we can show that
	\begin{align*}
	\forall \mathbf{k} \in \Gamma^*_L,\quad\forall \mathbf{r\in}\lscr^\mathbf{R},\quad \forall g \in G,\quad (s_Lgs_L^{-1}) \cdot u_{L,\mathbf{r}}(\mathbf{k},\cdot) = 	u_{L,g\cdot\mathbf{r}}(S_g\mathbf{k},\cdot) \, ,
	\end{align*}
	where $S_g$ is the linear part of $g$. Then, we deduce that
	\begin{align}
	\label{eq: A_L_k_symmetry_relations}
	\forall g \in G,\quad \forall \mathbf{k} \in \Gamma_L^*,\quad	g \cdot A_L(\mathbf{k}) \coloneqq A_L(S_g^{-1} \mathbf{k}) \, ,
	\end{align}
	where $S_g$ is the linear part of $g$.	Using identity \eqref{eq: A_L_k_symmetry_relations} with $g$ the parity symmetry $(\mathcal{P}\psi)(\mathbf{x}) = \psi(-\mathbf{x})$, we can show that
	\begin{align*}
	A_L(\mathbf{k}/L; \mathbf{a},\mathbf{a}) = \overline{A_L(\mathbf{k}/L; \mathbf{b},\mathbf{b})} = A_L(\mathbf{k}/L; \mathbf{b},\mathbf{b}) \pt
	\end{align*}		
	\paragraph{Second step.}
	Now, we estimate each matrix element of $A_L(\mathbf{k}/L)$, starting with the diagonal terms. We introduce the following equivalence relation on $\lscr$:
	\begin{align*}
	\mathbf{u}\sim_1\mathbf{u'} \Longleftrightarrow \exists k \in \{0,1,2\},\quad \mathbf{u'} = M_\mathcal{R}^k \mathbf{u} \, ,
	\end{align*}
	where $M_\mathcal{R}$ is the rotation matrix by $2\pi/3$. Now, using identity \eqref{eq: A_invariance} with $g$ the rotation by $2\pi/3$ about $\mathbf{a}$, we can write
	\begin{multline*}
	A_L(\mathbf{k}/L; \mathbf{a},\mathbf{a}) 
	= \sum_{\mathbf{u}\in\lscr} e^{i\mathbf{K_\star}\cdot\mathbf{u}} A_L(\mathbf{a,a+u})\\
	+ \sum_{\mathbf{u\in}\lscr/\sim_1} \left(e^{i\mathbf{k}\cdot \mathbf{u}}+e^{i\mathbf{k}\cdot  M_\mathcal{R}\mathbf{u}}+e^{i\mathbf{k}\cdot  M_\mathcal{R}^2 \mathbf{u}} -3 e^{i\mathbf{K}_\star\cdot \mathbf{u}}\right) A_L(\mathbf{a, a+u}) \pt
	\end{multline*}
	Let $\kappa \subset B(0,r)$ for some $r>0$.	Using that $M_\mathcal{R} \mathbf{K}_\star \in \mathbf{K}_\star + \lscr^*$, we have for all $\mathbf{u} \in \lscr$ and for all $\mathbf{k} = \mathbf{K}_\star + \kappa$
	\begin{align*}
	\absL{e^{i\mathbf{k}\cdot \mathbf{u}}+e^{i\mathbf{k}\cdot  M_\mathcal{R}\mathbf{u}}+e^{i\mathbf{k}\cdot  M_\mathcal{R}^2 \mathbf{u}} -3 e^{i\mathbf{K}_\star\cdot \mathbf{u}}} 
	= \absL{ e^{i\kappa\cdot \mathbf{u}}+e^{i\kappa\cdot  M_\mathcal{R}\mathbf{u}}+e^{i\kappa\cdot  M_\mathcal{R}^2 \mathbf{u}}-3 }
	\lesssim \abs{\kappa}^2 \abs{\mathbf{u}}^2 \pt
	\end{align*}
	Then, using the exponential localization of $A_L$ provided in \eqref{eq: estimate_GDG}, the expansion $A_L(\mathbf{a,a}) = - \mu_L + \grandO{T_L^{\frac{1+\delta}{\sqrt{3}}-} + T_L^{1-}}$ (see Proposition~\ref{prop: interaction_matrix}) and assumption \eqref{eq: assumption_gamma_L}, we have for $L$ large enough
	\begin{multline}
	\label{eq: diagonal}
	A_L\left(\frac{\mathbf{K}_\star + \kappa}{L}; \mathbf{a},\mathbf{a}\right) 
	= \sum_{\mathbf{u}\in\lscr} e^{i\mathbf{K_\star}\cdot\mathbf{u}} A_L(\mathbf{a,a+u}) + \abs{\kappa}^2 \grandO{\sum_{\mathbf{u} \in \lscr\setminus\{0\}} \abs{\mathbf{u}}^2 A_L(\mathbf{a,a+u})} \\
	= -\mu_L + (1+\abs{\kappa}^2)\grandO{T_L^{\min\left(1,\frac{1+\delta}{\sqrt{3}}\right)-}} = -\mu_L + (1+\abs{\kappa}^2) \petito{\abs{\theta_L}}\pt
	\end{multline}
	\paragraph{Third step.}
	Now, we consider the off-diagonal terms. If $g$ denotes the rotation by $2\pi/3$ about $\mathbf{a}$ then, for all $\mathbf{u}\in\lscr$, we have
	\begin{gather*}
	A_L(\mathbf{a,b+u}) = A_L(\mathbf{a},g\cdot(\mathbf{b+u})) = A_L(\mathbf{a},\mathbf{b}+M_\mathcal{R}(\mathbf{u}+\mathbf{u}_1)) \, ,\\
	A_L(\mathbf{a,b+u}) = A_L(\mathbf{a},g^2\cdot(\mathbf{b+u})) = A_L(\mathbf{a},\mathbf{b}+M^2_\mathcal{R}(\mathbf{u}+\mathbf{u}_2)) \pt
	\end{gather*}
	This leads to the following equivalence relation
	\begin{align}
	\label{eq: equivalence_relation}
	\mathbf{u}\sim_2\mathbf{u'} \Longleftrightarrow \exists k \in \{0,1,2\},\quad \mathbf{u'} = M_\mathcal{R}^k( \mathbf{u}+\mathbf{u}_k) \, ,
	\end{align}
	where we have used the convention $\mathbf{u}_0 = 0$. We write
	\begin{align*}
	A_L(\mathbf{k}/L; \mathbf{a},\mathbf{b}) = \sum_{\mathbf{u\in}\lscr/\sim_2} \left(e^{i\mathbf{k}\cdot \mathbf{u}}+e^{i\mathbf{k}\cdot  M_\mathcal{R}(\mathbf{u}+\mathbf{u}_1)}+e^{i\mathbf{k}\cdot  M_\mathcal{R}^2 (\mathbf{u}+\mathbf{u}_2)}\right) A_L(\mathbf{a, b+u}) \pt
	\end{align*}
	Using again that $M_\mathcal{R}\mathbf{K}_\star \in \mathbf{K}_\star + \lscr^*$, we have for $\mathbf{k} = \mathbf{K}_\star+\kappa$
	\begin{align*}
	A_L(\mathbf{k}/L; \mathbf{a},\mathbf{b}) = \sum_{\mathbf{u\in}\lscr/\sim_2} e^{i\mathbf{K}_\star\cdot \mathbf{u}} \left(e^{i\kappa\cdot \mathbf{u}}+e^{i\mathbf{K}_\star\cdot \mathbf{u}_1}e^{i\kappa\cdot  M_\mathcal{R}(\mathbf{u}+\mathbf{u}_1)}+e^{i\mathbf{K}_\star\cdot \mathbf{u}_2}e^{i\kappa\cdot  M_\mathcal{R}^2 (\mathbf{u}+\mathbf{u}_2)}\right) A_L(\mathbf{a, b+u}) \pt
	\end{align*}
	From now on, we only consider the case $\mathbf{K}_\star = \mathbf{K}$. The minor changes for the case $\mathbf{K}_\star=\mathbf{K}'$ are left to the reader.
	Using the identities $e^{i\mathbf{K}\cdot \mathbf{u}_k} = j^k$ where $k \in \{0,1,2\}$ and $j = e^{i\frac{2\pi}{3}}$, we have
	\begin{multline*}
	A_L\left(\frac{\mathbf{K}+\kappa}{L}; \mathbf{a},\mathbf{b}\right)
	= i \kappa \cdot \sum_{\mathbf{u} \in \lscr/\sim_2} e^{i\mathbf{K}\cdot\mathbf{u}} \left(\mathbf{u} + j M_\mathcal{R}(\mathbf{u+u}_1)+j^2 M^2_\mathcal{R}(\mathbf{u+u}_2)\right)A_L(\mathbf{a,b+u}) \\
	+ \abs{\kappa}^2 \grandO{ \sum_{\mathbf{u} \in \lscr/\sim_2} (1+\abs{\mathbf{u}}^2) A_L(\mathbf{a,b+u})} \pt
	\end{multline*}
	By assumption \eqref{eq: assumption_gamma_L} and Proposition~\ref{prop: interaction_matrix}, we have $A_L(\mathbf{a,b}) = \theta_L(1+\petito{1})$.
	Using the exponential localization \eqref{eq: estimate_GDG} of $A_L$, the main contributor for both sums when $L$ is large correspond to $\mathbf{u}=0$. In addition, a computation shows: $ jM_\mathcal{R} \mathbf{u}_1 + j^2M_\mathcal{R}^2\mathbf{u}_2 
	= \frac{\sqrt{3}}{2}(1~ i)^T$.
	Hence, we have
	\begin{align}
	\label{eq: off_diagonal}
	A_L\left(\frac{\mathbf{K}+\kappa}{L}; \mathbf{a},\mathbf{b}\right) = \frac{\sqrt{3}}{2}i \theta_L (\kappa_1+i\kappa_2+E(\kappa)) (1+\petito{1})\, ,
	\end{align}
	with $\absL{E(\kappa)} \leq  C \abs{\kappa}^2$ for all $\kappa \in B(0,r)$.
	Now, the eigenvalues $\lambda_{\pm,L}(\kappa)$ of $A_L\left(\frac{\mathbf{K}+\kappa}{L}\right)$ solve
	\begin{align*}
	\mu^2 - 2 A_L\left(\frac{\mathbf{K}+\kappa}{L};\mathbf{a,a}\right) \mu + \absL{A_L\left(\frac{\mathbf{K}+\kappa}{L};\mathbf{a,b}\right)}^2 = 0 \pt
	\end{align*}
	Using expansions \eqref{eq: diagonal} and \eqref{eq: off_diagonal}, we obtain that
	\begin{align}
	\label{eq: expansion_vp}
	\lambda_{\pm,L}(\kappa) = -\mu_L + \petito{\abs{\theta_L}} \pm \frac{\sqrt{3}}{2} \abs{\theta_L}\abs{\kappa} \left(1+E'(\kappa)\right)(1+\petito{1}) \, ,
	\end{align}
	where  $\abs{E'(\kappa)} \leq C' \abs{\kappa}$ for all $\kappa \in B(0,r)$. In particular, a conical singularity appears at $\kappa=0$.
	\paragraph{Fourth step.}
	Let $\lambda \in \intoo{-\infty}{-\mu_L+g/3}$ where $g$ denotes the first spectral gap of the mono-atomic operators defined in \eqref{eq: effective_mono_atomic_operator}.
	Following the proof of Proposition~\ref{prop: estimation_key}, it can be seen that for some constant $C>0$, we have: for all $\mathbf{k} \in \R^2$
	\begin{multline*}
	\norme{C_L(\mathbf{k}/L)^* (B_L(\mathbf{k}/L) - \lambda)^{-1} C_L(\mathbf{k}/L)} = \norme{(B_L(\mathbf{k}/L) - \lambda)^{-1/2} C_L(\mathbf{k}/L)}^2 \\
	\lesssim L^C \sum_{\mathbf{u} \in \lscr} e^{-i\mathbf{u\cdot k}} \left[\left(\sqrt{-\Delta+1} \mathcal{P}\psi_{L,\mathbf{r}}\right)\ast \left(\sqrt{-\Delta+1} \psi_{L,\mathbf{r}}\right)\right](\mathbf{u}) \, ,
	\end{multline*}
	where $\psi_{L,\mathbf{r}} = (1-\chi_{L,\mathbf{r}}) v_{L,\mathbf{r}}$ for any $\mathbf{r}\in\{\mathbf{a,b}\}$. Notice that the sum is real by Poisson summation formula~\eqref{prop: poisson_summation_formula}. Assume $\mathbf{r=a}$ (the case $\mathbf{r=b}$ is treated similarly). First, notice that $\mathcal{P}\psi_{L,\mathbf{a}} = \psi_{L,\mathbf{b}}$ and that, using the Fourier transform, $\sqrt{-\Delta+1}\psi_{L,\mathbf{a}}$ has the symmetries of $\psi_{L,\mathbf{a}}$. Then, it is not difficult to show that $ \left[\left(\sqrt{-\Delta+1} \psi_{L,\mathbf{b}}\right)\ast \left(\sqrt{-\Delta+1} \psi_{L,\mathbf{a}}\right)\right](\mathbf{u})$ is constant on the equivalence class associated with the equivalence relation $\sim_2$ introduced in \eqref{eq: equivalence_relation}. Then, we combine the arguments used in the proof of \eqref{eq: off_diagonal} with those used in the proof of Lemma~\ref{lemma: estimation_key8} to show: for all $\mathbf{k} = \mathbf{K}_\star +\kappa$ with $\kappa \in B(0,r)$
	\begin{align}
	\label{enfin}
	\normeL{C_L(\mathbf{k}/L)^* (B_L(\mathbf{k}/L) - \lambda)^{-1} C_L(\mathbf{k}/L)} =  (1+ \abs{\kappa} + E''(\kappa)) \grandO{T_L^{\frac{1+\delta}{\sqrt{3}}}} = (1+ \abs{\kappa} + E''(\kappa)) \petito{\abs{\theta_L}} \, ,
	\end{align}
	where there exists $C''>0$ such that $\abs{E''(\kappa)} \leq C''\abs{\kappa}^2$ for all $\kappa \in B(0,r)$. Finally, by using the estimate \eqref{enfin}, the expansion \eqref{eq: expansion_vp} and the relation \eqref{eq: FS_method}, we show the expansion \eqref{eq: expansion_eigenvalues_honeycomb}.

	\section{Proof of Theorem~\ref{theo: rHF}}\label{sec:reduced-periodic-hartree-fock-model-at-dissociation}
	
	In this section, we show that we can apply Theorem~\ref{theo: feshbach-schur} to the periodic rHF model introduced in Section~\ref{sec:the-rhf-model-for-periodic-systems-at-dissociation}. We first recall the main features of this model. For simplicity, we assume $q=1$. It consists in solving the minimization problem
	\begin{align*}
	E_L \coloneqq \inf \enstqbis{\mathcal{E}_L(\gamma)}{\gamma \in \mathcal{S}_{\mathrm{per},L} \et \trm_{\mathscr{L}_L} (\gamma) = N} ,
	\end{align*}
	where $\mathcal{S}_{\mathrm{per},L}$ (resp. $\mathcal{E}_L(\gamma)$) is defined in \eqref{eq: admissible_states} (resp. in \eqref{eq: rHF_energy}) and denotes the space of admissible states (resp. the periodic rHF energy of the state $\gamma$). For all $L\geq 1$, this model admits a unique minimizer, denoted by $\gamma_L$. We also denote by $\rho_L(\mathbf{x}) = \gamma_L(\mathbf{x,x})$ its one-body density. The minimizer $\gamma_L$ satisfies the self-consistent equation
	\begin{align*}
	\gamma_L = \mathds{1}_{\intof{-\infty}{\varepsilon_L}} \left(H^\mathrm{MF}_L\right) = \mathds{1}_{\intof{-\infty}{\varepsilon_L}} \left(-\Delta + V^\mathrm{MF}_L\right) \pt
	\end{align*}
	The mean-field potential $V^\mathrm{MF}_L$ is given by
	\begin{align*}
	V^\mathrm{MF}_L \coloneqq - W_L^\mathbf{R} + V_L^\mathbf{R} + \rho_L \ast_{L} W_L \, ,
	\end{align*}
	where $W_L$ is the periodic three-dimensional Coulomb interaction kernel (defined in~\eqref{eq: G_definition_bis} below) and where
	\begin{align*}
	- W_L^\mathbf{R} = \sum_{\mathbf{r\in R}} W_L(\cdot - L\mathbf{r}) \et V_L^\mathbf{R} = \sum_{\mathbf{r\in R}} \sum_{\mathbf{u}\in\lscr} V^\mathrm{pp}(\cdot - L(\mathbf{u+r}))\, ,
	\end{align*}
	are respectively the external potential induced by the nuclei of the lattice and a correction term given by an exact $\lscr^\mathbf{R}$-superposition of a radial and compactly supported potential $V^\mathrm{pp}\in L^p(\R^2)$ with $p>1$. The Fermi level $\varepsilon_L \in \R$ is chosen in order to have $\trm_{\mathscr{L}_L}(\gamma_L) = \int_{\Gamma_L} \rho_L = N$. Because Assumption~\ref{hypo_1} on the symmetry of the lattice does not depend on the model in consideration but only on the underlying lattice $\lscr^\mathbf{R}$, we will assume that
	\begin{align*}
	\boxed{\lscr^\mathbf{R} \text{ satisfies Assumption~\ref{hypo_1}.}}
	\end{align*}
	In this section, we denote by $G\subset E_2(\R)$ (resp. $G_L$) the symmetry group of $\lscr^\mathbf{R}$ (resp. $\lscr^\mathbf{R}_L$).
	
	\subsection{Properties of the periodic interaction kernel \texorpdfstring{$W_L$}{WL}}\label{sec:properties-of-the-periodic-interaction-kernel-gl}
	The $\lscr_L$-periodic interaction kernel $W_L$ is defined as the \emph{unique }solution in the sense of tempered distribution of the following system
	\begin{align}
	\label{eq: G_definition_bis}
	\left\lbrace
	\begin{aligned}
	& \sqrt{-\Delta} W_L= 2 \pi \left(\sum_{\mathbf{u} \in \mathscr{L}_L}\delta_\mathbf{u}-\frac{1}{\abs{\Gamma_L}}\right)\, , \\
	& W_L ~\text{is}~ \mathscr{L}_L\text{-periodic},~ \min_{\R^2} W_L = 0 \pt
	\end{aligned} \right.
	\end{align}
	This choice of kernel is motivated by the fact we consider a model where the charges are confined in two-dimensional space but interact with the three-dimensional Coulomb interaction. The potential $W_L$ is in fact the periodized version of $\frac{1}{\abs{\mathbf{x}}}$. It differs from the well-known three-dimensional periodic Coulomb interaction kernel (see for instance~\cite[Section XI]{lieb1977thomasfermi}) which is the solution of a similar system to \eqref{eq: G_definition_bis} where the constant $2\pi$ is replaced by $4\pi$ and the fractional Laplace operator $\sqrt{-\Delta}$ by the Laplace operator $-\Delta$. However, they share the same behavior $\abs{\mathbf{x}}^{-1}$ in the vicinity of the vertices of $\lscr_L$, as shown in Proposition~\ref{prop: madelung}.
	
	The kernel $W_L$ being the unique solution of the system \eqref{eq: G_definition_bis}, we have the relation $W_L = L^{-1} W_1(L^{-1}\cdot)$ for all $L\geq 1$. In addition, $W_L$ admits the following Fourier expansion of $W_L$: there exists $M = \fint_\Gamma W_1 >0$ such that
	\begin{align}
	\label{eq:fourier_expansion_W}
	W_L = L^{-1}M + \frac{2\pi}{\sqrt{\abs{\Gamma_L}}} \sum_{\mathbf{v} \in \lscr_L^* \setminus \{0\}} \frac{e_\mathbf{v}}{\abs{\mathbf{v}}} \, ,
	\end{align}
	where we recall that $e_\mathbf{v}(\mathbf{x}) = \abs{\Gamma_L}^{-1/2}e^{i\mathbf{x\cdot v}}$.
	\begin{lemma}
		\label{lemma: hypo2}
		The mean-field potential $V_L^\mathrm{MF}$ is $G_L$-invariant, that is,
		\[
		\boxed{V_L^\mathrm{MF} \text{ satisfies Assumption~\ref{hypo_2}.}}
		\]
	\end{lemma}
	\begin{proof}		
		By~\cite[Theorem 25.2]{armstrong1988groups}, the point group $J$ (independent from $L$) of $G_L$ acts on $\lscr_L$. In particular, the Wigner-Seitz cell $\Gamma_L$ of $\lscr_L$ is invariant with respect to $J$. Then, for all $\mathbf{(u,v)}\in\lscr_L\times\lscr^*_L$ and for all $g\in J$, we have $e^{g\mathbf{v}\cdot u} = e^{i\mathbf{v}\cdot g^{-1}\mathbf{u}} = 1 $ which shows that $J$ also acts on $\lscr_L^*$. Consequently, we can use the Fourier expansion~\eqref{eq:fourier_expansion_W} to show that the interaction kernel $W_L$ then $W_L^\mathbf{R}$ are invariant under the action of $G_L$. Since $V^\mathrm{pp}$ is radial, $V_L^\mathbf{R}$ is also invariant under the action of $G_L$. We deduce that the rHF energy $\mathcal{E}_L$ is invariant under the action of $G_L$. By uniqueness, this is also the case for $\gamma_L$ and $\rho_L$, hence also for $\rho_L*_LW_L$. This concludes the proof of Lemma~\ref{lemma: hypo2}.
	\end{proof}
	
	As for the three-dimensional periodic Coulomb interaction kernel, we can describe $W_L$ as a \emph{Madelung potential} that is a $\lscr_L$-periodic superposition of the potential induced by a neutral charge distribution (see for instance~\cite[Section XI.3.B]{lieb1977thomasfermi}). We introduce the following function
	\begin{align*}
	f_L(\mathbf{x}) \coloneqq \frac{1}{\abs{\mathbf{x}}} - \fint_{\Gamma_L} \frac{\diff \mathbf{y}}{\abs{\mathbf{x-y}}} = \left[|\cdot|^{-1} \ast \left(\delta_0 - \frac{\mathds{1}_{\Gamma_L}}{\abs{\Gamma_L}}\right)\right](\mathbf{x})\pt
	\end{align*}
	The next proposition states that $W_L$ is, up to an additive constant, given by the Madelung potential associated with the function $f_L$. 
	\begin{prop}[$W_L$ is a Madelung potential]
		\label{prop: madelung}
		There exists $M' \in \R$ such that
		\begin{align*}
		W_L = \sum_{\mathbf{u} \in \lscr_L} f_L(\cdot - \mathbf{u}) + L^{-1}M' \, ,
		\end{align*}
		where the series converges absolutely in $\mathcal{C}^\infty(\R^2 \setminus \lscr_L)$. In addition, the function $W_L$ is continuous on $\R^2 \setminus \lscr_L$, the function $W_L - |\cdot|^{-1}$ is bounded on $\Gamma_L$ and there exists $a \in \R$ and $C>0$ such that
		\begin{align}
		\label{eq: G_properties}
		\lim\limits_{\abs{\mathbf{x}} \to 0}\left( W_L(\mathbf{x}) - \abs{\mathbf{x}}^{-1}\right) = L^{-1}a \et \forall \mathbf{x} \in \Gamma_L,\quad \abs{W_L(\mathbf{x})} \leq C\abs{\mathbf{x}}^{-1}\pt
		\end{align}
	\end{prop}
	We do not show Proposition~\ref{prop: madelung} since the arguments are standard and we refer to~\cite[Proposition 3.38]{cazalisthesis} for the detailed proof (see also~\cite{lewin2022coulomb}).
	\begin{cor}
		\label{cor: inequality_potential}
		For all $r\in\intff{1}{\infty}$, there exists $C_r>0$ such that for all $L \geq 1$ and for all $u,v \in H^1_\mathrm{per}(\Gamma_L)$, we have
		\begin{gather}
		\label{eq: inequality_potential}
		\normeL{(uv) \ast_L W_L}_{\lper^r(\Gamma_L)} \leq C_r L^{1/r} \normeL{u}_{H^1_\mathrm{per}(\Gamma_L)}\normeL{v}_{H^1_\mathrm{per}(\Gamma_L)} \pt
		\end{gather}
	\end{cor}
	\begin{proof}
		By Young's inequality and Hölder's inequality, we have
		\begin{align*}
		\normeL{(uv)\ast_LW_L}_{\lper^1(\Gamma_L)} \leq \normeL{W_L}_{\lper^1(\Gamma_L)} \normeL{u}_{\lper^{2}(\Gamma_L)}\normeL{v}_{\lper^{2}(\Gamma_L)} \pt
		\end{align*}
		Then, inequality \eqref{eq: inequality_potential} for $r=1$ results from the identity $\normeL{W_L}_{\lper^{1}(\Gamma_L)} = L\normeL{W_1}_{\lper^{1}(\Gamma)}$ (consequence of the relation $W_L(\mathbf{x}) = L^{-1}W_1(L^{-1}\mathbf{x})$). When $r=\infty$, the proof is similar. We write for $A>0$
		\begin{align*}
		\normeL{(uv)\ast_LW_L}_{\lper^\infty(\Gamma_L)} 
		\leq \normeL{(uv)*_{L} W_L\mathds{1}_{\abs{W_L} <A}}_{\lper^\infty(\Gamma_L)}+ \normeL{(uv)*_{L} W_L\mathds{1}_{\abs{W_L} \geq A}}_{\lper^\infty(\Gamma_L)} \pt
		\end{align*}
		Then, using the inequality $\abs{W_L(\mathbf{x})} \leq C\abs{\mathbf{x}}^{-1}$ (see the right side of \eqref{eq: G_properties}) and Young's inequality, we get, for any $p\in\intoo{1}{2}$ and $p'\in\intoo{2}{\infty}$ such that $p^{-1} + (p')^{-1}=1$,
		\begin{align*}
		\normeL{(uv)\ast_LW_L}_{\lper^\infty(\Gamma_L)} \leq A \normeL{uv}_{\lper^1(\Gamma_L)} + C \normeL{|\cdot|^{-1}}_{L^p(B(0,CA^{-1}))} \normeL{uv}_{\lper^{p'}(\Gamma_L)} \pt
		\end{align*}
		We conclude using Hölder's inequality and the Sobolev embeddings. Finally, we obtain~\eqref{eq: inequality_potential} for any $r\in\intff{1}{\infty}$ by interpolation.
	\end{proof}
	As a direct consequence of Proposition~\ref{prop: madelung} and Corollary~\ref{cor: inequality_potential}, we have
	\begin{cor}
		The mean-field potential $V_L^\mathrm{MF}$ belongs to $\lper^p(\Gamma_L)$ for any $p\in \intoo{1}{2}$, that is
		\[
		\boxed{V_L^\mathrm{MF} \text{ satisfies Assumption~\ref{hypo_01}.}}
		\]
	\end{cor}
	
	\subsection{Reference model}
	
	In this section, we describe the mono-atomic Hartree model whose mean-field potential will give the reference potential satisfying Assumption~\ref{hypo_3} and appearing in Assumption~\ref{hypo_4}. We introduce the energy functional
	\begin{gather*}
	\mathcal{E}(u) \coloneqq \int_{\R^2} \abs{\nabla u(\mathbf{x})}^2\diff \mathbf{x} + \int_{\R^2}\left(-\frac{1}{\abs{\mathbf{x}}} + V^\mathrm{pp}(\mathbf{x}) \right) \abs{u(\mathbf{x})}^2 \diff \mathbf{x}  +  \frac{1}{2} \iint_{\R^2 \times \R^2} \frac{\abs{u(\mathbf{x})}^2 \abs{u(\mathbf{y})}^2}{\abs{\mathbf{x-y}}} \diff \mathbf{x} \diff \mathbf{y} \, ,
	\end{gather*} 
	where $V^\mathrm{pp}$ is radial, compactly supported and belongs to $L^p(\R^2)$ for some $p>1$. This functional is well-defined and continuous on $H^1(\R^2)$. For $\lambda \geq 0$, we consider the following minimization problem
	\begin{gather}
	\label{eq: min_problem_ow}
	I(\lambda) \coloneqq \inf \enstqbis{\mathcal{E}(u)}{u \in H^1(\R^2)\et \int_{\R^2} \abs{u}^2 = \lambda } \, ,
	\end{gather}
	which has been thoroughly studied in the literature, at least in its three-dimensional version~\cite{benguria1981thomasfermi, lieb1977thomasfermi, lieb1981thomasfermi, lions1981remarks}. In particular, there exists $\lambda_\mathrm{max} \geq 1$ such that if $\lambda \leq \lambda_\mathrm{max}$ then all the minimizing sequences for \eqref{eq: min_problem_ow} are precompact in $H^1(\R^2)$ and \eqref{eq: min_problem_ow} admits a unique minimizer, up to a phase factor. In addition, the following \emph{binding inequality} holds
	\begin{align}
	\label{eq: binding_inequality}
	\forall 0\leq \mu < \lambda \leq \lambda_\mathrm{max},\quad I(\lambda) < I(\mu) \pt
	\end{align}
	When $\lambda=1$, we denote by $v$ the minimizer of \eqref{eq: min_problem_ow} which is the ground state of the self-adjoint operator
	\begin{align*}
	H^\mathrm{MF} \coloneqq -\Delta - |\cdot|^{-1} + V^\mathrm{pp} + \abs{v}^2\ast|\cdot|^{-1} \, ,
	\end{align*}
	defined on the domain $\mathcal{D}(H^\mathrm{MF}) = \enstq{u \in H^1(\R^2)}{(-\Delta - |\cdot|^{-1} + V^\mathrm{pp})u \in L^2(\R^2)}$. In the following, we denote by 
	\begin{align*}
	V^\mathrm{MF} \coloneqq - |\cdot|^{-1} + V^\mathrm{pp} + \abs{v}^2\ast|\cdot|^{-1} \, ,
	\end{align*}
	the mean-field potential which will play the role of reference potential, according to Section~\ref{sec:reference-operator}. By Weyl's theorem, the essential spectrum of $H^\mathrm{MF}$ is given by the half-line $\intfo{0}{\infty}$. We choose $V^\mathrm{pp}$ such that $H^\mathrm{MF}$ admits at least one negative eigenvalue (such $V^\mathrm{pp}$ exists, see~\cite[Chapter 3 -- Appendix A]{cazalisthesis}) and we denote by $-\mu <0$ the lowest one. Our assumption that $-\mu<0$ is equivalent to $\lambda_\mathrm{max}>1$. By~\cite{goelden1977nondegeneracy}, it is non degenerate, the phase of $v$ can be chosen such that $v>0$ everywhere and, since $H^\mathrm{MF}$ is invariant under rotations, $v$ is radial.
	
	Recall that $V^\mathrm{pp}$ is compactly supported and belongs to $L^p(\R^2)$ for some $p>1$. Then, by following~\cite{cazalis} and using Proposition~\ref{prop: regularity}, one can show that:
	\begin{enumerate}[noitemsep, label=(\roman*)]
		\item We have $v \in H^{r}(\R^2)$ for $r=p$ if $p<2$ and any $r<2$ otherwise.
		\item There exists $C>0$ such that: for all $\mathbf{x} \in \R^2$,
		\begin{align*}
		\frac{1}{C} \frac{e^{-\sqrt{\mu} \abs{\mathbf{x}}}}{1+ \sqrt{\abs{\mathbf{x}}}} \leq v(\mathbf{x}) \leq C \frac{e^{-\sqrt{\mu} \abs{\mathbf{x}}}}{1+ \sqrt{\abs{\mathbf{x}}}}  \et \abs{\nabla v(\mathbf{x})} \leq C \frac{e^{-\sqrt{\mu} \abs{\mathbf{x}}}}{1+ \sqrt{\abs{\mathbf{x}}}} \pt
		\end{align*}
		\item We have $V^\mathrm{MF}(\mathbf{x}) \sim m_1/(4\abs{\mathbf{x}}^3)$ with $m_1 = \int_{\R^2} \abs{\mathbf{x}}^2\abs{v(\mathbf{x})}^2\diff \mathbf{x}>0$ when $\abs{\mathbf{x}} \to \infty$. In particular, we have $V^\mathrm{MF} \in L^p(\R^2)$.
		\item The energy functional $\mathcal{E}$ satisfies the following \emph{stability inequality}: there exists $C>0$ such that, for all $u\in H^1(\R^2)$ with $\norme{u}_{L^2(\R^2)}=1$, we have
		\begin{align}
		\label{eq: stability}
		\mathcal{E}(u) \geq \mathcal{E}(v) + C \min_{\theta \in \intfo{0}{2\pi}}\norme{e^{i\theta} u - v}^2_{H^1(\R^2)} \pt
		\end{align}
	\end{enumerate}
	\begin{rem}
		From this discussion, we see that
		\begin{align*}
		\boxed{V^\mathrm{MF} \text{ satisfies Assumption~\ref{hypo_3}}\pt}
		\end{align*}
	\end{rem}	
	
	\subsection{Convergence of the periodic model to the reference model}
	In this section, we show that the periodic rHF model \eqref{eq: rHF_model} is given, to leading order, by a periodic superposition of translated versions of the mono-atomic Hartree potential $V^\mathrm{MF}$, introduced in the previous section. In this direction, we use the concentration-compactness method~\cite{lions1984concentrationlocallyI, lions1984concentrationlocallyII}. Our arguments call on the binding inequality \eqref{eq: binding_inequality} and the stability inequality \eqref{eq: stability}.
	
	Let $\delta \in \intoo{0}{1/2}$ and $\chi \in \mathcal{C}^\infty_c(\R^2)$ be a localization function such that
	\begin{align}
	\label{eq: definition_chibis}
	0 \leq \chi \leq 1,\quad\chi \equiv 1 \quad \text{on} \quad B\left(0,\delta d_0/2\right) \et \supp \chi \subset  B\left(0,\delta d_0\right) \, ,
	\end{align}
	where $d_0$, defined in \eqref{eq: min_distance_lattice}, is the nearest neighbor distance of the lattice $\lscr^\mathbf{R}$.
	For $L\geq 1$ and $\mathbf{r} \in \lscr^\mathbf{R}$, we set $\chi_{L,\mathbf{r}}(\mathbf{x}) \coloneqq \chi(L^{-1}\mathbf{x} - \mathbf{r})$. Notice that the functions $\{\chi_{L,\mathbf{r}}\}_{\mathbf{r\in} \lscr^\mathbf{R}}$ have pairwise disjoint supports. We set $\rho_{L,\mathbf{r}} \coloneqq \abs{\chi_{L,\mathbf{r}}}^2 \rho_L$ for $\mathbf{r\in}\lscr^\mathbf{R}$. Because $\rho_L$ is $\lscr_L$-invariant by Lemma~\ref{lemma: hypo2}, we have
	\begin{align*}
	\forall \mathbf{u}\in\lscr,\quad \forall \mathbf{r}\in\lscr^\mathbf{R},\quad \rho_{L,\mathbf{u+r}} = \rho_{L,\mathbf{r}}(\cdot - L\mathbf{u}) \pt
	\end{align*}
	The following proposition states that the periodic rHF model \eqref{eq: rHF_model} is, in the vicinity of the singularities, well approached by the mono-atomic Hartree model \eqref{eq: min_problem_ow} with $\lambda=1$.
	\begin{prop}
		\label{prop: variational}
		We have
		\begin{align}
		\label{eq: variational1}
		\lim\limits_{L \to \infty} E_{L} = NI(1) \et \forall \mathbf{r\in R},\quad \lim\limits_{L \to \infty}\normeL{\sqrt{\rho_{L,\mathbf{r}}}-v(\cdot - L\mathbf{r})}_{H^1(\R^2)} = 0 \, ,
		\end{align}
		where $I(1)$ is defined in \eqref{eq: min_problem_ow} and $v$ is the associated unique Hartree minimizer. In addition, we have 
		\begin{align}
		\label{eq: variational2}
		\normeL{\rho_L \ast_L W_L - \sum_{\mathbf{r\in R}} \abs{v(\cdot-L\mathbf{r})}^2\ast |\cdot|^{-1}}_{\lper^\infty(\Gamma_L)} \limit{L\to\infty} 0 \pt
		\end{align}	
	\end{prop}
	We temporally admit the conclusions of Proposition~\ref{prop: variational}. This allows us to write the
	\begin{proof}[Proof of Theorem~\ref{theo: rHF}]
		We have that
		\[
		\boxed{V_L^\mathrm{MF} \text{ satisfies Assumption~\ref{hypo_04} and Assumption~\ref{hypo_4},}}
		\]
		with $V^\mathrm{MF}$ as reference potential. Indeed, the validity of Assumption~\ref{hypo_04} (resp. Assumption~\ref{hypo_4} with $V^\mathrm{MF}$ as reference potential) results from Proposition~\ref{prop: madelung} and \eqref{eq: variational2} (resp. the left side of \eqref{eq: G_properties}, \eqref{eq: variational2} and the fact that $V^\mathrm{pp}$ is compactly supported). Consequently, Theorem~\ref{theo: rHF} holds.
	\end{proof}	
	\begin{proof}[Proof of Corollary~\ref{cor:dirac_points_rHF}]
		We only have to show that the Fermi level $\epsilon_L$ is exactly equal to the energy level of the cones. This amounts to show that the two lowest bands of the dispersion relation only overlap at the vertices of the first Brillouin zone.
		
		By Theorem~\ref{theo: rHF}, we can apply Theorem~\ref{theo: feshbach-schur} and Theorem~\ref{theo: existence_dirac_cones}. The first one provides, for all $\kappa\in\Gamma^*$
		\begin{align*}
		\frac{1}{\abs{\theta_L}}\left(\mu_{+,L}\left(\frac{\kappa}{L}\right) - \mu_{-,L}\left(\frac{\mathbf{K}_\star+\kappa}{L}\right)\right) =
		2 \abs{1+e^{i\mathbf{\kappa\cdot u}_1}+e^{i\mathbf{\kappa\cdot u}_2}} + \petito{1} \, ,
		\end{align*}
		where we recall that the map $\kappa\in\Gamma^* \mapsto \abs{1+e^{i\mathbf{\kappa\cdot u}_1}+e^{i\mathbf{\kappa\cdot u}_2}}$ is equal to zero if and only if $\kappa\in\{\mathbf{K,K'}\}$ is a vertex of $\Gamma^*$. The second one gives
		\begin{align*}
		\frac{1}{\abs{\theta_L}}\left(\mu_{+,L}\left(\frac{\mathbf{K}_\star+\kappa}{L}\right) - \mu_{-,L}\left(\frac{\mathbf{K}_\star+\kappa}{L}\right)\right) = \sqrt{3}\abs{\kappa}(1+E(\kappa))(1+\petito{1}) \geq \abs{\kappa} \, ,
		\end{align*}
		for all $\mathbf{K}_\star\in\{\mathbf{K,K'}\}$, for all $\abs{\kappa}$ small enough and $L$ large enough. Hence, the two bands only overlap at the vertices of the first Brillouin zone. This concludes the proof of Corollary~\ref{cor:dirac_points_rHF}.		
	\end{proof}
	\begin{proof}[Proof of Proposition~\ref{prop: variational}]
		We divide the proof into six steps. First, we show that $E_L$ is bounded from above by $NI$ plus a small correction. For this purpose, we insert a trial state into the periodic rHF energy functional $\mathcal{E}_L$ defined in \eqref{eq: rHF_energy}. Afterwards, using localization methods, we show that $E_L$ is bounded from below by $\sum_{\mathbf{r\in R}} I\left(\norme{\rho_{L,\mathbf{r}}}_{L^1(\R^2)}\right)$ plus a small correction. In the third step, we show that $\norme{\rho_{L,\mathbf{r}}}_{L^1(\R^2)}$ converges to 1 when $L\to\infty$ using the binding inequality \eqref{eq: binding_inequality}. Then we prove that the kinetic energy $\trm_{\mathscr{L}_L}(-\Delta\gamma_L)$ is uniformly bounded with respect to $L\geq 1$. In the fifth step, we show the convergence \eqref{eq: variational1} thanks to the stability inequality \eqref{eq: stability} and the previous steps. We show the convergence of the potential~\eqref{eq: variational2} in the sixth and final step.
		
		\paragraph{First step.} We claim the upper bound
		\begin{align}
		\label{eq: upper_bound}
		E_L = \mathcal{E}_L(\gamma_L) \leq NI(1) + \grandO{L^{-1}}\pt
		\end{align}
		To show this upper bound, we construct an appropriate trial state for the minimization problem~\eqref{eq: rHF_model}.
		For $L\geq 1$ and $\mathbf{r}\in\lscr^\mathbf{R}$, we set
		\begin{align*}
		v_{L,\mathbf{r}} = \norme{\chi_{L,\mathbf{r}} v(\cdot - L\mathbf{r})}^{-1}_{L^2(\R^2)} \chi_{L,\mathbf{r}} v(\cdot - L\mathbf{r}) \pt
		\end{align*}
		Notice that these functions have disjoint supports and that, by construction, the family $\{v_{L,\mathbf{r}}\}_{\mathbf{r\in}\lscr^\mathbf{R}}$ forms an orthonormal system. In addition, in view of the exponential decay of $u$, we have
		\begin{align}
		\label{eq: approximation_v}
		\forall \mathbf{r\in}\lscr^\mathbf{R},\quad \norme{v_{L,\mathbf{r}}(\cdot+L\mathbf{r}) - v}_{H^1(\R^2)} = \grandO{L^{-\infty}} \pt
		\end{align}
		We consider the trial state
		\[
		\gamma_\mathrm{trial} = \sum_{\mathbf{r} \in \lscr^\mathbf{R}} \ket{v_{L,\mathbf{r}}} \bra{v_{L,\mathbf{r}}} \in \mathcal{S}_{\mathrm{per},L} \, ,
		\]
		which can be decomposed in fibers as
		\begin{align}
		\label{eq: 18_06_2022}
		\gamma_\mathrm{trial} \simeq \fint_{\Gamma_L^*}^\oplus \sum_{\mathbf{r\in R}} \ket{\mathcal{U}_\mathrm{BF}(v_{L,\mathbf{r}})(\mathbf{k},\cdot)}\bra{\mathcal{U}_\mathrm{BF}(v_{L,\mathbf{r}})(\mathbf{k},\cdot)} \diff\mathbf{k} \pt
		\end{align}
		From identity \eqref{eq: 18_06_2022} and the fact that the functions $v_{L,\mathbf{r}}$ have pairwise disjoint support, we can compute the one-body density $\rho_\mathrm{trial}$ of $\gamma_\mathrm{trial}$ and its kinetic energy. We have
		\begin{gather*}
		\rho_\mathrm{trial}(\mathbf{x}) = \sum_{\mathbf{r\in \lscr^R}} \absL{v_{L,\mathbf{r}}(\mathbf{x})}^2 \et
		\trm_{\mathscr{L}_L}(-\Delta \gamma_\mathrm{trial}) = \sum_{\mathbf{r\in R}} \norme{\nabla v_{L,\mathbf{r}}}_{L^2(\R^2)}^2 \pt
		\end{gather*}
		Then, its periodic rHF energy is given by
		\begin{align*}
		\mathcal{E}_L(\gamma_\mathrm{trial})
		=~& \sum_{\mathbf{r\in R}} \left(\norme{\nabla v_{L,\mathbf{r}}}_{L^2(\R^2)}^2  + \int_{\R^2}  \left(-W_L^{\mathbf{R}}(\mathbf{x})+V_L^{\mathbf{R}}(\mathbf{x})\right) \abs{v_{L,\mathbf{r}}(\mathbf{x})}^2\diff \mathbf{x}\right) \\
		&+ \frac{1}{2} \sum_{\mathbf{r,r'\in R}} \iint_{\R^2 \times \R^2} \abs{v_{L,\mathbf{r}}(\mathbf{x})}^2\abs{v_{L,\mathbf{r'}}(\mathbf{y})}^2W_L(\mathbf{x-y}) \diff \mathbf{x} \diff \mathbf{y} \pt
		\end{align*}
		Using that $\supp v_{L,\mathbf{r}} \subset \frac{1}{2}\Gamma_L + L\mathbf{r}$, the left side of \eqref{eq: G_properties} and the fact that $\supp V^\mathrm{pp}\subset B(0,R)$ for some fixed $R>0$, we have, for all $\mathbf{r\in R}$ and $L$ large enough,
		\begin{gather*}
		\int_{\R^2}  W_L(\mathbf{x}-L\mathbf{r}) \abs{v_{L,\mathbf{r}}(\mathbf{x})}^2 \diff \mathbf{x} = \int_{\R^2} \frac{\abs{v_{L,\mathbf{r}}(\mathbf{x})}^2}{\abs{\mathbf{x}-L\mathbf{r}}} \diff \mathbf{x} + \grandO{L^{-1}} \, ,\\
		\int_{\R^2}  V^\mathbf{R}_L(\mathbf{x}) \abs{v_{L,\mathbf{r}}(\mathbf{x})}^2 \diff \mathbf{x} = \int_{\R^2} V^\mathrm{pp}(\mathbf{x} - L\mathbf{r}) \abs{v_{L,\mathbf{r}}(\mathbf{x})}^2 \diff \mathbf{x} \pt
		\end{gather*}
		For $\mathbf{r\neq r'}$, we use that $\dist(\supp v_{L,\mathbf{r}},\lscr_L+L\mathbf{r'}) \geq L(1-\delta)d_0 $ and the left side of \eqref{eq: G_properties} to obtain
		\begin{align*}
		\int_{\R^2}  W_L(\mathbf{x}-L\mathbf{r'})\abs{v_{L,\mathbf{r}}(\mathbf{x})}^2 \diff \mathbf{x}
		= \sum_{\mathbf{u}\in\lscr} \int_{\R^2} \frac{\abs{v_{L,\mathbf{r}}(\mathbf{x})}^2}{\abs{\mathbf{x}-L(\mathbf{u+r'})}} \mathds{1}_{\mathbf{x}-L(\mathbf{u+r'}) \in\Gamma_L} \diff \mathbf{x} + \grandO{L^{-1}} = \grandO{L^{-1}} \pt
		\end{align*}
		Now, we turn our attention to the direct term. As $\supp v_{L,\mathbf{r}} \subset \frac{1}{2}\Gamma_L + L\mathbf{r}$ implies $\mathbf{x-y} \in \Gamma_L$ for all $(\mathbf{x,y})\in(\supp v_{L,\mathbf{r}})^2$, we have
		\begin{align*}
		\iint_{\R^2 \times \R^2} \abs{v_{L,\mathbf{r}}(\mathbf{x})}^2\abs{v_{L,\mathbf{r}}(\mathbf{y})}^2W_L(\mathbf{x-y}) \diff \mathbf{x} \diff \mathbf{y}
		= \iint_{\R^2 \times \R^2} \frac{\abs{v_{L,\mathbf{r}}(\mathbf{x})}^2\abs{v_{L,\mathbf{r'}}(\mathbf{y})}^2}{\abs{\mathbf{x-y}}} \diff \mathbf{x} \diff \mathbf{y} + \grandO{L^{-1}} \, ,
		\end{align*}
		by the left side of \eqref{eq: G_properties}. For $\mathbf{r\neq r'}$, we use the left side of \eqref{eq: G_properties} and the fact that (see \eqref{eq: definition_chibis})
		\begin{align*}
		\forall \mathbf{u}\in\lscr,\quad \dist(\supp v_{L,\mathbf{u+r}},\supp v_{L,\mathbf{r'}}) \geq L(1 - 2\delta) d_0 > 0\, ,
		\end{align*}
		to obtain
		\begin{align*}
		\iint_{\R^2 \times \R^2} &\abs{v_{L,\mathbf{r}}(\mathbf{x})}^2\abs{v_{L,\mathbf{r}}(\mathbf{y})}^2W_L(\mathbf{x-y}) \diff \mathbf{x} \diff \mathbf{y} = \grandO{L^{-1}} \pt
		\end{align*}
		Hence, we have shown
		\begin{align*}
		\mathcal{E}_L(\gamma_L) \leq \sum_{\mathbf{r\in R}} \mathcal{E}(v_{L,\mathbf{r}}(\cdot +L\mathbf{r})) + \grandO{L^{-1}} \pt
		\end{align*}
		Finally, from \eqref{eq: approximation_v} and the continuity of the functional $u \in H^1(\R^2) \mapsto \mathcal{E}(u)$, we have that $\mathcal{E}(v_{L,\mathbf{r}}(\cdot +L\mathbf{r})) = \mathcal{E}(v) + \grandO{L^{-\infty}}$ for all $\mathbf{r\in R}$. This concludes the proof of the upper bound \eqref{eq: upper_bound}.	
		\paragraph{Second step.} We claim the lower bound
		\begin{align}
		\label{eq: lower_bound}
		E_L  \geq \sum_{\mathbf{r\in R}}\mathcal{E}\left(\sqrt{\rho_{L,\mathbf{r}}(\cdot +L\mathbf{r})}\right) + \grandO{L^{-1}} \geq \sum_{\mathbf{r\in R}} I\left(\norme{\rho_{L,\mathbf{r}}}_{L^1(\R^2)}\right)  + \grandO{L^{-1}} \, ,
		\end{align}
		where $I(\lambda)$ is defined in \eqref{eq: min_problem_ow}. We consider the $\lscr_L$-periodic function $\chi_{L,\mathbf{0}}$ defined by
		\[
		\chi_{L,\mathbf{0}} = \sqrt{1 - \sum_{\mathbf{r\in}\lscr^\mathbf{R}} \abs{\chi_{L,\mathbf{r}}}^2} \pt
		\]
		We can always choose $\chi$ such that $\chi_{L,\mathbf{0}}$ is smooth. Notice that $\norme{\nabla \chi_{L,\mathbf{0}}}_{\lper^\infty(\Gamma_L)}$ and $\norme{\nabla \chi_{L,\mathbf{r}}}_{L^\infty(\R^2)}$ are $\grandO{L^{-1}}$. Finally, we set $\rho_{L,\mathbf{0}} = \abs{\chi_{L,\mathbf{0}}}^2 \rho_L$. We recall that we have also defined $\rho_{L,\mathbf{r}} = \abs{\chi_{L,\mathbf{r}}}^2 \rho_L$. This provides a periodic partition of unity
		\begin{align*}
		\abs{\chi_{L,\mathbf{0}}}^2 + \sum_{\mathbf{r\in R}} \sum_{\mathbf{u}\in\lscr} \abs{\chi_{L,\mathbf{u+r}}}^2 = 1 \, ,
		\end{align*}
		where the functions $\chi_{L,\mathbf{u+r}}$ for all $\mathbf{u}\in\lscr$ and $\mathbf{r\in R}$ have pairwise disjoint supports. We recall the periodic Hoffmann-Ostenhof inequality~\cite[Eq. (4.42)]{catto2001thermodynamic}
		\begin{align*}
		\int_{\Gamma_L} \absL{\nabla \sqrt{\rho_L}}^2 \leq \trm_{\mathscr{L}_L} (-\Delta \gamma_L) \pt
		\end{align*}
		By the periodic IMS formula, we can write
		\begin{multline*}
		\mathcal{E}_L(\gamma_L) 
		\geq \int_{\Gamma_L} \abs{\nabla \sqrt{\rho_L}}^2 + \int_{\Gamma_L} \left(-W_L^\mathbf{R}+V_L^\mathbf{R}\right)\rho_L  + \frac{1}{2} D_L(\rho_L,\rho_L)\\
		\geq \sum_{\mathbf{r\in R}} \sum_{\mathbf{u}\in\lscr} \left(\int_{\Gamma_L} \abs{\nabla \sqrt{\rho_{L,\mathbf{u+r}}}}^2  - \int_{\Gamma_L} \abs{\nabla \chi_{L,\mathbf{u+r}}}^2 \rho_L\right) + \int_{\Gamma_L} \abs{\nabla \sqrt{\rho_{L,\mathbf{0}}}}^2\\
		- \int_{\Gamma_L} \abs{\nabla \chi_{L,\mathbf{0}}}^2 \rho_L + \int_{\Gamma_L} \left(-W_L^\mathbf{R}+V_L^\mathbf{R}\right)\rho_L + \frac{1}{2} D_L(\rho_L,\rho_L) \pt
		\end{multline*}
		Using that $\supp V^\mathrm{pp} \subset B(0,R)$ for a fixed $R>0$ and the $\lscr_L$-periodicity of $W_L$ and $\rho_L$, we obtain
		\begin{align*}
		\mathcal{E}_L(\gamma_L) 
		\geq&~\sum_{\mathbf{r\in R}} \left(\int_{\R^2} \abs{\nabla \sqrt{\rho_{L,\mathbf{r}}}}^2 - \int_{\R^2} \abs{\nabla \chi_{L,\mathbf{r}}}^2 \rho_L + \int_{\R^2} \left(- W_L^\mathbf{R} + V^\mathrm{pp}(\cdot - L\mathbf{r})\right) \rho_{L,\mathbf{r}}\right) \\
		&+ \frac{1}{2} \sum_{\mathbf{r,r'\in R}} \iint_{\R^2 \times \R^2} \rho_{L,\mathbf{r}}(\mathbf{x}) \rho_{L,\mathbf{r'}}(\mathbf{y}) W_L(\mathbf{x-y}) \diff \mathbf{x} \diff \mathbf{y} \\
		&+\int_{\Gamma_L} \abs{\nabla \sqrt{\rho_{L,\mathbf{0}}}}^2 + D_L(\rho_L,\rho_{L,\mathbf{0}}) + \frac{1}{2}D_L(\rho_{L,\mathbf{0}},\rho_{L,\mathbf{0}}) - \int_{\Gamma_L} \abs{\nabla \chi_{L,\mathbf{0}}}^2 \rho_L - \int_{\Gamma_L} W_L^\mathbf{R} \rho_{L,\mathbf{0}}  \pt
		\end{align*}
		The terms where $\nabla \chi_{L,\mathbf{0}}$ or $\nabla \chi_{L,\mathbf{r}}$ appears are $\grandO{L^{-2}}$. Because of the non-negativity of $W_L$, we can bound from below by zero both the terms of the second line corresponding to $\mathbf{r\neq r'}$ and the three first terms of the third line. Also, using that $\dist(\supp \rho_{L,\mathbf{0}},\lscr_L^\mathbf{R}) \geq L\delta d_0/2$ and the right side of~\eqref{eq: G_properties}, we have
		\begin{align*}
		\int_{\Gamma_L} W_L^\mathbf{R} \rho_{L,\mathbf{0}}  = \sum_{\mathbf{r\in R}}\int_{\Gamma_L}W_L(\cdot - L\mathbf{r})\rho_{L,\mathbf{0}}  = \sum_{\mathbf{r\in R}} \int_{\Gamma_L} W_L  \rho_{L,\mathbf{0}}(\cdot +L\mathbf{r}) = \grandO{L^{-1}} \pt
		\end{align*}
		Reproducing the same arguments of the first step, we have, for all $\mathbf{r\in R}$,
		\begin{gather*}
		\int_{\R^2} W_L^\mathbf{R} \rho_{L,\mathbf{r}} = \int_{\R^2} \frac{\rho_{L,\mathbf{r}}(\mathbf{x})}{\abs{\mathbf{x-}L\mathbf{r}}} \diff \mathbf{x} + \grandO{L^{-1}} \et \\
		\iint_{\R^2 \times \R^2} \rho_{L,\mathbf{r}}(\mathbf{x}) \rho_{L,\mathbf{r}}(\mathbf{y}) W_L(\mathbf{x-y}) \diff \mathbf{x} \diff \mathbf{y} = \iint_{\R^2 \times \R^2} \frac{\rho_{L,\mathbf{r}}(\mathbf{x}) \rho_{L,\mathbf{r}}(\mathbf{y})}{\abs{\mathbf{x-y}}} \diff \mathbf{x} \diff \mathbf{y} + \grandO{L^{-1}} \pt
		\end{gather*}
		From all the previous estimates, we deduce that
		\begin{align*}
		\mathcal{E}_L(\gamma_L) \geq \sum_{\mathbf{r\in R}} \mathcal{E}\left(\sqrt{\rho_{L,\mathbf{r}}(\cdot +L\mathbf{r})}\right) + \grandO{L^{-1}} \geq \sum_{\mathbf{r\in R}} I\left(\norme{\rho_{L,\mathbf{r}}}_{L^1(\R^2)}\right)  + \grandO{L^{-1}} \pt
		\end{align*}	
		\paragraph{Third step.} We claim that
		\begin{align}
		\label{eq: convergence_weight}
		\forall \mathbf{r\in R},\quad \lim\limits_{L\to\infty} \norme{\rho_{L,\mathbf{r}}}_{L^1(\R^2)} = 1 \pt
		\end{align}	
		Since $\rho_L$ is $G_L$-invariant (see the proof of Lemma~\ref{lemma: hypo2}), $\chi$ is radial and  $G_L$ acts transitively on $\lscr^\mathbf{R}_L$, we obtain that $\norme{\rho_{L,\mathbf{r}}}_{L^1(\R^2)}$ is independent of $\mathbf{r\in R}$. Consequently, since $\norme{\rho_L}_{\lper^1(\Gamma_L)}=N$, we must have $\norme{\rho_{L,\mathbf{r}}}_{L^1(\R^2)} \leq 1$ for all $\mathbf{r\in R}$. 	
		We combine the upper bound \eqref{eq: upper_bound} and the lower bound \eqref{eq: lower_bound} to get
		\begin{align*}
		N I(1) \geq N I\left(\norme{\rho_{L,\mathbf{r}}}_{L^1(\R^2)}\right) + \grandO{L^{-1}} \, ,
		\end{align*}
		for any $\mathbf{r\in R}$. Assume there exists a subsequence $L_n\to\infty$ such that $\norme{\rho_{L_n,\mathbf{r}}}_{L^1(\R^2)} \to \lambda <1$ as $n\to\infty$. Then, since $\lambda \mapsto I(\lambda)$ is continuous on $\intff{0}{1}$, we obtain $I(1)\geq I(\lambda)$. This contradicts the binding inequality \eqref{eq: binding_inequality} and we must have \eqref{eq: convergence_weight}. In particular, $\sqrt{\rho_{L,\mathbf{r}}}$ is a minimizing sequence for $I(1)$.
		\paragraph{Fourth step.} We claim the uniform bound
		\begin{align}
		\label{eq: finite_kinetic_energy}
		\sup_{L\geq 1} \trm_{\mathscr{L}_L} (-\Delta \gamma_L) < \infty \pt
		\end{align}
		By Proposition~\ref{prop: madelung}, there exists $C>0$ such that for all $L\geq 1$ we have $\abs{W_L^\mathbf{R}(\mathbf{x})} \leq C\abs{\mathbf{x-r}}^{-1}$ in a neighborhood of any $\mathbf{r\in R}$. Then, by Proposition~\ref{lemma: kato_type_inequality}, Remark~\ref{rem: kato_type_inequality_consequence} and the periodic Hoffmann-Ostenhof inequality, we have
		\begin{align*}
		\absL{\int_{\Gamma_L}  W_L^\mathbf{R} \rho_L} \leq \frac{1}{4} \norme{\nabla \sqrt{\rho_L}}^2_{\lper^2(\Gamma_L)}+ C'\norme{\rho_L}_{\lper^1(\Gamma_L)} \leq \frac{1}{4} \trm_{\mathscr{L}_L} (-\Delta \gamma_L) + C'N\, ,
		\end{align*}
		for some constant $C'>0$ independent of $L$. We have a similar inequality when $W_L^\mathbf{R}$ is replaced by $V_L^\mathbf{R}$ since $\norme{V_L^\mathbf{R}}_{\lper^p(\Gamma_L)}$ does not depend on $L$. In addition, because $W_L\geq 0$, the direct energy term $D_L(\rho_L,\rho_L)$ is non-negative. Thus, we have 
		\[
		E_L = \mathcal{E}_L(\gamma_L) 
		= \trm_{\mathscr{L}_L}(-\Delta\gamma_L) + \int_{\Gamma_L} \left(- W_L^\mathbf{R}+V_L^\mathbf{R}\right)\rho_L + \frac{1}{2} D_L(\rho_L,\rho_L)\geq \frac{1}{2} \trm_{\mathscr{L}_L} (-\Delta\gamma_L) - C'N \pt
		\]
		Moreover, by the first step, we also have $\sup_{L\geq 1} E_L < \infty$. The uniform bound \eqref{eq: finite_kinetic_energy} follows.
		\paragraph{Fifth step.}
		We prove the convergences \eqref{eq: variational1}. For $\mathbf{r\in R}$, we can write 
		\begin{multline*}
		\mathcal{E}\left(\sqrt{\rho_{L,\mathbf{r}}(\cdot+L\mathbf{r})}\right) = \norme{\rho_{L,\mathbf{r}}}_{L^1(\R^2)}\left[ \mathcal{E}\left(\sqrt{\frac{\rho_{L,\mathbf{r}}(\cdot+L\mathbf{r})}{\norme{\rho_{L,\mathbf{r}}}_{L^1(\R^2)}}}\right)  \right. \\
		\left. +  \frac{\norme{\rho_{L,\mathbf{r}}}_{L^1(\R^2)} - 1}{2} \iint_{\R^2 \times \R^2} \frac{\rho_{L,\mathbf{r}}(\mathbf{x}) \rho_{L,\mathbf{r}}(\mathbf{y})}{\abs{\mathbf{x-y}}} \diff \mathbf{x} \diff \mathbf{y}\right] \pt
		\end{multline*}
		By the Hardy-Littlewood-Sobolev inequality and the Sobolev embedding $L^{8/3}(\R^2) \subset H^1(\R^2)$, we have
		\begin{align*}
		\iint_{\R^2 \times \R^2} \frac{\rho_{L,\mathbf{r}}(\mathbf{x}) \rho_{L,\mathbf{r}}(\mathbf{y})}{\abs{\mathbf{x-y}}} \diff \mathbf{x} \diff \mathbf{y}
		\lesssim  \norme{\sqrt{\rho_{L,\mathbf{r}}}}^4_{L^{8/3}(\R^2)} \leq  \norme{\sqrt{\rho_{L,\mathbf{r}}}}^4_{H^1(\R^2)} \pt
		\end{align*}
		Using that $\supp \rho_{L,\mathbf{r}} \subset \frac{1}{2}\Gamma_L + L\mathbf{r}$ and periodic Hoffmann-Ostenhof's inequality, we deduce
		\begin{align*}
		\iint_{\R^2 \times \R^2} \frac{\rho_{L,\mathbf{r}}(\mathbf{x}) \rho_{L,\mathbf{r}}(\mathbf{y})}{\abs{\mathbf{x-y}}} \diff \mathbf{x} \diff \mathbf{y}
		\lesssim \norme{\sqrt{\rho_L}}_{H^1_\mathrm{per}(\Gamma_L)}^4 \lesssim \absL{\trm_{\mathscr{L}_L} (-\Delta\gamma_L)}^2 \pt
		\end{align*}
		Recall that $\norme{\rho_{L,\mathbf{r}}}_{L^1(\R^2)} = 1 + \petito{1}$ from the third step. Using in addition the uniform estimate~\eqref{eq: finite_kinetic_energy} from the fourth step, we obtain
		\begin{align*}
		\mathcal{E}\left(\sqrt{\rho_{L,\mathbf{r}}(\cdot+L\mathbf{r})}\right) =  \norme{\rho_{L,\mathbf{r}}}_{L^1(\R^2)} \mathcal{E}\left(\sqrt{\frac{\rho_{L,\mathbf{r}}(\cdot+L\mathbf{r})}{\norme{\rho_{L,\mathbf{r}}}_{L^1(\R^2)}}}\right) + \petito{1} \, ,
		\end{align*}
		which, together with the stability inequality \eqref{eq: stability}, leads to
		\begin{multline*}
		E_L \geq \sum_{\mathbf{r\in R}} \norme{\rho_{L,\mathbf{r}}}_{L^1(\R^2)}  \mathcal{E}\left(\sqrt{\frac{\rho_{L,\mathbf{r}}(\cdot+L\mathbf{r})}{\norme{\rho_{L,\mathbf{r}}}_{L^1(\R^2)}}}\right) + \petito{1} \\
		\geq NI(1) + C \sum_{\mathbf{r\in R}} \min_{\theta \in \intfo{0}{2\pi}} \normeL{e^{i\theta}\sqrt{\rho_{L,\mathbf{r}}(\cdot+L\mathbf{r})} - \sqrt{\norme{\rho_{L,\mathbf{r}}}_{L^1(\R^2)}}v}^2_{H^1(\R^2)} + \petito{1} \\
		\geq NI(1) + C \sum_{\mathbf{r\in R}} \min_{\theta \in \intfo{0}{2\pi}} \normeL{e^{i\theta}\sqrt{\rho_{L,\mathbf{r}}(\cdot+L\mathbf{r})} - v}^2_{H^1(\R^2)} + \petito{1} \pt
		\end{multline*}
		In the last inequality, to extract the $\petito{1}$ from the norm, we have used, as above, the inequality $\norme{\sqrt{\rho_{L,\mathbf{r}}}}_{H^1(\R^2)} \leq \trm_{\mathscr{L}_L}(-\Delta\gamma_L)$ and the estimate \eqref{eq: finite_kinetic_energy} from the fourth step. Finally, since $\rho_{L,\mathbf{r}}$ and $v$ are positive functions, the functional
		\[
		\theta\in\intfo{0}{2\pi}\mapsto\normeL{e^{i\theta}\sqrt{\rho_{L,\mathbf{r}}(\cdot+L\mathbf{r})} - v}^2_{H^1(\R^2)} \, ,
		\]
		is minimal for $\theta = 0$. Now, recalling the upper bound \eqref{eq: upper_bound}, we write
		\begin{align*}
		NI(1) + \grandO{L^{-1}} \geq E_L \geq NI(1) + \sum_{\mathbf{r\in R}}\normeL{\sqrt{\rho_{L,\mathbf{r}}(\cdot+L\mathbf{r})} - v}^2_{H^1(\R^2)} + \petito{1} \, ,
		\end{align*}
		which concludes the proof of \eqref{eq: variational1}.	
		\paragraph{Sixth step.} We show estimate \eqref{eq: variational2}. Using the periodicity of $W_L$, we can write
		\begin{align*}
		\rho_L \ast_L W_L 
		= \rho_{L,\mathbf{0}} \ast_L W_L + \sum_{\mathbf{r\in R}} \left(\sum_{\mathbf{u}\in\lscr} \rho_{L,\mathbf{u+r}}\right) \ast_L W_L 
		= \rho_{L,\mathbf{0}} \ast_L W_L + \sum_{\mathbf{r\in R}} \rho_{L,\mathbf{r}} \ast W_L \pt
		\end{align*}
		Let $\mathbf{r\in R}$ and $\mathbf{x}\in\Gamma_L$. Using the left side of \eqref{eq: G_properties} and the fact that $\supp \rho_{L,\mathbf{r}} \subset \frac{1}{2}\Gamma_L+L\mathbf{r}$, we have
		\begin{align*}
		(\rho_{L,\mathbf{r}} \ast (W_L- |\cdot|^{-1}))(\mathbf{x}) 
		= \sum_{\mathbf{u}\in\lscr}\int_{\R^2} \frac{\rho_{L,\mathbf{r}}(\mathbf{y})}{\abs{\mathbf{x-y}-L\mathbf{u}}} \mathds{1}_{\mathbf{x-y}-L\mathbf{u}\in\Gamma_L}\diff \mathbf{y} + \grandO{L^{-1}} = \grandO{L^{-1}} \, ,
		\end{align*}
		where the $O$ is uniform in $\mathbf{x}\in\Gamma_L$. We recall~\cite[Lemma 3]{cazalis} which is the non-periodic version of Corollary~\ref{cor: inequality_potential}: for all $r\in\intof{2}{\infty}$ and for all $(u,w) \in H^1(\R^2)$, we have
		\begin{align*}
		\normeL{(uw)\ast |\cdot|^{-1}}_{L^r(\R^2)} \lesssim \norme{u}_{H^1(\R^2)} \norme{w}_{H^1(\R^2)} \pt
		\end{align*}
		Using this inequality, we have
		\begin{align*}
		\normeL{\left(\rho_{L,\mathbf{r}} - \abs{v(\cdot - L\mathbf{r})}^2\right) \ast |\cdot|^{-1}}_{L^\infty(\R^2)}
		\lesssim\left( \sup_{L\geq 1} \normeL{\sqrt{\rho_{L,\mathbf{r}}}}_{H^1(\R^2)} + \normeL{v}_{H^1(\R^2)}\right) \normeL{\sqrt{\rho_{L,\mathbf{r}}} - v(\cdot-L\mathbf{r})}_{H^1(\R^2)} \pt
		\end{align*}
		In the fifth step, we have shown that $\sup_{L\geq 1} \normeL{\sqrt{\rho_{L,\mathbf{r}}}}_{H^1(\R^2)} < \infty$. Thus, by \eqref{eq: variational1}, we obtain
		\begin{align*}
		\rho_L \ast_L W_L  = \rho_{L,\mathbf{0}} \ast_L W_L  + \sum_{\mathbf{r\in R}} \absL{v(\cdot - L\mathbf{r})}^2\ast|\cdot|^{-1} + \petito{1} \, ,
		\end{align*}
		where the $o$ makes sense in $L^\infty(\Gamma_L)$. It remains to show that $\norme{\rho_{L,\mathbf{0}} \ast_L W_L}_{\lper^\infty(\Gamma_L)} = \petito{1}$ to conclude the proof of Proposition~\ref{prop: variational}. By the third step and the normalization $\norme{\rho_L}_{\lper^1(\Gamma_L)}=N$, we see that
		\begin{align}
		\label{eq: blop}
		\norme{\rho_{L,\mathbf{0}}}_{\lper^1(\Gamma_L)} = \petito{1} \pt
		\end{align}
		In addition, from the identity $\nabla \sqrt{\rho_{L,\mathbf{0}}} = \chi_{L,\mathbf{0}} \nabla \sqrt{\rho_L} + \sqrt{\rho_L} \nabla \chi_{L,\mathbf{0}}$, the fact that $\normeL{\nabla \chi_{L,\mathbf{0}}}_{\lper^\infty(\Gamma_L)} = \grandO{L^{-1}}$, the periodic Hoffmann-Ostenhof inequality and the uniform bound \eqref{eq: finite_kinetic_energy}, we have 
		\begin{align}
		\label{eq: blap}
		\sup_{L\geq 1}\normeL{\sqrt{\rho_{L,\mathbf{0}}}}_{H^1_\mathrm{per}(\Gamma_L)} <\infty \pt
		\end{align}
		From estimates \eqref{eq: blop} and \eqref{eq: blap}, the Sobolev embedding $\lper^{2q}(\Gamma_L) \subset H^1_\mathrm{per}(\Gamma_L)$ (where the continuity constant does not depend on $L$, see~\cite[Theorem 2.28]{aubin1998somenonlinear}) and an interpolation argument, we obtain
		\begin{align}
		\label{eq: blom}
		\forall q \in \intfo{1}{\infty},\quad \normeL{\rho_{L,\mathbf{0}}}_{\lper^q(\Gamma_L)} = \petito{1} \pt
		\end{align}
		From the proof of Corollary~\ref{cor: inequality_potential}, we see that for any $q\in\intoo{2}{\infty}$ we have
		\begin{align*}
		\normeL{\rho_{L,\mathbf{0}}\ast_LW_L}_{\lper^\infty(\Gamma_L)} \lesssim \normeL{\rho_{L,\mathbf{0}}}_{\lper^1(\Gamma_L)}+\normeL{\rho_{L,\mathbf{0}}}_{\lper^q(\Gamma_L)} \, ,
		\end{align*}	
		which is a $\petito{1}$ by \eqref{eq: blom}. This concludes the proof of Proposition~\ref{prop: variational}.	
	\end{proof}
	
	\appendix
	
	\section{The weak constrast regime}\label{sec:the-weak-constrast-regime}
	
	In this appendix, we consider the honeycomb lattice $\lscr^H_L = L\lscr^H$ (introduced in Section~\ref{sec:the-triangular-and-the-hexagonal-lattices}) and the \emph{weak contrast regime} which corresponds to the limit $L\to 0$. We do not provide any proof and instead give appropriate references.
	
	\paragraph{Dirac points in the weak contrast regime}
	
	Let $p\in\intof{1}{\infty}$. For all $L\in\intoo{-1}{1}$, we consider a real-valued $\lscr_L$-invariant potential $V_L\in\lper^p$ and we denote by $H_L = -\Delta + LV_L$ the Schrödinger operator associated with the potential $L V_L$. We assume that there exists $\textbf{x}_0 \in \R^2$ such that $\tilde{V}_L = V_L(\cdot - \textbf{x}_0)$ satisfies the following symmetry conditions:
	\begin{enumerate}[noitemsep, label=(\roman*)]
		\item $\tilde{V}_L$ is even, i.e. $\tilde{V}_L(-\mathbf{x}) = \tilde{V}_L(\mathbf{x})$;
		\item $\tilde{V}_L$ is invariant under rotation by $2\pi/3$, i.e. $\tilde{V}_L(M_\mathcal{R}^*\mathbf{x}) = \tilde{V}_L(\mathbf{x})$ where $M_\mathcal{R}$ is the rotation by $2\pi/3$ matrix.
	\end{enumerate}
	For fixed $L$, the potential $V_L$ shares the symmetries of honeycomb lattice potentials (see~\cite[Definition 2.1]{fefferman2012honeycomb}) except that we allow for more singular potentials. The following theorem, proved in~\cite[Chapter 4]{cazalisthesis}, extends some conclusions of~\cite[Theorem 5.1]{fefferman2012honeycomb}.	
	\begin{theo}
		\label{rfy_th_non_isolated}
		Assume that the map $L\mapsto V_L \in \lper^p$ is continuous at 0 and that $c_{1,1}\coloneqq\widehat{V}_0(\mathbf{v}_1+\mathbf{v}_2) \neq 0$. Then there exists $L_0 >0$ such that for all $L\in\intoo{-L_0}{L_0}\setminus\{0\}$ the periodic Schrödinger operator $H_L$ admits Dirac points between the first and second bands if $Lc_{1,1}>0$ and between the second and third bands if $Lc_{1,1}<0$.
	\end{theo}
	The proof of Theorem~\ref{rfy_th_non_isolated} follows the arguments of~\cite{fefferman2012honeycomb} (see also~\cite{grushin2009multiparameter,berkolaiko2018symmetry}). There, the authors consider the operator $H_\epsilon=-\Delta +\epsilon V$ where $V$ is a honeycomb lattice potential and show that $H_\epsilon$ admits Dirac point for all $\epsilon\in\R$ except in a countable and closed set. In the weak contrast regime, that is for $\abs{\epsilon} \ll 1$, their argument uses the Lyapunov-Schmidt reduction. In Theorem~\ref{rfy_th_non_isolated}, there are two main differences in the assumptions compared to~\cite[Theorem 5.1]{fefferman2012honeycomb}: first the potential may be more singular and second it may depend on the amplitude $L$. 
	
	\paragraph{The rHF model in the weak contrast regime}
	
	We consider the periodic rHF model for graphene, introduced in Section~\ref{sec:the-rhf-model-for-periodic-systems-at-dissociation}. For simplicity, we assume there is no pseudo-potential, that is $V^\mathrm{pp}=0$. Therefore, the periodic potential generated by the lattice $\lscr_L^H$ reads
	\begin{align*}
	-W_L^H =  -W_L(\cdot - L\mathbf{a}) - W_L(\cdot - L\mathbf{b})\pt
	\end{align*}
	The periodic rHF model on $\lscr^H_L$ consists in solving the minimization problem
	\begin{align}
	\label{rfy_rHF}
	E_L = \inf \enstq{\mathcal{E}_L(\gamma)}{\gamma \in \mathcal{S}_{\mathrm{per},L} \et \trm_{\mathscr{L}_L} (\gamma) = 2/q} \, ,
	\end{align}
	where $q\in\N^*$ denotes the number of spin states. Using the dilatation by $L$ isometry, defined by $d_L\varphi(\mathbf{x}) = L\varphi(L\mathbf{x})$, one can reformulate the minimization problem \eqref{rfy_rHF} on the fixed lattice $\lscr^H$ (we drop the subscript $L$ from the notations when $L=1$):
	\begin{align}
	\label{rfy_new_rHF}
	E_L = L^{-2} F_L \ou F_L \coloneqq \inf\enstq{\mathcal{F}_L(\gamma)}{\gamma \in \mathcal{S}_\mathrm{per} \et \trm_{\mathscr{L}}(\gamma)=2/q}  \, ,
	\end{align}
	where the energy functional $\mathcal{F}_L$ is defined for all $\gamma\in\mathcal{S}_\mathrm{per,1}$ by
	\begin{align*}
	\mathcal{F}_L(\gamma) \coloneqq \trm_{\mathscr{L}} (-\Delta \gamma) + L \left(- \int_{\Gamma} W^H \rho_\gamma + \frac{q}{2}D(\rho_\gamma,\rho_\gamma)\right) \pt
	\end{align*}
	Following the proofs of~\cite[Theorem 1]{cances2008newapproach} and~\cite[Theorem 2.1]{catto2001thermodynamic}, one can show that \eqref{rfy_new_rHF} is well-posed and admits a unique minimizer $\gamma_L\in\mathcal{S}_\mathrm{per}$. If we denote by $\rho_L(\mathbf{x}) = \gamma_L(\mathbf{x,x})$ its one-body density then $\gamma_L$ is the unique solution of the mean-field equation
	\begin{align*}
	\gamma_L = \mathds{1}_{\intof{-\infty}{\epsilon_L}} \left(H_L^\mathrm{MF}\right) \, ,
	\end{align*}
	where $\epsilon_L\in\R$ is the Fermi level and where the mean-field hamiltonian is given by
	\begin{align*}
	\boxed{H_L^\mathrm{MF} = -\Delta + L V_L^\mathrm{MF} \quad \text{with}\quad V_L^\mathrm{MF} \coloneqq -W^H + q \rho_L *_{\Gamma} W  \pt}
	\end{align*}
	In the next theorem, proved in~\cite[Chapter 4]{cazalisthesis}, we state that $H_L^\mathrm{MF}$ satisfies the assumptions of Theorem~\ref{rfy_th_non_isolated} and that, in addition, the Fermi level \emph{does not} coincide with the cones energy.
	\begin{theo}
		\label{rfy_prop_continuity}
		The map $L \mapsto \rho_L*_\Gamma W$ is continuous from $\R_+$ to $L^\infty_\mathrm{per}$. In addition, $V_0^\mathrm{MF}$ satisfies $\widehat{V}^\mathrm{MF}_0(\mathbf{v}_1+\mathbf{v}_2) > 0$. Consequently, for all $L>0$ small enough, $H_L^\mathrm{MF}$ admits Dirac points at the vertices of $\Gamma^*$ between the first and second bands.
		
		If we assume that $q=2$ then, for all $L$ small enough, we have $\epsilon_L < \lambda_L$ where $\lambda_L$ denotes the energy level of these Dirac points.
	\end{theo}
	This result is illustrated by the numerical simulations by DFTK software, presented in Figure~\ref{fig: phase_transition}. For $L=0$, the first two Bloch bands of the free Laplacian $-\Delta$ coincide along the quasi-momenta section $(\mathrm{K,M})$, represented in Figure~\ref{fig:triangular_lattice_BZ}. Then, when $L$ is non zero but small enough (for instance, in Figure~\ref{fig:L=0.5} where $L=1/2$), the bands still overlap which causes the Fermi level to be smaller than the cones energy.

	\section{Perturbation theory for singular potentials}\label{sec:perturbation-theory-for-singular-potentials}
	
	Let $p\in\intof{1}{\infty}$. For all $L\geq 1$, we consider potentials $V$ and $V_L\in L^p(\R^2)+L^\infty(\R^2)$. We denote by $H=-\Delta+V$ and $H_L = -\Delta + V_L$ the associated self-adjoint Schrödinger operators, given by the Friedrichs extension if $p\in\intoo{1}{2}$.
	\begin{prop}[Singular perturbation theory]
		\label{prop: perturbation_theory}
		Assume that
		\begin{align}
		\label{eq:perturbation_theory_assumption}
		\lim\limits_{L\to\infty} \norme{V-V_L}_{L^p(\R^2)+L^\infty(\R^2)} = 0 \, ,
		\end{align}
		and that the discrete spectrum of $H$ is non empty. Let $\lambda \in \sigma_\mathrm{d}(H)$. Then, for all $\epsilon>0$ small enough there exists $\mathscr{C} \subset \C$ a contour enclosing $\lambda$ such that $\dist\left(\mathscr{C}, \sigma_\mathrm{d}(H_L)\right) \geq \epsilon$ for all $L$ large enough. In this case, we denote by
		\begin{align*}
		P \coloneqq \frac{-1}{2\pi i} \oint_\mathscr{C} \frac{\diff z}{H - z} \et P_{L} \coloneqq \frac{-1}{2\pi i} \oint_\mathscr{C} \frac{\diff z}{H_L - z} \, ,
		\end{align*}
		the spectral projections of $H$ and $H_L$ associated with the interval of real numbers enclosed by $\mathscr{C}$. Then, for all $L$ large enough, the ranks of $P_L$ and $P$ are equal and there exists $C>0$ such that	
		\begin{align*}
		\normeL{(-\Delta+1)^{\min (\frac{p}{2}, 1)} (P - P_L)} \leq C\norme{V - V_L}_{L^p(\R^2)+L^\infty(\R^2)} \pt
		\end{align*}	
	\end{prop}
	In the case where $\lambda$ is non degenerate, we can be more precise.
	\begin{cor}[Non degenerate case]
		\label{cor:perturbation_theory}
		Assume \eqref{eq:perturbation_theory_assumption} and that the discrete spectrum of $H$ is non empty. Let $\lambda\in\sigma_\mathrm{d}(H)$ be non degenerate. Then, for $L$ large enough, the contour $\mathscr{C}$, given by Proposition~\ref{prop: perturbation_theory}, encloses only one discrete eigenvalue $\lambda_L$ of $H_L$ and we have
		\begin{align*}
		\lim\limits_{L\to\infty} \lambda_L = \lambda \pt
		\end{align*}
		In addition, if we denote by $v$ (resp. $v_L$) a normalized eigenfunction associated with $\lambda$ (resp. $\lambda_L$) then there exists a constant $C>0$ such that 
		\begin{align*}
		\min_{\theta\in\intff{0}{2\pi}}\norme{e^{i\theta}v - v_L}_{H^{\min(p,2)}(\R^2)} \leq C \norme{V-V_L}_{L^p(\R^2)+L^\infty(\R^2)} \pt
		\end{align*}
	\end{cor}
	We do not write the proofs of Proposition~\ref{prop: perturbation_theory} and Corollary~\ref{cor:perturbation_theory} since they follow standard perturbation theory arguments~\cite{kato1995perturbation, reed1978methodsIV} and we refer the reader to~\cite[Chapter 3 -- Appendix B]{cazalisthesis}.
	

	\small{}

\end{document}